\documentclass[11pt, letterpaper]{article}

\usepackage[margin=1in]{geometry}
\usepackage{lmodern}
\usepackage{graphicx}
\usepackage{multirow}
\usepackage{amsmath,amssymb,amsfonts,amsthm}
\usepackage{bm}
\usepackage{mathrsfs}
\usepackage{physics}
\usepackage{xcolor}
\usepackage{booktabs}
\usepackage{algorithm}
\usepackage{algorithmicx}
\usepackage[noend]{algpseudocode}
\usepackage{mathtools}
\usepackage{nicematrix}
\usepackage[sort&compress,numbers]{natbib}
\usepackage[inline]{enumitem}
\usepackage[normalem]{ulem}
\usepackage[hypertexnames=false,colorlinks=true,citecolor=black,linkcolor=black]{hyperref}
\usepackage[capitalize]{cleveref}
\usepackage[labelformat=simple]{subcaption}

\mathtoolsset{centercolon}

\usepackage{thmtools,thm-restate}

\declaretheoremstyle[
  shaded={bgcolor=gray!15},
]{thmsty}

\declaretheorem[
  name=Theorem,
  refname={Theorem,Theorems},
  style=thmsty,
  numberwithin=section,
]{theorem}
\declaretheorem[
  name=Corollary,
  refname={Corollary,Corollaries},
  style=thmsty,
  numberlike=theorem,
]{corollary}
\declaretheorem[
  name=Proposition,
  refname={Proposition,Propositions},
  style=thmsty,
  numberlike=theorem,
]{proposition}
\declaretheorem[
  name=Lemma,
  refname={Lemma,Lemmas},
  style=thmsty,
  numberlike=theorem,
]{lemma}
\declaretheorem[
  name=Definition,
  refname={Definition,Definitions},
  style=thmsty,
  numberlike=theorem,
]{definition}
\declaretheorem[
  name=Claim,
  refname={Claim,Claims},
  style=thmsty,
  numberlike=theorem,
]{claim}
\declaretheorem[
  name=Remark,
  refname={Remark,Remarks},
  style=thmsty,
  numberlike=theorem,
]{remark}

\crefname{algorithm}{Algorithm}{Algorithms}
\crefname{definition}{Definition}{Definitions}
\crefname{line}{Line}{Lines}
\crefname{lemma}{Lemma}{Lemmas}
\crefname{theorem}{Theorem}{Theorems}
\crefname{section}{Section}{Sections}
\crefname{appendix}{Appendix}{Appendices}
\crefname{table}{Table}{Tables}
\crefname{figure}{Figure}{Figures}
\crefname{equation}{}{}
\crefname{item}{}{}
\Crefname{equation}{Eq.}{Eqs.}
\crefname{assumption}{Assumption}{Assumptions}
\crefname{question}{Question}{Questions}

\allowdisplaybreaks[4]
\DeclareMathOperator{\Order}{O}

\DeclareMathOperator{\sgn}{sgn}
\DeclareMathOperator{\rmodd}{odd}
\DeclareMathOperator{\Pf}{Pf}
\DeclareMathOperator{\dist}{dist}
\DeclareMathOperator{\wt}{length}

\newcommand{\relmiddle}[1]{\mathrel{}\middle#1\mathrel{}}

\newcommand{\bfk}{\mathbf{K}}
\newcommand{\bfm}{\mathbf{M}}

\newcommand{\calk}{\mathcal{K}}
\newcommand{\call}{\mathcal{L}}
\newcommand{\calp}{\mathcal{P}}
\newcommand{\calq}{\mathcal{Q}}
\newcommand{\calr}{\mathcal{R}}
\newcommand{\calt}{\mathcal{T}}
\newcommand{\calw}{\mathcal{W}}

\newcommand{\symgrp}{\mathfrak{S}}

\newcommand{\symdiff}{\mathbin{\triangle}}

\newcommand{\skewsymmetric}{alternating}

\newcommand{\lv}{\ensuremath{\mathsf{level}}}
\newcommand{\setlower}{\ensuremath{\mathsf{AssignLowerLevel}}}
\newcommand{\setupper}{\ensuremath{\mathsf{AssignUpperLevel}}}
\newcommand{\grow}{\ensuremath{\mathsf{Grow}}}
\newcommand{\blossom}{\ensuremath{\mathsf{Blossom}}}
\newcommand{\union}{\ensuremath{\mathsf{Union}}}
\newcommand{\find}{\ensuremath{\mathsf{Find}}}
\newcommand{\vl}{\ensuremath{\mathsf{VL}}}

\newcommand{\kmin}{k_{\min}}
\newcommand{\kmax}{k_{\max}}

\newcommand{\rmb}{\mathrm{B}}
\newcommand{\rmn}{\mathrm{N}}

\newcommand{\st}{\calt}

\algrenewcommand\textproc{\textsf}
\algnewcommand{\Continue}{\textbf{continue}}

\NiceMatrixOptions{code-for-first-row = \scriptstyle, code-for-first-col = \scriptstyle, cell-space-limits = 1pt}

\title{A Shortest Augmenting Path Algorithm for Linear Matroid Parity}

\author{Kou Hamada\thanks{Department of Mathematical Informatics, The University of Tokyo (Research Fellow of Japan Society for the Promotion of Science), Tokyo 113-8656, Japan. kouhamada2002@gmail.com} \and Satoru Iwata\thanks{Department of Mathematical Informatics, The University of Tokyo, Tokyo 113-8656, Japan, and Institute for Chemical Reaction Design and Discovery, Hokkaido University, Sapporo, Hokkaido 001-0021, Japan. iwata@mist.i.u-tokyo.ac.jp}}
\date{October 15, 2026}

\begin{document}
\maketitle

\begin{abstract}
  The matroid parity problem serves as a fundamental framework that generalizes both graph matching and matroid intersection. Although the general version is intractable, Lov\'{a}sz (1981) developed a polynomial-time algorithm for the linear matroid parity problem, assuming the availability of matrix representations. Subsequently, Gabow and Stallmann (1986) presented an augmenting path algorithm, which has long been recognized as one of the fastest deterministic algorithms. 
  Since the use of shortest augmenting paths successfully improved algorithms for graph matching (Micali \& Vazirani, 1980) and linear matroid intersection (Cunningham, 1986), extending these techniques to linear matroid parity appears to be a natural progression. Such an approach is expected to reduce the cost of matrix updates and yield a faster algorithm, yet its development has remained elusive for four decades.

  In this paper, we present the first shortest augmenting path algorithm for linear matroid parity. Our approach synthesizes the augmenting path algorithm of Gabow and Stallmann with the synchronized blossom formation of the Micali--Vazirani framework, whose correctness was recently proven by Vazirani (2024). Our key technical contributions are threefold: (i) a linear-algebraic argument that bounds the lengths of shortest augmenting paths for linear matroid parity, which generalizes Cunningham's bound for linear matroid intersection; (ii) an a priori characterization of shortest search paths through lower bounds on their lengths; and (iii) an extension of the structural properties for graph matching established by Izumi, Kitamura, and Yamaguchi~(2025) to the linear matroid parity setting.
 
  Our algorithm deterministically solves the linear matroid parity problem in $\Order(n r^2 \log r)$ time, where $n$ is the size of the ground set and $r$ is the rank of the matroid. By incorporating fast matrix multiplication, this complexity can be further reduced to $\Order(n r^2)$. These results improve upon the long-standing deterministic bounds of $\Order(n r^3)$ and $\Order(n r^\omega)$.
\end{abstract}

\pagenumbering{gobble}
\clearpage

\tableofcontents

\clearpage
\pagenumbering{arabic}

\section{Introduction}

\subsection{Background}
\label{sec:background}

The \textit{matroid parity} problem, introduced by Lawler~\cite{Lawler1976CombinatorialOptimization}, serves as a unifying framework for graph matching and matroid intersection. Although this problem necessitates an exponential number of independence oracle calls in general~\cite{Lovasz1981MatroidMatching, JensenK1982ComplexityMatroid}, Lov\'{a}sz~\cite{Lovasz1981MatroidMatching} developed the first polynomial-time algorithm for \textit{linear matroid parity}, in which a matrix representation of the matroid is available.

The linear matroid parity problem is formulated as follows. Given an $r \times n$ matrix $A$ of full row rank, where $n$ is even and the columns are partitioned into pairs, find a maximum \textit{independent parity set}, which is defined as a set of linearly independent columns that is a union of such pairs.

As a broad generalization of fundamental combinatorial problems, this problem finds applications in diverse areas, including path packing problems~\cite{Lovasz1980MatroidMatching,Schrijver2003CombinatorialOptimization,TanigawaY2016PackingNonzero,Yamaguchi2016Packing$A$paths}, the rigidity of planar structures~\cite{Lovasz1980MatroidMatching}, maximum-genus cellular embedding~\cite{FurstGM1988FindingMaximumgenus}, and the unique solvability of electric networks~\cite{Milic1974GeneralPassive,Recski2019HybridDescription}.

Currently, the fastest known deterministic algorithm is the augmenting path algorithm of Gabow and Stallmann~\cite{GabowS1986AugmentingPath}, which runs in $\Order(n r^{\omega})$ time, where $\omega \in [2, 3)$ is the exponent of fast matrix multiplication originated by Strassen~\cite{Strassen1969GaussianElimination} and improved by subsequent studies. Analogous to Edmonds' blossom algorithm~\cite{Edmonds1965MaximumMatching} for graph matching, the Gabow--Stallmann algorithm utilizes blossoms along with a linear-algebraic technique to identify augmenting paths.

Meanwhile, in the 1980s, significant progress was made in graph matching and linear matroid intersection through the study of shortest augmenting paths. Micali and Vazirani~\cite{MicaliV1980$Osqrt|V||E|$Algorithm} leveraged this approach to develop a faster graph matching algorithm. Similarly, Cunningham~\cite{Cunningham1986ImprovedBounds} refined the time complexity for linear matroid intersection by bounding the lengths of shortest augmenting paths. Although Cunningham's bound applies to general matroids, its algorithmic benefits are primarily realized in the linear setting.

Despite these advancements, the linear matroid parity problem has resisted attempts to leverage shortest augmenting paths for four decades. This absence can be attributed to several fundamental bottlenecks.
\begin{enumerate}
  \item While extending Cunningham's bound is essential for guaranteeing the efficiency of a shortest augmenting path algorithm for linear matroid parity, its original graph-theoretic proof appears inapplicable to this setting; such an approach fails to exploit the matrix representation, which is integral to the tractability of linear matroid parity.
  \item Augmenting paths in the Gabow--Stallmann algorithm are sensitive to the search order. Unlike graph matching, skipping a valid concatenation of the search paths can trigger numerical cancellations that invalidate subsequent paths, which precludes a straightforward a priori definition of ``shortest search paths.''
  \item The correctness proof of the Micali--Vazirani algorithm remained open for decades. Their original proposal~\cite{MicaliV1980$Osqrt|V||E|$Algorithm} provided an algorithm description without a proof of correctness. A definitive proof has emerged only recently through successive attempts~\cite{Vazirani1994TheoryAlternating,Vazirani2013SimplificationMV,Vazirani2024TheoryAlternating}. Furthermore, the resulting proof requires a substantial inductive argument, mainly due to a ``chicken-and-egg'' problem: the key to defining blossoms is the identification of their base vertices, yet these vertices must be identified through the structural properties of the blossoms themselves.\footnote{Even if one focuses on finding a single shortest augmenting path in graph matching, which suffices for linear matroid parity, this inherent challenge remains. See \cite{IzumiKY2025ForgettingAlternation,Vazirani2024TheoryAlternating} for details.}
\end{enumerate}

However, recent breakthroughs have laid a foundation for addressing these issues. Iwata and Kobayashi~\cite{IwataK2022WeightedLinear} provided a combinatorial deterministic algorithm for the weighted linear matroid parity problem, reformulating the linear-algebraic technique of Gabow and Stallmann. More recently, Vazirani~\cite{Vazirani2024TheoryAlternating} provided the correctness proof of the Micali--Vazirani algorithm through a substantial inductive argument. Building on this work, Izumi, Kitamura, and Yamaguchi~\cite{IzumiKY2025ForgettingAlternation} proposed a new framework for graph matching that yields a shortest augmenting path algorithm with a drastically simplified correctness proof, at the cost of a logarithmic factor in time complexity. The framework bypasses the complexity of blossoms via \textit{alternating base trees} and employs \textit{incoming edges} to exactly characterize the edges necessary for shortest search paths. These advancements enable us to bridge the gap between the Gabow--Stallmann and Micali--Vazirani paradigms.

\subsection{Techniques and Contributions}

In this paper, we provide the first shortest augmenting path algorithm for the linear matroid parity problem. Our approach synthesizes the Gabow--Stallmann algorithm with the synchronized blossom formation of the Micali--Vazirani framework, while incorporating the recent linear-algebraic reformulation by \cite{IwataK2022WeightedLinear}. To establish its efficiency and correctness, we address the aforementioned challenges as follows.
\begin{enumerate}
  \item Building on the algebraic formulation of linear matroid parity~\cite{GeelenI2005MatroidMatching}, we provide a novel linear-algebraic approach to bounding the shortest augmenting path length. Our bound generalizes Cunningham's bound in the linear matroid intersection case.
  \item Leveraging the linear-algebraic technique of Iwata and Kobayashi~\cite{IwataK2022WeightedLinear}, we identify an a priori lower bound on search path lengths. We demonstrate that this bound is tight for the proposed algorithm, thereby characterizing ``shortest search paths'' without explicitly enumerating all candidate search paths.
  \item We extend the structural properties of incoming edges introduced by Izumi, Kitamura, and Yamaguchi~\cite{IzumiKY2025ForgettingAlternation} to linear matroid parity, proving that the incoming edges examined by the algorithm yield the shortest search paths and augmenting paths.
\end{enumerate}

Consequently, we improve the best known deterministic complexity for solving the unweighted linear matroid parity problem from $\Order(n r^3)$ or $\Order(n r^\omega)$ by Gabow and Stallmann~\cite{GabowS1986AugmentingPath} to $\Order(n r^2 \log r)$ or $\Order(n r^2)$, depending on whether fast matrix multiplication is employed or not. This improvement significantly narrows the performance gap with the current best randomized algorithm of Cheung, Lau, and Leung~\cite{CheungLL2014AlgebraicAlgorithms}, which runs in $\Order(n r^2)$ or $\Order(n r^{\omega - 1})$ time. Notably, the former of our time bounds aligns with Cunningham's time bound for linear matroid intersection.

The following theorem formally states the time complexity, where the matroid is represented by an $r \times n$ matrix over a field $\bfk$. Following \cite{GabowS1986AugmentingPath,OrlinV1990SolvingLinear}, we adopt the uniform cost model for $\bfk$. This model assumes that each arithmetic operation on $\bfk$ takes constant time, which holds for any fixed finite field. Although the case $\bfk = \mathbb{Q}$ does not generally satisfy this assumption, it can be reduced to polynomially many instances over finite fields~\cite[Theorem~11.3]{IwataK2022WeightedLinear}.

\begin{restatable}{theorem}{ourtimebound}
  \label{thm:our_bound}
  Under the uniform cost model for $\bfk$, the unweighted linear matroid parity problem can be solved deterministically in $\Order(n r^2 \log r)$ time. This complexity is reduced to $\Order(n r^2)$ if fast matrix multiplication is employed.
\end{restatable}

The time bounds can be further refined by using the maximum size $\nu$ of an independent parity set, which is at most $r$. In particular, the complexity of the randomized algorithm is improved to $\Order(nr + n \nu^{\omega-1})$ by Cheung, Kwok, and Lau~\cite{CheungKL2013FastMatrix}. The time bounds of our algorithm are refined as follows.
\begin{restatable}{remark}{ourrefinedtimebound}
  \label{rem:refined_bound}
  In terms of the maximum cardinality $\nu$ of an independent parity set, the complexity bounds in \cref{thm:our_bound} can be refined to $\Order(n r \nu \log \nu)$ and $\Order(n r \nu)$, respectively.
\end{restatable}

\subsection{Related Work}

\paragraph{Deterministic Algorithms for Linear Matroid Parity.}

Lov\'{a}sz discovered a min-max formula \cite{Lovasz1980SelectingIndependent} and presented the first polynomial-time algorithm \cite{Lovasz1981MatroidMatching}, albeit with a high degree, for linear matroid parity. Gabow and Stallmann~\cite{GabowS1986AugmentingPath} provided an augmenting path algorithm, which runs in $\Order(n r^\omega)$ time. Subsequently, Orlin and Vande Vate~\cite{OrlinV1990SolvingLinear} designed an algorithm that solves the linear matroid parity problem via a reduction to a sequence of matroid intersection problems. Orlin~\cite{Orlin2008FastSimpler} later improved this algorithm to run in $\Order(n r^3)$ time, matching the running time bound of the Gabow--Stallmann algorithm when fast matrix multiplication is not employed. 
Regarding the weighted version of this problem, Iwata and Kobayashi~\cite{IwataK2022WeightedLinear} presented a deterministic polynomial-time algorithm based on the augmenting path approach of Gabow and Stallmann, establishing a primal-dual relation for this problem. Independently, Pap~\cite{Pap2013WeightedLinear} obtained a deterministic polynomial-time algorithm based on a different approach.

\paragraph{Randomized Algorithms for Linear Matroid Parity.}

The influential work of 
Lov\'{a}sz~\cite{Lovasz1979DeterminantsMatchings} characterized the linear matroid parity problem via the rank of an \skewsymmetric{} matrix containing indeterminates, thereby obtaining a randomized algorithm based on random substitution. Building upon this approach, Cheung, Lau, and Leung~\cite{CheungLL2014AlgebraicAlgorithms} developed an $\Order(n r^{\omega - 1})$-time algorithm. Furthermore, Cheung, Kwok, and Lau~\cite{CheungKL2013FastMatrix} combined this algorithm with their fast rank computation, leading to a complexity of $\Order(n r + n \nu^{\omega - 1})$. This randomized framework initiated by Lov\'{a}sz has also been extended to the weighted linear matroid parity problem~\cite{CameriniGM1992RandomPseudopolynomial,CheungLL2014AlgebraicAlgorithms}, although the resulting time complexity is pseudopolynomial.

\paragraph{Approximation Algorithms for General Matroid Parity.}

Since solving the general matroid parity problem demands an exponential number of oracle calls~\cite{JensenK1982ComplexityMatroid,Lovasz1980MatroidMatching}, researchers have focused on developing approximation algorithms. A simple greedy algorithm provides an approximate solution of factor $2$, where the approximation ratio is defined as the optimal value divided by the approximate value. Fujito~\cite{Fujito1993$23$approximation} improved this ratio to $3/2$ by refining the greedy solution via a local search that replaces one element with two. Subsequently, Lee, Sviridenko, and Vondr{\'a}k~\cite{LeeSV2013MatroidMatching} demonstrated that a local search algorithm with a sufficiently large neighborhood yields a polynomial-time approximation scheme (PTAS).

\subsection{Organization}

The rest of the paper is organized as follows. \cref{chap:preliminaries} introduces necessary preliminaries, including notation and problem formulation.
\cref{chap:sap_length} provides a linear-algebraic proof for our bound on the shortest augmenting path lengths, which generalizes Cunningham's bound for linear matroid intersection.  
\cref{chap:sap_algo_lmp} presents a shortest augmenting path algorithm for linear matroid parity, followed by its correctness proof. \cref{sec:overall_time_bound} analyzes the overall time complexity of the algorithm. Finally, \cref{chap:conclude} summarizes our results and offers possible future directions.

\section{Augmenting Path Algorithm for Linear Matroid Parity}
\label{chap:preliminaries}

In this section, we formalize the linear matroid parity problem and review the Gabow--Stallmann algorithm, an augmenting path algorithm for linear matroid parity~\cite{GabowS1986AugmentingPath}.

\subsection{Notation and Problem Formulation}

\paragraph{Linear Matroids.} Let $A$ be an $r \times n$ matrix of full row rank over a field $\bfk$. Throughout this paper, we assume the uniform cost model, where arithmetic operations over $\bfk$ are performed in constant time. The \textit{linear matroid} $\bfm(A)$ represented by $A$ has the column vectors of $A$ as its \textit{ground set} $V$, and its \textit{independent sets} are the subsets of columns that are linearly independent. A \textit{base} is an independent set of maximum cardinality $r$. Note that elementary row operations preserve the structure of linear matroids.

\paragraph{Linear Matroid Parity.} Let $\bfm(A)$ be a linear matroid represented by an $r \times n$ matrix $A$ of full row rank over a field $\bfk$, where $n$ is even. The ground set $V$ is partitioned into a collection $L$ of pairs, referred to as \textit{lines}. For each $v \in V$, the paired column is called its \textit{mate} and is denoted by $\bar{v}$. A set is termed a \textit{parity set} if it is a union of lines. The linear matroid parity problem seeks an independent parity set of maximum cardinality. A parity set that is also a base of $\bfm(A)$ is called a \textit{parity base}.

\subsection{Linear-Algebraic Framework}
\label{sec:linear_algebraic_setup_technical_outline}

Conventional augmenting path approaches typically maintain a solution $J$, which is initialized to $\emptyset$, and iteratively expand it via the symmetric difference $J \gets J \symdiff P$ with an augmenting path $P$. 
This section describes a linear-algebraic setup for linear matroid parity that uses elementary row operations to transform the matrix $A$ into a representation suitable for augmenting path searches. 

To this end, the Gabow--Stallmann algorithm~\cite{GabowS1986AugmentingPath} extends the current independent parity set $J$ with a set of auxiliary columns $S$, termed \textit{singletons}, such that $B = J \cup S$ forms a base of the resulting matrix $\bar{A}$ with column set $V \cup S$. Within this framework, a set $P$ is defined as an augmenting path if $B \gets B \symdiff P$ remains a base while the number of lines in $B$ increases by one.

For the extended matrix $\bar{A}$ and a subset of its columns $X \subseteq V \cup S$, let $\bar{A}[X]$ denote the submatrix of $\bar{A}$ consisting of the columns in $X$. By applying row operations to $\bar{A}$ so that the submatrix indexed by $B$ forms the identity matrix, we obtain $\qty[I \mid C]$, where  $C \coloneqq \bar{A}[B]^{-1} \bar{A}[V \setminus J]$ is referred to as the \textit{fundamental cocircuit matrix} with respect to the base $B$. The entries of $C$ are naturally indexed by $(u, v) \in B \times (V \setminus J)$, and $C_{uv}$ is nonzero if and only if $B - u + v$ is a base. More generally, $B \symdiff P$ remains a base if and only if the submatrix indexed by the row set $B \cap P$ and the column set $(V \setminus J) \cap P$, denoted by $C[P]$, is nonsingular; this equivalence follows since the eliminated matrix $\bar{A}[B]^{-1} \bar{A}$ has the following form, where $\alpha = C[P]$ holds:
\vspace{0.9em}
\begin{equation}
\label{eq:pivoting_before}
\bar{A}[B]^{-1} \bar{A} =
\begin{bNiceArray}{cccc}[first-col,first-row,margin]
  & B \cap P & B \setminus P & (V \setminus J) \cap P & (V \setminus J) \setminus P \\
  B \cap P      &        I &             O & \alpha & \beta  \\
  B \setminus P &        O &             I & \gamma & \delta
  \CodeAfter
  \OverBrace[yshift=1pt]{0-2}{0-3}{\raisebox{-4pt}{$\scriptstyle B \symdiff P$}}
\end{bNiceArray}.
\end{equation}
This leads to the following definition of augmenting paths for linear matroid parity.
\footnote{This definition does not require an augmenting path $P$ to be a path in the auxiliary graph $G$ constructed by the augmenting path search (defined in \cref{sec:gs_overview}). In general, $P$ may decompose into an alternating path between singletons and additional alternating cycles.}
\begin{definition}
  \label{def:lmp_aug}
  Let $J$ be an independent parity set and $S$ be a set of singletons such that $B = J \cup S$ forms a base. A set $P \subseteq V \cup S$ is defined to be an \textit{augmenting path} if the following three conditions hold:
  \begin{enumerate*}[label=(\roman*)]
    \item $P$ contains exactly two singletons, i.e., $\abs{P \cap S} = 2$;
    \item $P \setminus S$ is a union of lines; and
    \item the submatrix $C[P]$ is nonsingular.
  \end{enumerate*}
  Moreover, we define the length of an augmenting path $P$ as the number of constituent lines, i.e., $\frac{\abs{P \setminus S}}{2} = \frac{\abs{P}}{2} - 1$.
\end{definition}

\begin{remark}
    \label{rem:ap_odd_length}
    The length of an augmenting path is always odd. Indeed, if $P \cap J$ consists of $k$ lines, $P \setminus B$ must consist of $k + 1$ lines so that $C[P]$ is square, resulting in a total length of $2k + 1$.
\end{remark}

Next, we provide the update formula of the fundamental cocircuit matrix under the base exchange $B \gets B \symdiff P$. As a result of row operations, we obtain $\bar{A}[B \symdiff P]^{-1} \bar{A}$ in the form of
\begin{equation}
\bar{A}[B \symdiff P]^{-1} \bar{A} =
  \begin{bNiceArray}{cccc}[first-col,first-row,margin]
    & B \cap P & B \setminus P & (V \setminus J) \cap P & (V \setminus J) \setminus P \\
    (V \setminus J) \cap P & \alpha^{-1} &             O &      I & \alpha^{-1} \beta \\
    B \setminus P & - \gamma \alpha^{-1} &             I &      O & \delta - \gamma \alpha^{-1} \beta
  \end{bNiceArray}.
\end{equation}
This gives the following update formula of the fundamental cocircuit matrix, which we refer to as \textit{pivoting around} $P$:
\begin{equation}
  \label{eq:pivoting}
  C =
  \begin{bNiceArray}{cc}[first-col,first-row,margin]
    & (V \setminus J) \cap P & (V \setminus J) \setminus P \\
    B \cap P      & \alpha & \beta  \\
    B \setminus P & \gamma & \delta
  \end{bNiceArray}
  \quad \xrightarrow{\text{Pivoting around $P$}} \quad
  \begin{bNiceArray}{cc}[first-col,first-row,margin]
    & B \cap P & (V \setminus J) \setminus P \\
    (V \setminus J) \cap P & \alpha^{-1} & \alpha^{-1} \beta \\
    B \setminus P & - \gamma \alpha^{-1} & \delta - \gamma \alpha^{-1} \beta
  \end{bNiceArray}.
\end{equation}
This operation is also utilized for blossom formation in the weighted linear matroid parity algorithm of \cite{IwataK2022WeightedLinear}. The update requires $\Order(n r \ell)$ time, where $\ell \coloneqq \abs{P}$. This complexity can be further improved by using fast matrix multiplication, initiated by Strassen~\cite{Strassen1969GaussianElimination}. Specifically, the resulting complexity is $\Order(n r \ell^{\omega - 2})$, where $\omega \in [2, 3)$ denotes the exponent of matrix multiplication, with the best published value being $\omega \approx 2.37$ due to \cite{AlmanDWX2025MoreAsymmetry}.

Finally, while singletons facilitate the characterization of augmenting paths, it is unnecessary to explicitly maintain these singletons as extended columns when running the algorithm. Assume that $S$ initially constitutes a unit matrix attached to $A$. In this case, the matrix $A$ itself serves as the fundamental cocircuit matrix $C$. Also, the subsequent updates can be performed by applying Eq.~\eqref{eq:pivoting} and discarding singletons that are excluded from the base.

\subsection{Search Procedure}
\label{sec:gs_overview}

To find an augmenting path that satisfies the conditions in \cref{def:lmp_aug}, the Gabow--Stallmann algorithm utilizes an auxiliary graph $G$ derived from the fundamental cocircuit matrix $C$. The graph $G$ is defined as a bipartite undirected graph on the vertex set $V \cup S$, partitioned into $B = J \cup S$ and $V \setminus J$. An edge exists between $u \in B$ and $v \in V \setminus J$ if and only if $C_{uv} \neq 0$.

The search for an augmenting path traverses the auxiliary graph $G$ in a similar fashion to breadth-first search by maintaining the following two primary structures.
\begin{itemize}
  \item \textbf{Search Paths:} The algorithm explores the graph $G$ to establish search paths starting from singletons. Once a search path $P(v)$ is found for a vertex $v$, the path remains fixed during the search, and $v$ is thereafter referred to as \textit{labeled}. This path $P(v)$ is an odd-length sequence of vertices $v_1 v_2 \dots v_{2k+1}$, where $v_1$ is a singleton and $v_{2k + 1} = v$. The sequence alternates such that $\qty{v_{2i - 1}, v_{2i}}$ forms an edge and $\qty{v_{2i}, v_{2i + 1}}$ represents a line for each $i = 1, \dots, k$.
  \item \textbf{Blossoms:} When the search encounters certain cycles, the algorithm identifies a blossom, a union of lines that effectively functions as a single line. These blossoms are constructed by merging lines and existing outermost blossoms, thereby forming a laminar family.
\end{itemize}

The search process generally consists of two main operations: $\grow$ and $\blossom$. Intuitively, $\grow$ extends the search forest by attaching half-labeled lines, whereas $\blossom$ forms a blossom on the search forest and labels the remaining halves of the lines within the blossom.
\begin{itemize}
  \item \textbf{Grow:} Given a labeled vertex $v$, the $\grow(v)$ operation identifies an adjacent vertex $u$ such that neither $u$ nor $\bar{u}$ is labeled. The operation then assigns a search path to its mate $\bar{u}$, while $u$ remains unlabeled at this stage. When both $u$ and $\bar{u}$ are adjacent to $v$, a special process forming a blossom is triggered, the details of which are omitted here.
  \item \textbf{Blossom:} Suppose that two labeled vertices $u$ and $v$ form an edge, where $P(u)$ and $P(v)$ stem from a common singleton. The $\blossom(u, v)$ operation first backtracks along $P(u)$ and $P(v)$ to their intersection. Each unlabeled vertex encountered during this traversal is assigned a search path formed by extending $P(u)$ or $P(v)$ with the traced segment of the opposite path. Afterward, the process merges all traversed blossoms into a single new blossom. To guarantee the nonsingularity of $C[P]$ for the final output path $P$, a new vertex is inserted at the intersection point via a linear-algebraic technique, the details of which are omitted here.
\end{itemize}

Utilizing these two operations, the search procedure first labels all the singletons and then iterates through labeled vertices in the order of their assignment. For each such vertex $u$, the algorithm executes $\grow(u)$ and subsequently performs $\blossom(u, v)$ for each adjacent labeled vertex $v$ also in the order of assignment, as long as $P(u)$ and $P(v)$ share a common singleton. If $P(u)$ and $P(v)$ stem from different singletons, the search terminates by returning their concatenation as an augmenting path. 

A rigorous analysis of the processing order, combined with the linear-algebraic treatment for blossom formation, ensures that this identified path is indeed an augmenting path. Furthermore, if no augmenting path is identified at the end of the search procedure, one can obtain a dual certificate to show the optimality of the independent parity set, which provides an alternative proof of the min-max formula of Lov\'asz~\cite{Lovasz1980SelectingIndependent}.

\section{Analysis of Shortest Augmenting Path Lengths}
\label{chap:sap_length}

\subsection{Motivation and Overview}

In the augmenting path algorithm for linear matroid intersection, the pivoting operation along an augmenting path, which incurs a cost proportional to the path length $\ell$, dominates the overall time complexity. Consequently, as demonstrated by Cunningham~\cite{Cunningham1986ImprovedBounds}, the total time complexity can be reduced by deriving a tighter bound on the lengths of shortest augmenting paths. 

Since pivoting is also the computational bottleneck in the Gabow--Stallmann algorithm, establishing an equivalent bound is vital for the shortest augmenting path approach to linear matroid parity. To this end, we prove the following lemma, generalizing Cunningham's result for the linear matroid intersection case to the parity setting.
\begin{restatable}{lemma}{paritybound}
  \label{lem:matroid_parity_number_of_augmenting_paths}
  Let $\nu$ denote the maximum size of an independent parity set. For any independent parity set $J$, there exist $\qty(\nu - \abs{J}) / 2$ disjoint augmenting paths.
\end{restatable}
This lemma bounds the length $\ell$ of the shortest augmenting path as $\Order(\nu / (\nu - \abs{J}))$ via a simple counting argument (see \cref{cor:lmp_sap_length}). Consequently, assuming that the shortest augmenting path can be efficiently identified in each iteration, this yields an improved total time complexity of $\Order(n r \nu \log \nu) = \Order(n r^2 \log r)$, which matches Cunningham's time bound.

Cunningham originally derived the counterpart of this bound for general matroid intersection via a graphical approach. However, such a strategy appears difficult to extend to linear matroid parity, because it neglects the linear-algebraic representation that is indispensable for efficiently solving linear matroid parity.

To overcome this challenge, we develop a novel linear-algebraic approach to prove \cref{lem:matroid_parity_number_of_augmenting_paths}. In what follows, we first introduce linear-algebraic tools employed in the proof, then present an initial attempt that highlights a technical difficulty, and finally detail the complete proof. 

While our result for linear matroid parity encompasses Cunningham's result for the linear matroid intersection case, our proof can also be tailored to the intersection setting using analogous tools. Although we omit the details, the corresponding difficulty can be readily resolved in this setting, which indicates that the technical difficulty is intrinsic to linear matroid parity.

\subsection{Linear-Algebraic Tools}

\paragraph{Pfaffians.} A square matrix $K$ is \textit{\skewsymmetric{}} if $K = -K^{\top}$ and its diagonal entries are all zero. We note that the latter requirement is redundant unless the underlying field has characteristic two. Let $\symgrp_{n}$ be the symmetric group of degree $n$. We define $F_{2k} \subseteq \symgrp_{2k}$ as the set of permutations $\sigma$ such that $\sigma(2i - 1) < \sigma(2i)$ for each $i$ and $\sigma(1) < \sigma(3) < \dots < \sigma(2k - 1)$. This set $F_{2k}$ corresponds to the set of perfect matchings on $\qty{1, \dots, 2k}$. The \textit{Pfaffian} of a $2k \times 2k$ \skewsymmetric{} matrix $K$ is then defined by
\begin{equation}
  \Pf K = \sum_{\sigma \in F_{2k}} \sgn \sigma \prod_{i = 1}^{k} K_{\sigma(2i - 1), \sigma(2i)}.
\end{equation}
A fundamental property of the Pfaffian is the identity $(\Pf K)^2 = \det K$, which implies that an \skewsymmetric{} matrix $K$ of even size is nonsingular if and only if $\Pf K \neq 0$.

\paragraph{Algebraic Formulation of Linear Matroid Parity.} The following proposition characterizes the maximum cardinality of an independent parity set via an \skewsymmetric{} matrix containing indeterminates.

\begin{proposition}[\cite{GeelenI2005MatroidMatching}]
  \label{prop:rank_lmp}
  Let $x_1, \dots, x_{n/2}$ be distinct indeterminates, and $\Delta$ the $n \times n$ \skewsymmetric{} block-diagonal matrix consisting of $n / 2$ blocks, each of which is the $2 \times 2$ matrix $\mqty(0 & x_i \\ -x_i & 0)$.
  Consider the \skewsymmetric{} matrix $\Phi = \begin{bmatrix}
      O         & A      \\
      -A^{\top} & \Delta
    \end{bmatrix}$.
  Then, the maximum size $\nu$ of an independent parity set is characterized by the rank of $\Phi$ as $\nu = \rank \Phi - n$.
\end{proposition}

\begin{corollary}
  \label{cor:pf_lmp}
  A parity base exists if and only if the matrix $\Phi$ is nonsingular.
\end{corollary}

\subsection{Initial Strategy}
\label{sec:initial_strategy}

Our strategy is based on two key observations. First, the auxiliary graph $G$ from the Gabow--Stallmann algorithm (defined in \cref{sec:gs_overview}) is a subgraph of the \textit{support graph} $G_\Phi$ of the \skewsymmetric{} matrix $\Phi$, which we define later. Second, the nonsingularity of $\Phi$, combined with the definition of the Pfaffian, yields a collection of ``alternating'' paths. However, these paths may fail to be augmenting paths, necessitating a more sophisticated strategy for path construction in \cref{sec:complete_approach}. 

In what follows, we detail the procedure for constructing the alternating paths. For simplicity, we assume the existence of a parity base, i.e., the nonsingularity of $\Phi$ (\cref{cor:pf_lmp}). Since the fundamental cocircuit matrix $C$ is obtained via row operations on $A$, we may apply row and column operations to $\Phi$ to yield the following form:
\begin{equation}
  \label{eq:phi_reduced}
\Phi=
  \begin{bNiceArray}{cccc}[first-row,first-col,margin,hvlines,corners]
                  & S                                         & J  & J'                  & V \setminus J  \\
    S             & \Block{2-2}{O}                            &    & O                   & \Block{2-1}{C} \\
    J             &                                           &    & I                   &                \\
    J'            & O                                         & -I & \Block{2-2}{\Delta} &                \\
    V \setminus J & \Block{1-2}{\raisebox{-2pt}{$-C^{\top}$}} &    &                     &                \\
  \end{bNiceArray}.
\end{equation}
For notational simplicity, we continue to denote the resulting matrix by $\Phi$. Since the second and third blocks both have $J$ as their index set, we denote the latter by $J'$ to avoid ambiguity. Accordingly, we extend the set $L$ of lines to include the corresponding lines on $J'$, and denote the resulting set of lines by $L_\Phi$.

\begin{figure}
  \centering
  \includegraphics[width=0.45\linewidth]{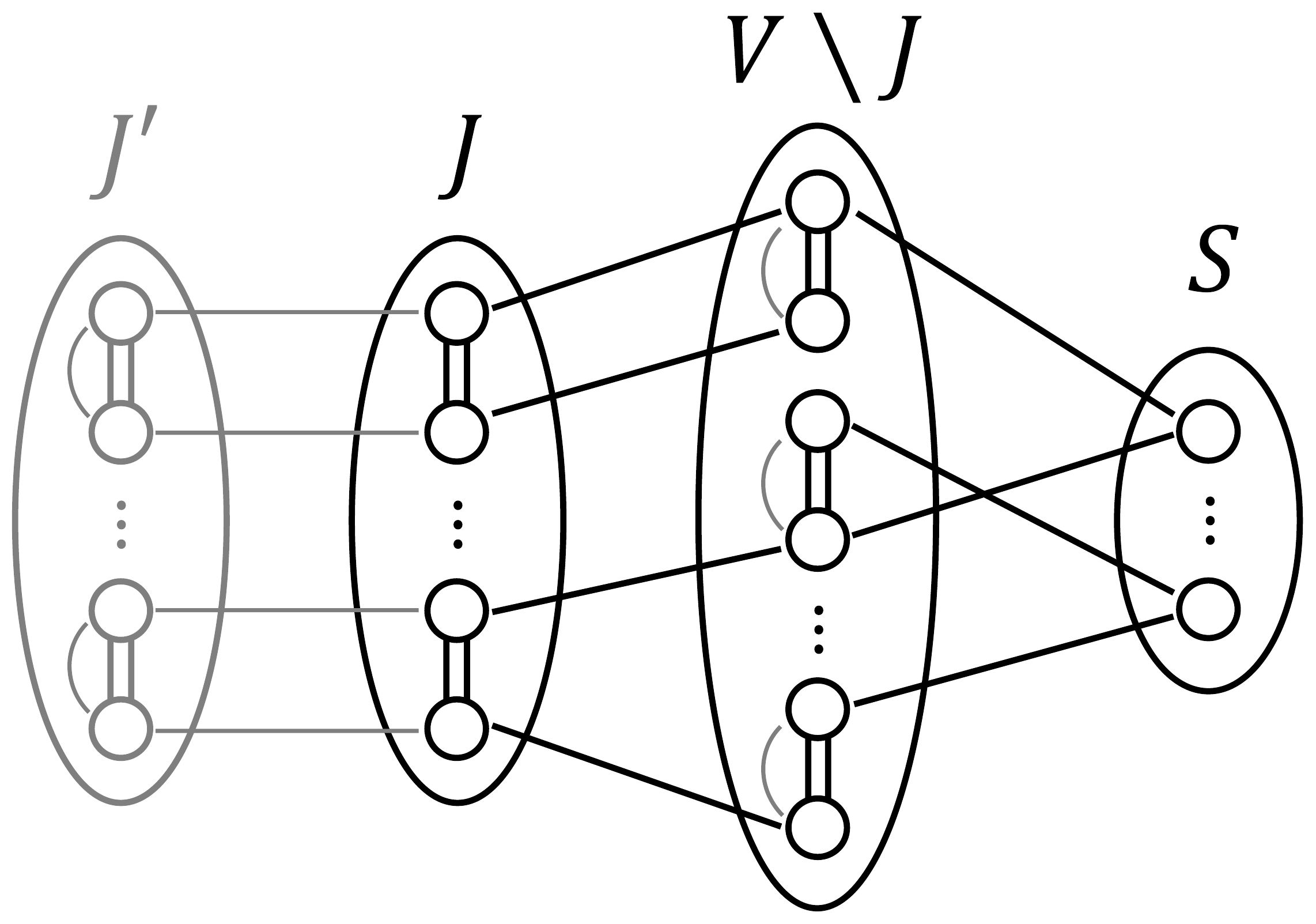}
  \caption{The undirected graph $G_{\Phi}$ and the set $L_\Phi$ of lines defined over the index set of $\Phi$. Single bonds represent the edges of $G_{\Phi}$ and double bonds represent the lines in $L_\Phi$. The subgraph $G$ and the line set $L$ are shown in black, while the remaining parts of $G_\Phi$ and $L_\Phi$ are shown in gray.}
  \label{fig:matroid_parity_pfaffian_graph}
\end{figure}

The support graph $G_\Phi = (V_\Phi, E_\Phi)$ of the matrix $\Phi$ is an unweighted undirected graph (see \cref{fig:matroid_parity_pfaffian_graph}), where $V_\Phi$ is the row/column set of $\Phi$, and $(u, v) \in E_\Phi$ if and only if $\Phi_{uv} \neq 0$. Notably, this construction embeds the auxiliary graph $G$ as a subgraph of $G_\Phi$.

Let $2k = (r + n)$ be the order of the matrix $\Phi$, which is necessarily even because of the existence of a parity base. Since the matrix $\Phi$ is nonsingular, it holds that $\Pf \Phi \neq 0$, i.e.,
\begin{equation}
  \label{eq:Pfaffian_nonzero_technical_outline}
  \sum_{\sigma \in F_{2k}} \qty(\sgn \sigma \prod_{i = 1}^{k} \Phi_{\sigma(2i - 1), \sigma(2i)}) \neq 0.
\end{equation}
Let $\sigma \in F_{2k}$ be a permutation such that the corresponding term $f(\sigma)$ in the summation is nonzero. Since this implies $\Phi_{\sigma(2i - 1), \sigma(2i)} \neq 0$ for all $i = 1, \dots, k$, the permutation directly defines a perfect matching $M_\sigma \coloneqq \qty{(\sigma(2i - 1), \sigma(2i)) \mid 1 \le i \le k}$ on $G_{\Phi}$. We now consider the graph $G_{\sigma} \coloneqq (V_\Phi, M_{\sigma} \uplus L_\Phi)$, where $\uplus$ denotes the union as multisets; $G_\sigma$ may thus contain parallel edges. Since $L_\Phi$ forms a perfect matching on $V_{\Phi} \setminus S$, the graph $G_{\sigma}$ consists of alternating cycles (potentially of length two) and $\abs{S} / 2$ alternating paths between two distinct vertices in $S$. Here, ``alternating'' signifies that the edges of these components alternate between $M_\sigma$ and $L_\Phi$. Moreover, if an edge of $G_\sigma$ connects two corresponding vertices $u \in J$ and $u' \in J'$, this edge must be part of a cycle on $\qty{u, \bar{u}, u', \bar{u}'}$. Consequently, the alternating paths never intersect $J'$ and are restricted to $V \cup S$, implying that they alternate between the edges of $G$ and the lines in $L$. 

The above discussion constructs $\abs{S} / 2$ disjoint alternating paths from the auxiliary graph $G$. For each path $P$, the induced subgraph $G[P]$ possesses a perfect matching, implying that $\det C[P]$ contains a nonzero term in its summation. However, this argument alone cannot eliminate the possibility of cancellation in the final sum, and thus falls short of guaranteeing that these paths are augmenting paths.\footnote{This issue is specific to linear matroid parity. In linear matroid intersection, since alternating paths without shortcuts are always augmenting paths, removing shortcuts from the obtained paths readily yields augmenting paths.}

\subsection{Complete Approach}
\label{sec:complete_approach}

A naive choice of $\sigma \in F_{2k}$ with $f(\sigma) \neq 0$ does not necessarily prevent an accidental cancellation within $\det C[P]$ for each constructed path $P$. To ensure nonsingularity, we group the terms of $\Pf \Phi$ according to their associated path and cycle structures. Among groups with a nonzero sum, we select the one corresponding to the finest partition of the vertex set. This choice induces a decomposition of $\Phi$ into nonsingular blocks, which guarantees that the paths are augmenting paths.

To formalize this, we define a partial order on the set of partitions based on refinement. Let $\Pi(X)$ denote the collection of all partitions of a finite set $X$. A partition $\calq \in \Pi(X)$ is a \textit{refinement} of $\calp \in \Pi(X)$, denoted by $\calq \le \calp$, if every component of $\calq$ is a subset of some component of $\calp$. 

Building upon the initial attempt, we first present a complete proof under the assumption that a parity base exists, and then reduce the general case to this setting.

\begin{lemma}
  \label{lem:matroid_parity_number_of_augmenting_paths_parity_base}
  Consider an instance of the linear matroid parity problem in which a parity base exists. Then, for any independent parity set $J$, there exist $\qty(r - \abs{J}) / 2$ disjoint augmenting paths.
\end{lemma}

\begin{proof}
In the following proof, we use the notation introduced in \cref{sec:initial_strategy}. We first define a map $\Gamma\colon F_{2k} \to \Pi(V_{\Phi})$, where $\Gamma(\sigma)$ is the set of all connected components of $G_{\sigma}$. Then, the sum in the definition of the Pfaffian can be partitioned as
\begin{equation}
  \Pf \Phi = \sum_{\sigma \in F_{2k}} f(\sigma) = \sum_{\calp \in \Pi(V_{\Phi})} \qty(\sum_{\sigma \in \Gamma^{-1}(\calp)} f(\sigma)).
\end{equation}
Since this is nonzero, there exists a partition $\calp \in \Pi(V_{\Phi})$ such that $\sum_{\sigma \in \Gamma^{-1}(\calp)} f(\sigma) \neq 0$. We choose $\calp$ to be minimal under the refinement order; that is, $\sum_{\sigma \in \Gamma^{-1}(\calq)} f(\sigma) = 0$ for every $\calq \lneq \calp$. The sum over all refinements of $\calp$ then yields
\begin{equation}
  \label{eq:refinements_sum_technical_outline}
  \sum_{\calq \in \Pi(V_\Phi) \colon \calq \le \calp} \qty(\sum_{\sigma \in \Gamma^{-1}(\calq)} f(\sigma)) \neq 0.
\end{equation}
Crucially, the set $\bigcup_{\calq \le \calp} \Gamma^{-1}(\calq)$ corresponds to the set of matchings whose edges are restricted to the individual components of $\calp$. Consequently, the sum in Eq.~\eqref{eq:refinements_sum_technical_outline} is equivalent to the Pfaffian of a block-diagonal matrix $\Phi_\calp$ with each block being the principal submatrix $\Phi[X]$ for $X \in \calp$. Since Eq.~\eqref{eq:refinements_sum_technical_outline} ensures the nonsingularity of $\Phi_\calp$, each block $\Phi[X]$ is also nonsingular.

Since there exists $\sigma \in \Gamma^{-1}(\calp)$ with $f(\sigma) \neq 0$, the discussion in \cref{sec:initial_strategy} shows that $\calp$ contains $\abs{S} / 2$ alternating paths that connect distinct vertices in $S$ and are disjoint from $J'$. Hence, it remains to verify that each alternating path $P \in \calp$ is an augmenting path. Since $P \cap J' = \emptyset$, the nonsingular block $\Phi[P]$ can be expressed as
\begin{equation}
  \Phi[P]=
  \begin{bNiceArray}{cc}[first-row,first-col,margin]
    & P \cap (J \cup S) & P \cap (V \setminus J)  \\
    P \cap (J \cup S) & O & C[P] \\
    P \cap (V \setminus J) & -C[P]^{\top} & * \\
  \end{bNiceArray}.
\end{equation}
Since these sub-blocks of $\Phi[P]$ are square and share the same size, the upper-right sub-block $C[P]$ is nonsingular. As the remaining requirements in \cref{def:lmp_aug} follow directly, $P$ is an augmenting path. 
\end{proof}

Next, we generalize this lemma to general cases where a parity base may not exist. In the following proof, we leverage the proof of \cite[Theorem 2.1]{GabowS1986AugmentingPath} to identify a maximum independent parity set $J^*$ whose span includes that of $J$. While we prove the existence of $J^*$ via a combinatorial argument, it also follows from the validity of augmenting path algorithms.
\paritybound*
\begin{proof}
    We first prove the existence of a maximum independent parity set $J^*$ whose span includes that of $J$. Let $J^*$ be a maximum independent parity set that minimizes $\abs{J \symdiff J^*}$. Assume for contradiction that there exists a vector $v \in J$ that is not spanned by $J^*$. Then, $J^* \cup \qty{v}$ is independent, whereas $J^* \cup \qty{v, \bar{v}}$ is dependent by the maximality of $J^*$. Thus, $\bar{v}$ can be uniquely represented as a linear combination of $J^* \cup \qty{v}$. Since $J$ is independent and contains $\qty{v, \bar{v}}$, this representation must have a nonzero coefficient for some $w \in J^* \setminus J$. Consequently, $J^* \cup \qty{v, \bar{v}} \setminus \qty{w}$ is independent. However, this implies that $J^* \cup \qty{v, \bar{v}} \setminus \qty{w, \bar{w}}$ is also a maximum independent parity set, contradicting the minimality of $\abs{J \symdiff J^*}$.

    To facilitate the reduction to \cref{lem:matroid_parity_number_of_augmenting_paths_parity_base}, let $S$ be a set of singletons such that $B = J^* \cup S$ forms a base of the linear matroid $\bfm(\bar{A})$ on $V \cup S$ such that its restriction on $V$ is identical to $\bfm(A)$. In matroid theoretic terminology, we now contract $S$ from the linear matroid $\bfm(\bar{A})$ to obtain a new linear matroid $\bfm(A')$ on $V$. This can be done by performing Gaussian elimination on the columns of $B$ in $\bar{A}$ and removing the rows and columns corresponding to $S$, which yields a $\nu \times n$ matrix $A'$. 
    
    The resulting linear matroid $\bfm(A')$ has a ground set identical to that of $A$ and satisfies the following properties:
    \begin{enumerate*}[label=(\arabic*)]
        \item any independent set of $\bfm(A')$ is also independent in $\bfm(A)$;
        \item $J$ remains an independent parity set of $\bfm(A')$; and
        \item $J^*$ is a parity base of $\bfm(A')$.
    \end{enumerate*}
    Consequently, applying \cref{lem:matroid_parity_number_of_augmenting_paths_parity_base} to this matroid yields $(\nu - \abs{J}) / 2$ disjoint augmenting paths for $J$ in $\bfm(A')$, which are also augmenting paths for $J$ in $\bfm(A)$.
\end{proof}
\begin{remark}
    Following the reduction in \cite[Section 6.5]{CheungLL2014AlgebraicAlgorithms}, one can develop an alternative linear-algebraic reduction. Because the principal submatrix $\Phi[V_\Phi \setminus S]$ is nonsingular, it can be extended to a nonsingular submatrix of order $\rank \Phi = \nu + n$. Although this submatrix may not be a principal submatrix, the argument in \cite[Proposition 7.3.6]{Murota2010MatricesMatroids} allows for transforming it into a nonsingular principal submatrix of the same size containing $\Phi[V_\Phi \setminus S]$. Consequently, applying the same argument as in \cref{lem:matroid_parity_number_of_augmenting_paths_parity_base} to this principal submatrix yields the desired augmenting paths.
\end{remark}

Finally, we derive the bound on the length of the shortest augmenting path, which is used in \cref{sec:overall_time_bound} to analyze the time complexity of our shortest augmenting path algorithm.
\begin{corollary}
  \label{cor:lmp_sap_length}
  Let $\nu$ denote the maximum size of an independent parity set. For a non-maximum independent parity set $J$, the shortest augmenting path length is at most $\frac{2\abs{J}}{\nu-|J|} + 1$.
\end{corollary}
\begin{proof}
  Consider the $(\nu - \abs{J}) / 2$ disjoint augmenting paths provided by \cref{lem:matroid_parity_number_of_augmenting_paths}. Since $J$ contains $\abs{J} / 2$ lines, the shortest among these paths contains at most $\frac{\abs{J}}{\nu - \abs{J}}$ lines from $J$. \cref{rem:ap_odd_length} then ensures that this path contains at most $\frac{\abs{J}}{\nu - \abs{J}} + 1$ lines from $V \setminus J$, yielding a total length of at most $\frac{2\abs{J}}{\nu-|J|} + 1$. Since the overall shortest augmenting path cannot be longer than this path, the desired bound follows. 
\end{proof}

\subsection{Application to Local Search}
\label{sec:local_search}

We conclude this section by demonstrating a direct consequence of applying \cref{cor:lmp_sap_length} to the local search algorithm of Lee, Sviridenko, and Vondr{\'a}k~\cite{LeeSV2013MatroidMatching}. Their algorithm iteratively performs a brute-force search for augmenting paths of length at most $2 s(\epsilon) + 1$, effectively replacing $s(\epsilon)$ lines in $J$ with $s(\epsilon) + 1$ lines. They analyzed this algorithm for general matroid parity and established a $(1 + \epsilon)$-approximation by setting $s(\epsilon) = 5^{\lceil 1 / (2 \epsilon) \rceil}$, thereby providing a polynomial-time approximation scheme (PTAS).

Our analysis of the shortest augmenting paths demonstrates that $s(\epsilon) = \lceil 1 / \epsilon \rceil$ is sufficient for linear matroid parity, generalizing the well-known parameter bound for graph matching.
\begin{corollary}
  \label{cor:local_search_lmp}
  The local search algorithm of Lee, Sviridenko, and Vondr\'ak with the neighborhood size $s(\epsilon) = \lceil 1 / \epsilon \rceil$ gives a $(1 + \epsilon)$-approximation algorithm for linear matroid parity.
\end{corollary}
\begin{proof}
  When the algorithm terminates, the length of the shortest augmenting path is more than $2 s(\epsilon) + 1$. Hence, \cref{cor:lmp_sap_length} gives
  $\displaystyle 2 s(\epsilon) + 1 < \frac{2 \abs{J}}{\nu - \abs{J}} + 1$, which is equivalent to $\displaystyle\frac{s(\epsilon)}{1 + s(\epsilon)} \nu < \abs{J}$. 
 This guarantees a $(1 + \epsilon)$-approximation by setting $s(\epsilon) = \lceil 1 / \epsilon \rceil \ge 1 / \epsilon$.
\end{proof}

While known exact algorithms for linear matroid parity rely on the matrix representation of the linear matroid, this algorithm also applies to cases where such a representation is unavailable, as long as the independence oracle is accessible.

\section{Search Procedure for a Shortest Augmenting Path}
\label{chap:sap_algo_lmp}

\subsection{Technical Outline}

\subsubsection{Key Features of Our Algorithm}

Although our search procedure for a shortest augmenting path follows the Gabow--Stallmann framework presented in \cref{sec:gs_overview}, it is primarily distinguished by the synchronization strategy for $\grow$ and $\blossom$ operations. Specifically, while the Gabow--Stallmann algorithm processes vertices in a simple sequential order, our algorithm adopts level-based synchronization inspired by the Micali--Vazirani algorithm for graph matching~\cite{MicaliV1980$Osqrt|V||E|$Algorithm}.

\cref{alg:sap_lmp_informal} outlines our proposed search procedure for a shortest augmenting path. We define the \textit{level} of a vertex $v$, denoted by $\lv(v)$, as the number of lines within $P(v)$. The process is partitioned into search levels $k = 0, 1, 2, \dots$, where each phase $k$ first executes $\grow(v)$ on vertices of level $k$ and then performs $\blossom(u, v)$ for edges with a level sum of $2k + 1$. Because a parity argument ensures this sum is always odd, the procedure eventually exhausts all edges between labeled vertices. Within each phase $k$, the $\blossom$ operations are ordered by their maximum level; any remaining ties in $\grow$ or $\blossom$ are broken by search order, following the Gabow--Stallmann algorithm.

\begin{algorithm}[tb]
  \caption{Ordering of operations for finding a shortest augmenting path (informal).}
  \label{alg:sap_lmp_informal}

  \begin{algorithmic}[1]
    \For{$k = 0, 1, \dots$} \Comment{Vertices are always selected in the order they were labeled.}
        \For{$v$ with $\lv(v) = k$}
            \State $\grow(v)$
        \EndFor
        \For{$\kmax = k+1, k+2, \dots, 2k+1$}
            \State $\kmin \gets 2k + 1 - \kmax$
            \For{$u$ with $\lv(u) = \kmax$}
                \For{$v$ with $\lv(v) = \kmin$}
                    \State $\blossom(u, v)$ \Comment{An augmenting path is found if their starting singletons differ.}
                \EndFor
            \EndFor
        \EndFor
    \EndFor
  \end{algorithmic}
\end{algorithm}

While this modification is conceptually straightforward and seems promising for finding a shortest augmenting path, establishing a formal proof of correctness presents significant challenges identified in \cref{sec:background}. In particular, the challenge of defining shortest search paths a priori compels us to replace the linear-algebraic process in blossom formation with the variant employed in the weighted linear matroid parity algorithm~\cite{IwataK2022WeightedLinear}, as it is more suitable for analyzing the properties of search path lengths.

We briefly review this linear-algebraic process, which is applied to the intersection of $P(u)$ and $P(v)$ during the execution of $\blossom(u, v)$. Every time this process is invoked, it introduces a \textit{dummy line}, which consists of two auxiliary vertices. These lines facilitate the search for an augmenting path, enabling the algorithm to identify a path $P$ such that, even if $P$ contains dummy lines, the remaining component induces a nonsingular submatrix. Since dummy lines are not used in the base update, they are excluded from the definition of levels, which is intended to quantify the augmenting path length and its corresponding base update cost.

Each dummy line is inserted by first extending the fundamental cocircuit matrix with a new row and column, and then pivoting around it. At any point during the algorithm, we denote the set of these newly added vertices by $T$, and the matrix obtained so far by $C^*$. Then, the nonsingularity of $C[P \setminus T]$ is guaranteed via the nonsingularity of $C^*[P]$ and the ``coherence'' of $P$ with dummy lines, a structural property defined later.

Finally, since a naive implementation of the entire algorithm suffers several computational bottlenecks, such as blossom management and pivoting operations, \cref{sec:implementation_details} describes an efficient procedure for these processes. Consequently, given the fundamental cocircuit matrix, one can identify a shortest augmenting path in $\Order(n r)$ time. Accordingly, the total $\nu / 2$ search iterations remain within the time bounds claimed in \cref{thm:our_bound} and \cref{rem:refined_bound}.

\subsubsection{Correctness Proof}

We establish the correctness of our algorithm by proving the following three properties: (i) \textit{validity:} the identified path is indeed an augmenting path; (ii) \textit{minimality:} the discovered path is the shortest; and (iii) \textit{optimality:} the termination without finding a path guarantees that the current independent parity set is maximum.

The proofs for validity and optimality are built upon the theoretical frameworks of \cite{GeelenIM2003LinearDeltamatroid,IwataK2022WeightedLinear}, which relate the nonsingularity of \skewsymmetric{} matrices to perfect matchings on their support graphs. In what follows, we focus on the proof of minimality. As noted in the introduction, we characterize shortest search paths by first establishing an a priori lower bound on search path lengths and then verifying that our algorithm attains this bound.

\paragraph{Lower Bound on Search Path Lengths.}
The validity proof demonstrates that, for every labeled vertex $v$, the set $P(v) \setminus \qty{v}$ is structurally analogous to an augmenting path: it induces a nonsingular submatrix $C^*[P(v) \setminus \qty{v}]$ and remains ``coherent'' with dummy lines, as defined later. For an augmenting path $P$, these properties collectively guarantee the nonsingularity of $C[P \setminus T]$. By formally characterizing such a structure as a \textit{coherent nonsingular path} to $\bar{v}$, we can naturally derive a lower bound on $\lv(v)$ as follows. We define the length of a path $P$ as the number of lines (excluding dummy lines) that intersect $P$, and denote by $\dist(x)$ the minimum length of a coherent nonsingular path to $x$. Since $P(v) \setminus \qty{v}$ is a coherent nonsingular path to $\bar{v}$, we obtain the lower bound $\lv(v) \ge \dist(\bar{v})$. Crucially, this bound is invariant under the insertion of dummy lines. Furthermore, this bound naturally applies to the case where $v \in T$, whereas establishing a similar bound within the linear-algebraic framework of Gabow and Stallmann is nontrivial.

\paragraph{Attaining the Bound.}
We prove that the proposed algorithm attains the bound $\lv(v) = \dist(\bar{v})$ by induction on the search level $k$. Intuitively, since the algorithm assigns levels from smaller to larger values, each vertex $v$ has a ``deadline'' for level assignment derived from the bound. To ensure that each assignment is performed within this deadline, our core task is to verify the existence of an associated ``trigger'' for the $\grow$ and $\blossom$ operations. In the following proof sketches, we omit dummy lines for clarity.

For $\grow$ operations, the trigger for labeling $v$ is a labeled vertex $u$ adjacent to $\bar{v}$. The existence of such a vertex $u$ is guaranteed via the Laplace expansion. Suppose a $\grow$ operation must assign level $k + 1$ to a vertex $v$, i.e., $\dist(\bar{v}) = k + 1$. By applying the Laplace expansion to the row or column of $\bar{v}$ in $C^*[P]$, where $P$ is a shortest coherent nonsingular path to $\bar{v}$, we can decompose $P$ into an edge $(u, \bar{v})$ and a shorter coherent nonsingular path $P \setminus \qty{u, \bar{v}}$ to $\bar{u}$. Since this subpath ensures $\dist(\bar{u}) = k$, the inductive assumption guarantees that $\lv(u) = k$. Consequently, $\grow(u)$ will identify $\bar{v}$ and assign level $k + 1$ to $v$ as required.  

For $\blossom$ operations, the trigger is an \textit{incoming edge} between labeled vertices, a concept originally introduced by Izumi, Kitamura, and Yamaguchi~\cite{IzumiKY2025ForgettingAlternation} for graph matching. Informally, an incoming edge is one that enters the subtree of the search forest rooted at $v$, excluding the forest edge itself. This characterization ensures that, if $\blossom$ is performed on the incoming edge,  the backtracking process labels $v$. To guarantee the existence of an incoming edge that meets the deadline, we employ a variant of the decomposition technique developed in \cref{sec:complete_approach}. Specifically, a shortest coherent nonsingular path to $v$ is decomposed into an incoming edge $(y, z)$ and two coherent nonsingular paths $Y$ and $Z$ (see \cref{lem:yz_decomp}). By the inductive assumption, these paths ensure that $y$ and $z$ are labeled with levels sufficiently small to trigger the operation in time.

\subsubsection{Section Organization}

The remainder of this section is organized as follows. \cref{sec:lmp_algo_description} introduces necessary notation and details our procedure for finding a shortest augmenting path.
\cref{sec:super_singleton_reduction} provides a reduction from the multi-source to a single-source setting, which facilitates the subsequent proofs.
\cref{sec:basic_properties} establishes several fundamental properties of the algorithm.
\cref{sec:lmp_validity} shows the validity: the algorithm finds an augmenting path defined in \cref{def:lmp_aug}.
\cref{sec:lmp_minimality} proves the minimality: the augmenting path found by the algorithm is the shortest.
\cref{sec:lmp_optimality} shows the optimality: the current solution is optimal if the algorithm does not find any augmenting paths.
\cref{sec:implementation_details} presents the implementation details required to achieve the desired efficiency.

\subsection{Algorithm Description}
\label{sec:lmp_algo_description}

\subsubsection{Notation and Definitions}

This section introduces the notation and definitions required to describe our shortest augmenting path algorithm and verify its correctness.

\paragraph{Search Paths.}
For a search path $P(x)$ and an intermediate vertex $u \in P(x)$, let $P(x|u)$ denote the subpath of $P(x)$ starting from the successor of $u$ (thereby excluding $u$ itself). We also write $\overline{P}$ for the reverse of a path $P$, and $PQ$ for the concatenation of paths $P$ and $Q$ in this order. While each search path $P(x)$ is formally defined as a sequence of vertices, we also let $P(x)$ denote its vertex set when the context is clear.

As discussed in \cref{sec:gs_overview}, the Gabow--Stallmann algorithm maintains search paths to labeled vertices, which remain fixed once assigned. While our algorithm also relies on search paths, these paths may be modified upon the insertion of dummy lines, as detailed later in \ref{item:detail_4}. Although this process appears to drastically alter the paths, \cref{cor:justify_line_4} demonstrates that the procedure simply inserts the new dummy line into the middle of each path.

\paragraph{Dummy Lines and Blossoms.}

Each blossom is assigned a specific auxiliary line called its \textit{dummy line}. When the algorithm forms a new blossom $H$, it either introduces a new dummy line or reuses the dummy line of a sub-blossom $H' \subseteq H$; in the latter case, the same dummy line is shared by both $H$ and $H'$. While $H$ may contain multiple dummy lines as part of its constituent lines, the dummy line of $H$ only refers to the one assigned at its formation.

Each dummy line consists of a \textit{bud} $b$ and a \textit{tip} $t$, which are both contained in the associated blossom $H$.\footnote{Note that this definition differs from that of \cite{IwataK2022WeightedLinear}, where blossoms act as shrunk vertices. Following \cite{GeelenIM2003LinearDeltamatroid}, blossoms act as shrunk lines in this paper.} For convenience, we define $\bar{b} = t$ and $\bar{t} = b$. A dummy line is introduced by first appending a new row and a new column to the fundamental cocircuit matrix and then pivoting around them. We distinguish the initially existing lines as \textit{original lines}, and collectively refer to both original lines and dummy lines as \textit{lines}. Let $T$ denote the set of buds and tips. Since introducing these vertices affects several algorithmic notions, let $B^*$, $C^*$, and $G^* = (V^*, E^*)$ denote the resulting base, fundamental cocircuit matrix, and graph, respectively. It follows that $V^* = V \cup S \cup T$.

As discussed later, the insertion of a dummy line involves a specific vertex $g$, whose incident edges relevant to the search are effectively removed by the associated pivoting operation. Since $g$ can no longer be used for further search, the algorithm not only labels $g$ but also marks $g$ as \textit{retired}. Hence, retired vertices are always labeled. This definition is crucial for establishing the subsequent validity proof and achieving the desired time bound. For convenience, a vertex is termed \textit{active} if it is not retired.

For each vertex $x \in V \cup T$, we define $K(x)$ as the maximal blossom $H$ containing $x$. If no such blossom exists, we define $K(x) = \qty{x, \bar{x}}$. Note that $K(x)$ is not defined on singletons, because they never belong to any blossom during the search. While blossoms are crucial for the proof of correctness, the algorithm only needs to maintain the maximal blossoms $K(x)$ during its execution.

\paragraph{Grow and Blossom Operations.}

For each line, we define the \textit{lower} vertex as the one labeled first and the \textit{upper} vertex as the other. The lower vertex of an original line is labeled by $\grow$, whereas that of a dummy line is labeled upon insertion. In either case, the upper vertex of any line is always labeled by $\blossom$. The terms ``lower'' and ``upper'' reflect the algorithmic property that, once both levels are determined, the level of the former is no greater than that of the latter. Note that these definitions do not yet apply to lines whose vertices both remain unlabeled.

The $\grow$ operation records a predecessor $\rho(u) \coloneqq v$ when finding a new original line $\qty{u, \bar{u}}$ via an edge $(v, u)$. The predecessor $\rho(u)$ is defined exclusively for upper vertices $u$ and is utilized when labeling $u$ in the $\blossom$ operation. 

Upon labeling a vertex $x$, the algorithm description also specifies a suffix $L(x)$ of the search path $P(x)$, which is defined only for the sake of proof and is not required for running the algorithm. Specifically, this fragment is utilized for providing a search path decomposition and proving that the found path is an augmenting path. Unlike search paths, $L(x)$ remains fixed once established.

\paragraph{Synchronization Strategy for Operations.}

Finally, we introduce the notation necessary for formally describing our synchronization strategy. The \textit{level} of a labeled vertex $v$, denoted by $\lv(v)$, is defined as the number of original lines contained within the search path $P(v)$. Although this count excludes dummy lines, $\lv(v)$ is still defined for labeled vertices of dummy lines.

The algorithm also maintains an array of lists $\vl$, where each list can be scanned sequentially and extended only at its end. The $k$-th list $\vl[k]$ stores the vertices of level $k$ in the order they were labeled. For computational efficiency, the algorithm omits from these lists vertices that do not contribute to the search, namely tips and retired vertices. Initially, $\vl[0]$ is set to $S$ in an arbitrary order, and the other lists are set to empty. The search order is determined based on $\vl$.

For ease of description, we define a total order $\le$ on the set of all labeled vertices. For labeled vertices that are inserted into $\vl$, we write $u<v$ if $\lv(u) < \lv(v)$, or if $\lv(u) = \lv(v) \eqqcolon k$ and $u$ precedes $v$ in $\vl[k]$. For tips and retired vertices, we specify their positions in the total order at the time they are labeled. We write $u \le v$ if $u<v$ or $u=v$. 

\subsubsection{Search Procedure}

\cref{alg:sap_lmp} presents the procedure for finding a shortest augmenting path. It utilizes the following two subroutines: $\setlower(k)$ and $\setupper(k, \kmax)$, detailed in \cref{alg:set_lower,alg:set_upper}. The subroutine $\setlower(k)$ executes $\grow(v)$ on vertices $v$ of level $k$ to find new original lines and assign the lower level $k + 1$. The subroutine $\setupper(k, \kmax)$ assigns upper levels by identifying edges $(u, v)$ such that $\lv(u) = \kmax$ and $\lv(v) = 2k + 1 - \kmax$. For each such edge, it performs $\blossom(u, v)$ if $P(u)$ and $P(v)$ stem from the same singleton; otherwise, it returns $P(u)\overline{P(v)}$ as a shortest augmenting path.
Finally, if no augmenting path is returned and no vertices or edges remain to process, i.e., the condition in Line~\ref{line:while_worth_searching} of \cref{alg:sap_lmp} is no longer met, the search procedure reports that no augmenting path exists.

\begin{algorithm}[tb]
  \caption{Procedure for finding a shortest augmenting path.}
  \label{alg:sap_lmp}

  \begin{algorithmic}[1]
    \Procedure{FindShortestAugmentingPath}{}
    \For{$s \in S$}
    \State $P(s) \gets s$
    \State $L(s) \gets s$
    \State Insert $s$ into $\vl[0]$.
    \EndFor
    \State $k \gets 0$
    \While{$\exists v \in V^*$ s.t. $\lv(v) \ge k$} \label{line:while_worth_searching}
    \State $\setlower(k)$
    \For{$\kmax=k+1, k+2, \dots, 2k+1$}
    \State $P \gets \setupper(k, \kmax)$ \Comment{This may find an augmenting path.}
    \If{$P \neq \emptyset$}
    \State \Return $P$
    \EndIf
    \EndFor
    \State $k \gets k + 1$
    \EndWhile
    \State \Return $\emptyset$ \Comment{No augmenting path has been found.}
    \EndProcedure
  \end{algorithmic}
\end{algorithm}

\begin{algorithm}[tb]
  \caption{Subroutine for assigning lower levels.}
  \label{alg:set_lower}

  \begin{algorithmic}[1]
    \Procedure{\setlower}{$k$}
    \For{$v \in \vl[k]$}
    \State $\grow(v)$
    \EndFor
    \EndProcedure
  \end{algorithmic}
\end{algorithm}

\begin{algorithm}[tb]
  \caption{Subroutine for assigning upper levels.}
  \label{alg:set_upper}

  \begin{algorithmic}[1]
    \Procedure{\setupper}{$k$, $\kmax$}
    \State $\kmin \gets 2k + 1 - \kmax$
    \For{$u \in \vl[\kmax]$}
    \For{$v \in \vl[\kmin]$}
    \If{$(u, v) \in E^*$ and $K(u) \neq K(v)$}
    \If{$P(u)$ and $P(v)$ stem from the same singleton}
    \State $\blossom(u, v)$
    \Else \Comment{Found an augmenting path.}
    \State \Return $P(u) \overline{P(v)}$
    \EndIf
    \EndIf
    \EndFor
    \EndFor
    \State \Return $\emptyset$ \Comment{No augmenting path has been found.}
    \EndProcedure
  \end{algorithmic}
\end{algorithm}

Although $\grow$ internally uses $\blossom$ in a special case,\footnote{The Gabow--Stallmann algorithm handles this case by creating a \textit{degenerate blossom}; this approach was later unified with the handling of ordinary blossoms in \cite{IwataK2022WeightedLinear}.} we first present $\grow$ in \cref{alg:grow} because of its simplicity. Line~\ref{line:degblm} addresses the special case where a blossom consisting of $\qty{u, \bar{u}}$ and a dummy line should be formed, which is handled instantly during the $\grow$ operation.

\begin{algorithm}[tb]
  \caption{Operation for identifying new original lines.}
  \label{alg:grow}

  \begin{algorithmic}[1]
    \Procedure{\grow}{$v$}
    \For{$u \in V^*$ s.t. $(u, v) \in E^*$}
    \If{neither $u$ nor $\bar{u}$ is labeled}
    \State $P(\bar{u}) \gets P(v) u\bar{u}$.
    \State Insert $\bar{u}$ at the end of $\vl[\lv(\bar{u})]$.
    \State $\rho(u) \gets v$
    \If{$(\bar{u}, v) \in E^*$}
    \State $\blossom(\bar{u}, v)$ \label{line:degblm}
    \Comment{This overwrites $\rho(u)$ with the tip of a new blossom.}
    \State $L(\bar{u}) \gets \rho(u)u\bar{u}$
    \Else
    \State $L(\bar{u}) \gets \bar{u}$
    \EndIf
    \EndIf
    \EndFor
    \EndProcedure
  \end{algorithmic}
\end{algorithm}

We then detail the $\blossom$ operation. The precondition for $\blossom(u, v)$ is $u > v$, which always holds whenever the operation is called. While $\grow$ can be simply described in pseudocode, $\blossom$ involves the following complex steps. See \cref{fig:blossom_pivoting} for the pivoting operation performed in \ref{item:detail_3} and \cref{fig:partition,fig:c_eq_d,fig:c_neq_d} for the subsequent updates to $K(x)$ executed in \ref{item:detail_6}.

\begin{figure}[tb]
    \centering
    \includegraphics[width=0.85\linewidth]{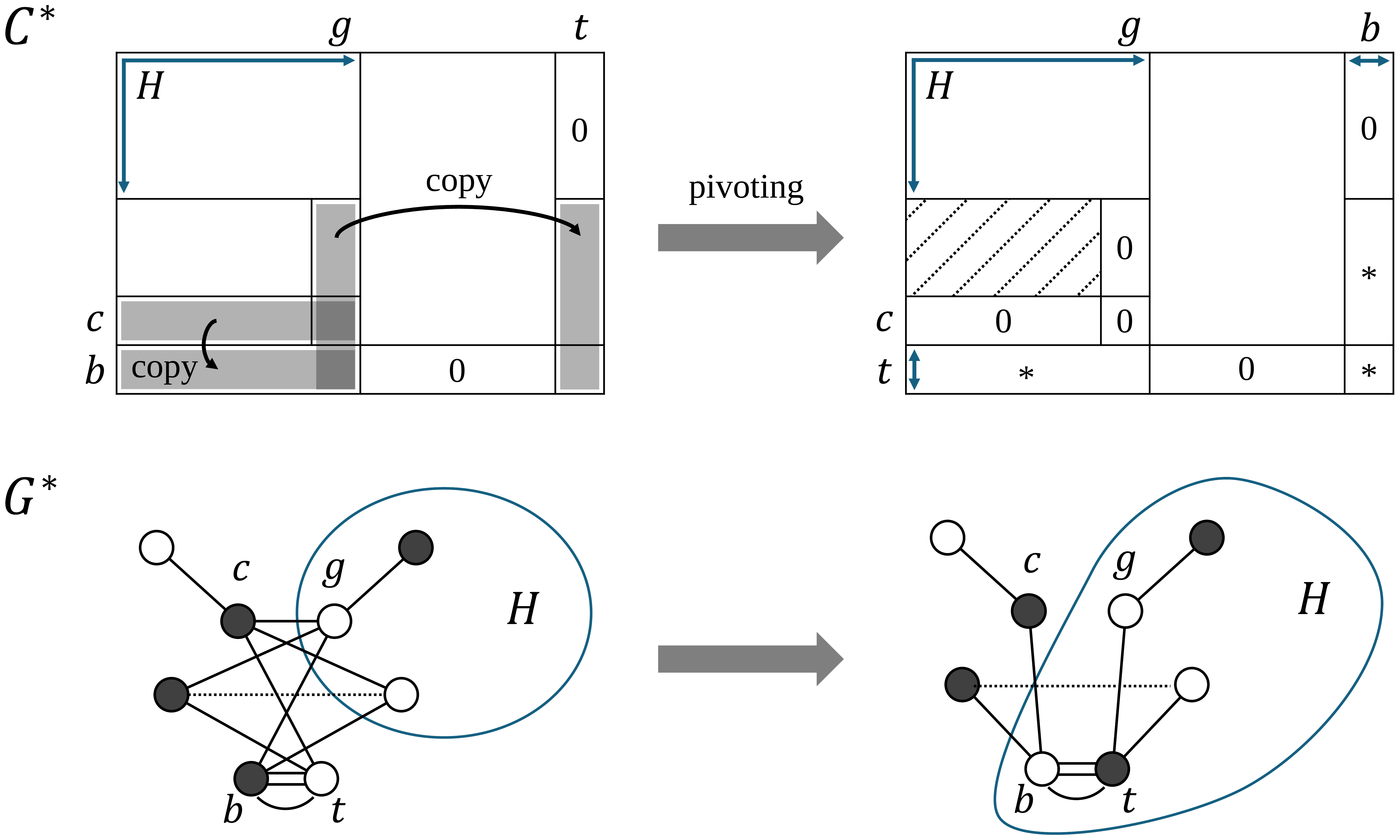}
    \caption{Illustration of \ref{item:detail_3} when $c \in B^*$ and $g \notin B^*$. Black and white vertices represent those in $B^*$ and $V^* \setminus B^*$, respectively. The pivoting operation alters the hatched region in $C^*$; specifically, \cref{prop:pivoting_new_edge} ensures that edges $(x, y) \in H \times (V^* \setminus H)$ with $(x, c), (g, y) \in E^*$ may be added or deleted, as illustrated by the dotted edge. It also scales each region marked with an asterisk in $C^*$, preserving the corresponding edges in $G^*$ (see \cref{prop:neighbor_of_bt}).}
    \label{fig:blossom_pivoting}
\end{figure}

\begin{enumerate}[label=(Step \arabic*), leftmargin=*, itemsep=4.0pt plus 2.0pt minus 1.0pt, parsep=2.0pt]
  \item \label{item:detail_2}
        Let $c$ be the last vertex in $P(u)$ such that $K(c)$ shares a vertex with $P(v)$. Let $d$ be the last vertex in $P(v)$ such that $d \in K(c)$. Note that $K(c) = K(d)$ holds. If $c = d$, set $H \gets \bigcup \qty{K(x) \mid x \in P(u|c) \cup P(v|d)}$, and proceed to \ref{item:detail_3}. If $c \neq d$, set $H \gets \bigcup \qty{K(x) \mid x \in P(u|c) \cup P(v|d) \cup \qty{c}}$, and skip to \ref{item:detail_5}.

  \item \label{item:detail_3}
        Let $g$ be the successor of $c$ on $P(u)$.\footnote{Our algorithm always calls $\blossom(u, v)$ with a vertex $u$ such that $P(u|c)$ is nonempty, ensuring that $g$ can be defined. In contrast, $P(v|d)$ may become empty; for example, it is always empty when the $\blossom$ operation is triggered within a $\grow$ operation.} Mark $g$ as retired. Extend the fundamental cocircuit matrix $C^*$ by inserting $\qty{b, t}$ as a new row and column, with entries defined as follows.
        \begin{itemize}
          \item If $c \in B^*$ and $g \notin B^*$, let $b$ be a new row and $t$ be a new column. Define $C^*_{bt} \coloneqq C^*_{cg}$. For $y \in V^* \setminus B^*$, define $C^*_{by} \coloneqq C^*_{cy}$ if $y \in H$ and $C^*_{by} \coloneqq 0$ otherwise. For $x \in B^*$, define $C^*_{xt} \coloneqq C^*_{xg}$ if $x \notin H$ and $C^*_{xt} \coloneqq 0$ otherwise. 
          \item If $c \notin B^*$ and $g \in B^*$, let $b$ be a new column and $t$ be a new row. Define $C^*_{tb} \coloneqq C^*_{gc}$. For $x \in B^*$, define $C^*_{xb} \coloneqq C^*_{xc}$ if $x \in H$ and $C^*_{xb} \coloneqq 0$ otherwise. For $y \in V^* \setminus B^*$, define $C^*_{ty} \coloneqq C^*_{gy}$ if $y \notin H$ and $C^*_{ty} \coloneqq 0$ otherwise. 
        \end{itemize}
        Expand $B^*$ and $V^*$ according to the index set of $C^*$. Finally, update $B^* \gets B^* \symdiff \qty{b, t}$ by pivoting around $\qty{b, t}$, modify $G^*$ accordingly, and set $H \gets H \cup \qty{b, t}$.

  \item \label{item:detail_4}
        Update every search path $P(x)$ intersecting $H$ by setting $P(x) \coloneqq P(c)btP(x|c)$. Label $t$ with $P(t) \coloneqq P(c)bt$ and define $L(t) \coloneqq t$.
        Note that $\lv(t) = \lv(c)$. Extend the total order $\le$ so that $t$ is immediately after $c$. Set $\rho(x) \coloneqq t$ for each vertex $x \in H$ with $\rho(x) = c$, and set $\rho(b) \coloneqq c$. 

  \item \label{item:detail_5}
        If $c = d = v$, label $\bar{u}$ with $P(\bar{u}) \coloneqq P(v) bt u \bar{u}$ and $L(\bar{u}) \coloneqq t u \bar{u}$, place $\bar{u}$ as a successor of $u$ in the total order $\le$ since $\bar{u}$ is retired, and proceed to \ref{item:detail_6}. Otherwise, label each unlabeled vertex $x \in P(u|c) \cup P(v|d)$, excluding the new bud $b$ if $c = d$, in descending order of $\rho(x)$ as follows.
        \begin{enumerate}[label=(\arabic*), nosep]
            \item Set $P(x) \coloneqq P(v)\overline{P(u|x)}x$ if $x \in P(u|c)$, and $P(x) \coloneqq P(u)\overline{P(v|x)}x$ if $x \in P(v|d)$.
            \item Unless $x$ is retired, append $x$ to the end of $\vl[\lv(x)]$. 
        \end{enumerate}
        After labeling all such $x$, define $L(x) \coloneqq P(x|y)$ for each $x$, where $y$ is the last vertex in $P(x)$ that is unlabeled or satisfies $y > x$. When $c = d$, the retired vertex $g$ and a potential tie are handled as follows (no ties occur when $c \neq d$). Since $g$ is excluded from $\vl$, its position within the total order $\le$ is determined according to this labeling order. Here, a tie occurs if and only if there exists an unlabeled vertex $h \neq g$ such that $\rho(h) = \rho(g) = t$; in this case, break the tie in the total order $\le$ by setting $g < h$, and define $L(g) \coloneqq P(g|b)$.

  \item \label{item:detail_6}
        If $c = d$, mark $b$ and $t$ as the bud and tip of $H$, respectively. Otherwise, set the bud and tip of $H$ to those of $K(c)$. Finally, register $H$ as a blossom and update $K(x) = H$ for each $x \in H$.
\end{enumerate} 

We defer algorithmic properties to \cref{sec:basic_properties}. Here, we describe the properties of the pivoting operation.
\begin{lemma}
  \label{prop:bt}
  The following properties hold after each $\blossom$ operation.
  \begin{enumerate}[label=(BT\arabic*),leftmargin=*]
    \item Let $C'$ be the matrix obtained by applying the pivoting operation around $T$ to the fundamental cocircuit matrix $C^*$. Then, $C'[V \cup S] = C$ holds.
    \item For any blossom $H$, tips and retired vertices within $H$ are not adjacent to any vertex in $V^* \setminus H$.
  \end{enumerate}
\end{lemma}
\begin{proof}
  Property (BT1) is straightforward by the algorithm description. It suffices to verify (BT2) when $\blossom$ creates a dummy line $\qty{b, t}$, since the statement inductively follows if the dummy line is reused.
  
  Prior to the pivoting operation, $b$ is not adjacent to any vertex in $V^* \setminus H$. Hence, the update formula in Eq.~\eqref{eq:pivoting} implies that $t$ is not adjacent to any vertex in $V^* \setminus H$.

  It remains to prove that $g$ is not adjacent to any vertex in $V^* \setminus H$. We assume $c \in B^*$ and $g \notin B^*$ because the other case follows by a similar argument. Let $C^*$ and $\widehat{C}^*$ be the fundamental cocircuit matrix before and after the pivoting operation. The update formula in Eq.~\eqref{eq:pivoting} yields the equation $\widehat{C}^*_{xy} = C^*_{xy} - C^*_{by} C^*_{xt} / C^*_{bt}$ for $x \in B^*$ and $y \in V^* \setminus B^*$. Hence, for $x \in B^* \setminus H$ and $y = g$, we have $\widehat{C}^*_{xg} = C^*_{xg} - C^*_{bg} C^*_{xt} / C^*_{bt} = 0$ by $C^*_{xg} = C^*_{xt}$ and $C^*_{bg} = C^*_{bt}$, which completes the proof.
\end{proof}
\begin{lemma}
  \label{prop:pivoting_new_edge}
  Let $E^*$ and $\widehat{E}^*$ be the edge sets before and after the pivoting operation performed by a $\blossom$ operation.
  Any edge $(x, y) \in E^* \symdiff \widehat{E}^*$ satisfies (i) $\qty{x, y} \cap \qty{b, t} \neq \emptyset$, or (ii) $x \in H$, $y \notin H$, and $(x, c), (g, y) \in E^*$, up to interchanging $x$ and $y$.
\end{lemma}
\begin{proof}
  Assuming that condition~(i) does not hold, we show that condition~(ii) follows. We assume $c \in B^*$ and $g \notin B^*$ since the other case follows similarly. The update formula of the entry $\widehat{C}^*_{xy} = C^*_{xy} - C^*_{by} C^*_{xt} / C^*_{bt}$ presented in the proof of \cref{prop:bt} implies that $C^*_{by} \neq 0$ and $C^*_{xt} \neq 0$ is necessary. Hence, the definition of these entries concludes the proof.
\end{proof}
\begin{lemma}
    \label{prop:neighbor_of_bt}
    Suppose a $\blossom$ operation inserts a dummy line $\qty{b, t}$. Then, if a vertex $v$ that existed before this operation (i.e., $v \notin \qty{b, t}$) is adjacent to $b$ (resp.~$t$), it was adjacent to $g$ (resp.~$c$) before the operation.
\end{lemma}
\begin{proof}
    It follows from the description of \ref{item:detail_3} and the pivoting formula in Eq.~\eqref{eq:pivoting}.
\end{proof}

\subsection{Super Singleton Reduction}
\label{sec:super_singleton_reduction}

To simplify the subsequent proofs, we introduce a reduction from the multi-source to a single-source setting. The reduction consolidates the set of singletons $S = \qty{s_1, \dots, s_{\abs{S}}}$ into a new singleton $s_0$ by extending the graph as in \cref{fig:super_singleton_reduction}. This is achieved by extending the fundamental cocircuit matrix as follows:
\begin{equation}
  \label{eq:super_singleton_c}
  \NiceMatrixOptions{code-for-first-row = {}, code-for-first-col = {}, cell-space-limits = 2pt}
  C^* \gets \quad \begin{bNiceArray}{c|c|c}[first-col,first-row,margin]
    & V \setminus J & \bar{q} & q \\[-2pt]
    \rule{0pt}{1.3em} J & \Block{2-1}{\raisebox{-2pt}{$C^*$}} & & \\
    \rule{0pt}{1.3em} S & & & \\ \hline
    \rule{0pt}{1.3em} \bar{S} &  & \raisebox{1pt}{$\bm{1}$} & \\ \hline
    s_0 &  &  & \raisebox{-1pt}{$1$} \\
  \end{bNiceArray},
\end{equation}
where all blank blocks are zero and $\bm{1}$ denotes an all-ones vector of the corresponding size. Here, we introduce two new columns $q$ and $\bar{q}$ that form a line. We also added a set $\bar{S} \coloneqq \qty{\bar{s}_1, \dots, \bar{s}_\abs{S}}$ as new rows, where $\qty{s_i, \bar{s_i}}$ forms a line for each $i = 1, \dots, \abs{S}$. Finally, a new row $s_0$ is inserted as a unique singleton, which we call a \textit{super singleton}. The base $B^*$ is also updated as $B^* \gets B^* \cup \bar{S} \cup \qty{s_0}$.

\begin{figure}
  \centering
  \includegraphics[width=0.55\linewidth]{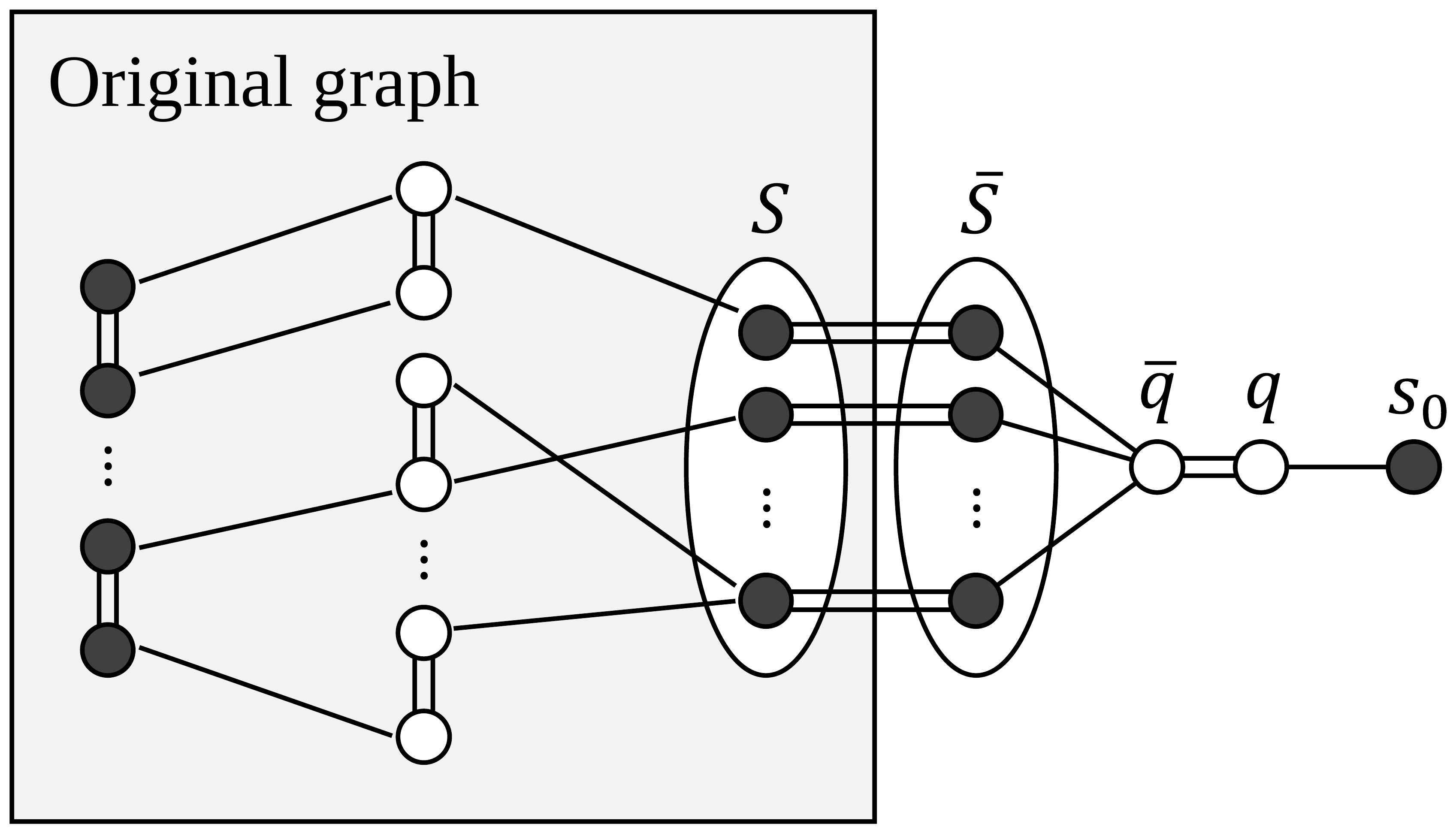}
  \caption{Graph extension via super singleton reduction. Black vertices belong to the new base.}
  \label{fig:super_singleton_reduction}
\end{figure}

While this reduction preserves the structure of search paths and blossoms, it reduces finding an augmenting path to forming a blossom that contains any two of the original singletons. This allows for a unified treatment of blossom formation and augmenting path detection in the subsequent proofs.

\subsection{Basic Properties}
\label{sec:basic_properties}

Here, we outline the fundamental properties of the algorithm to elucidate its mechanism and to facilitate the subsequent proofs. Unless otherwise specified, the invariant conditions discussed in the subsequent proofs are preserved by each $\grow$ or $\blossom$ operation, excluding nested $\blossom$ calls within $\grow$.

We first define the search tree $\st$ as follows. The vertex set $V(\st)$ consists of labeled vertices and their mates. The edge set $E(\st)$ comprises all edges of the forms $(x, \rho(x))$ and $(x, \bar{x})$ for every upper vertex $x$.

While the search tree $\st$ provides a fundamental framework, it does not fully characterize the structure of the search paths and related concepts. To bridge this gap, we incorporate \textit{blossom-forming edges}, defined as the edges to which $\blossom$ has been applied. As an exception, whenever the $\blossom$ operation is called within a $\grow$ operation, the corresponding blossom-forming edge is defined as $(u, t)$. The following lemma characterizes the structure of the search paths using these concepts.
\begin{lemma}
  \label{prop:path_edge_taxonomy}
  Any two consecutive vertices in a search path constitute either a search-tree edge or a blossom-forming edge.
\end{lemma}
\begin{proof}
  The statement holds inductively when the search path is defined or updated by $\grow$ or $\blossom$. Hence, the only potentially problematic case is the deletion of edges upon the introduction of dummy lines. However, since a removed edge $(x, y)$ must satisfy condition~(ii) from \cref{prop:pivoting_new_edge}, every search path $P(v)$ including such an edge intersects $H$. Consequently, \ref{item:detail_4} must have updated this path to $P(c)btP(v|c)$, preserving the condition.
\end{proof}

Based on these concepts, we characterize the structure of the maximal blossoms $K(x)$. A simple induction verifies that blossoms are a union of lines. Moreover, maximal blossoms constitute a partition of the search tree $\st$ except for the root $s_0$. The following lemma characterizes the maximal blossoms via the search tree and the search paths (see \cref{fig:partition,fig:c_eq_d,fig:c_neq_d}).

\begin{figure}[htbp]
  \centering
  \includegraphics[width=0.64\textwidth]{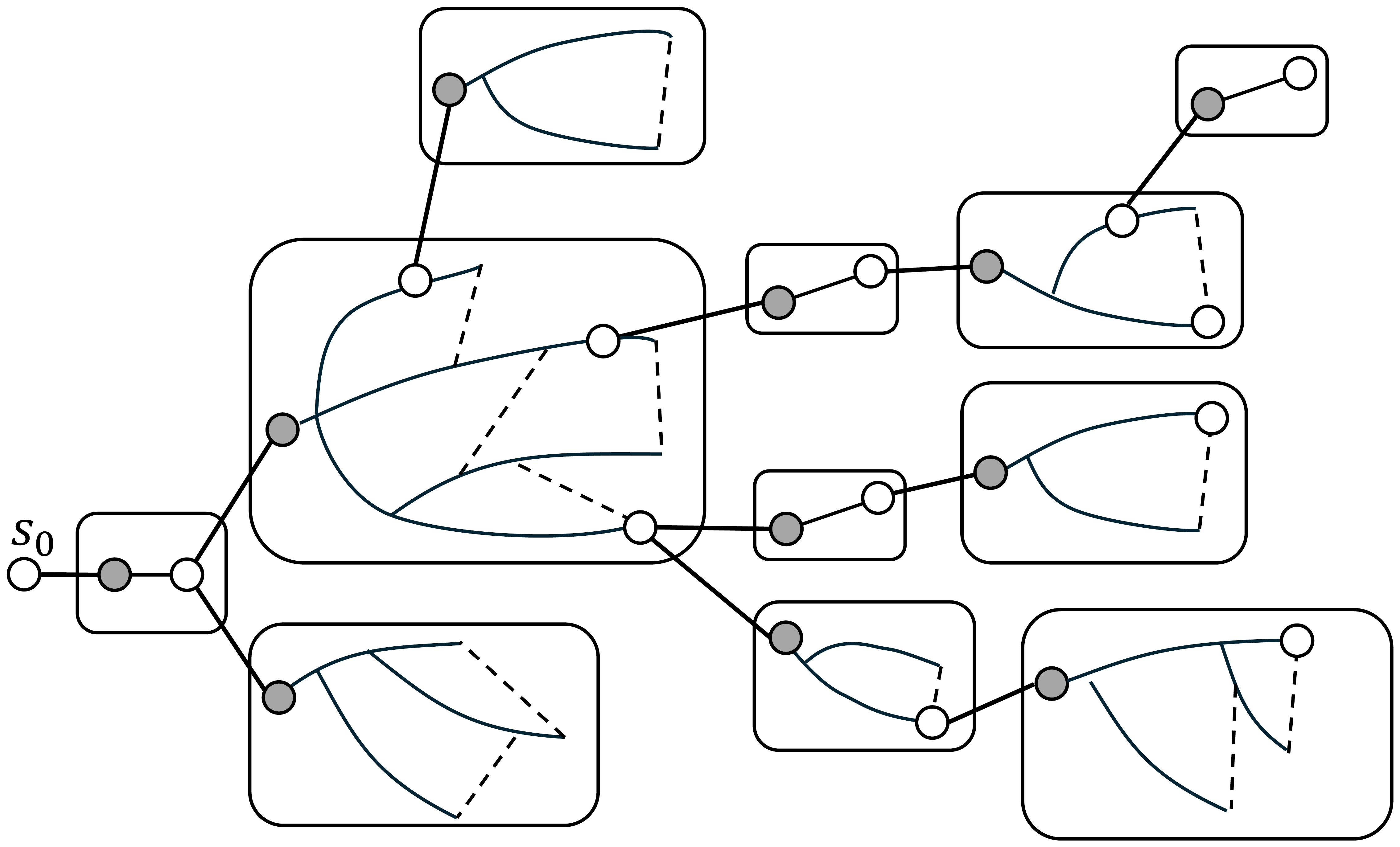}
  \caption{Partition of $V(\st) \setminus \qty{s_0}$ given by $K(x)$. Gray nodes represent the unique unlabeled vertex in each component. Solid lines and curves represent single and multiple search-tree edges, respectively; dashed lines represent single blossom-forming edges.}
  \label{fig:partition}
\end{figure}

\begin{figure}[htbp]
  \centering
  \includegraphics[width=0.64\textwidth]{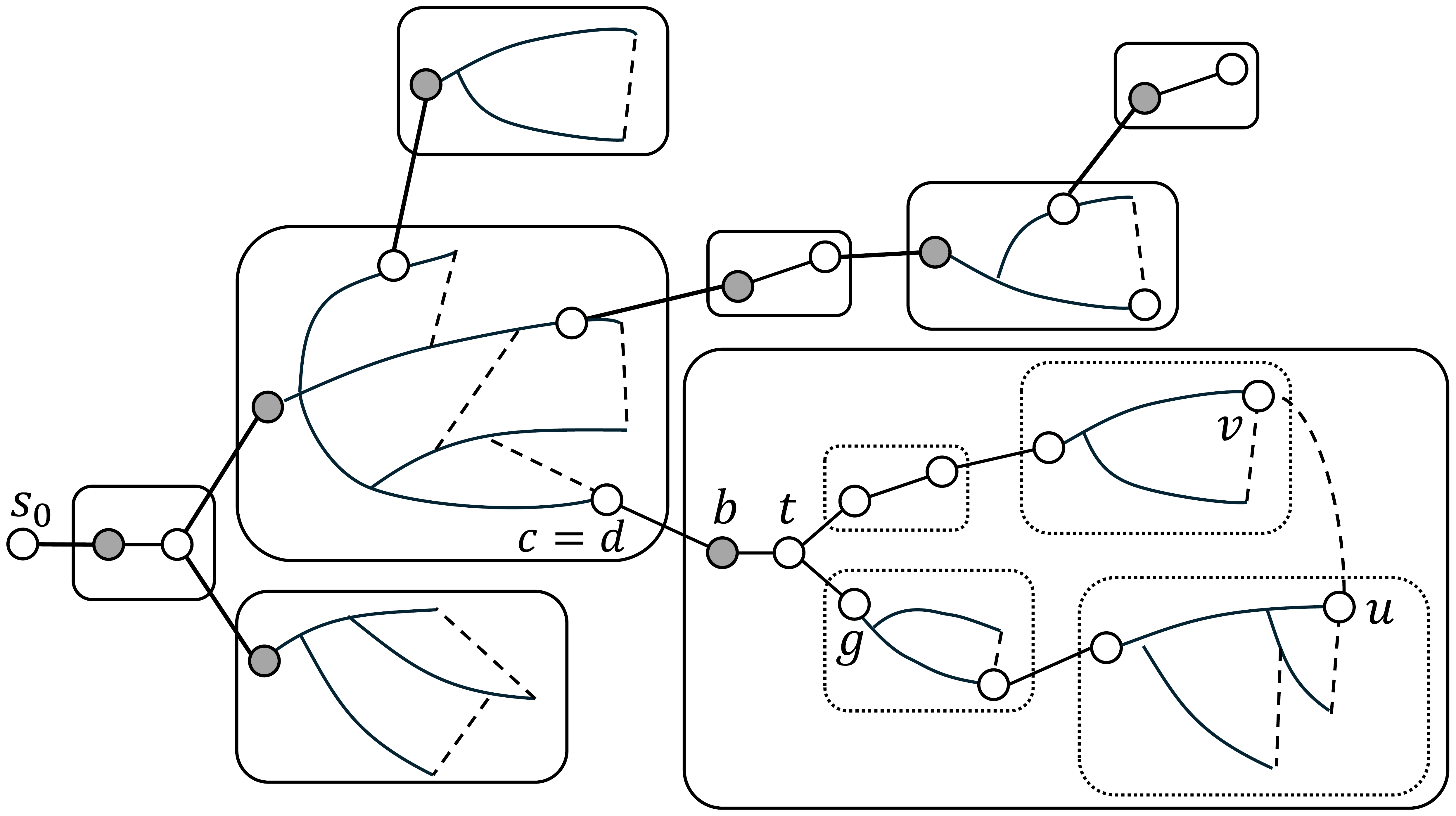}
  \caption{Update of the partition by a $\blossom$ operation in the case where $c = d$.}
  \label{fig:c_eq_d}
\end{figure}

\begin{figure}[htbp]
  \centering
  \includegraphics[width=0.64\textwidth]{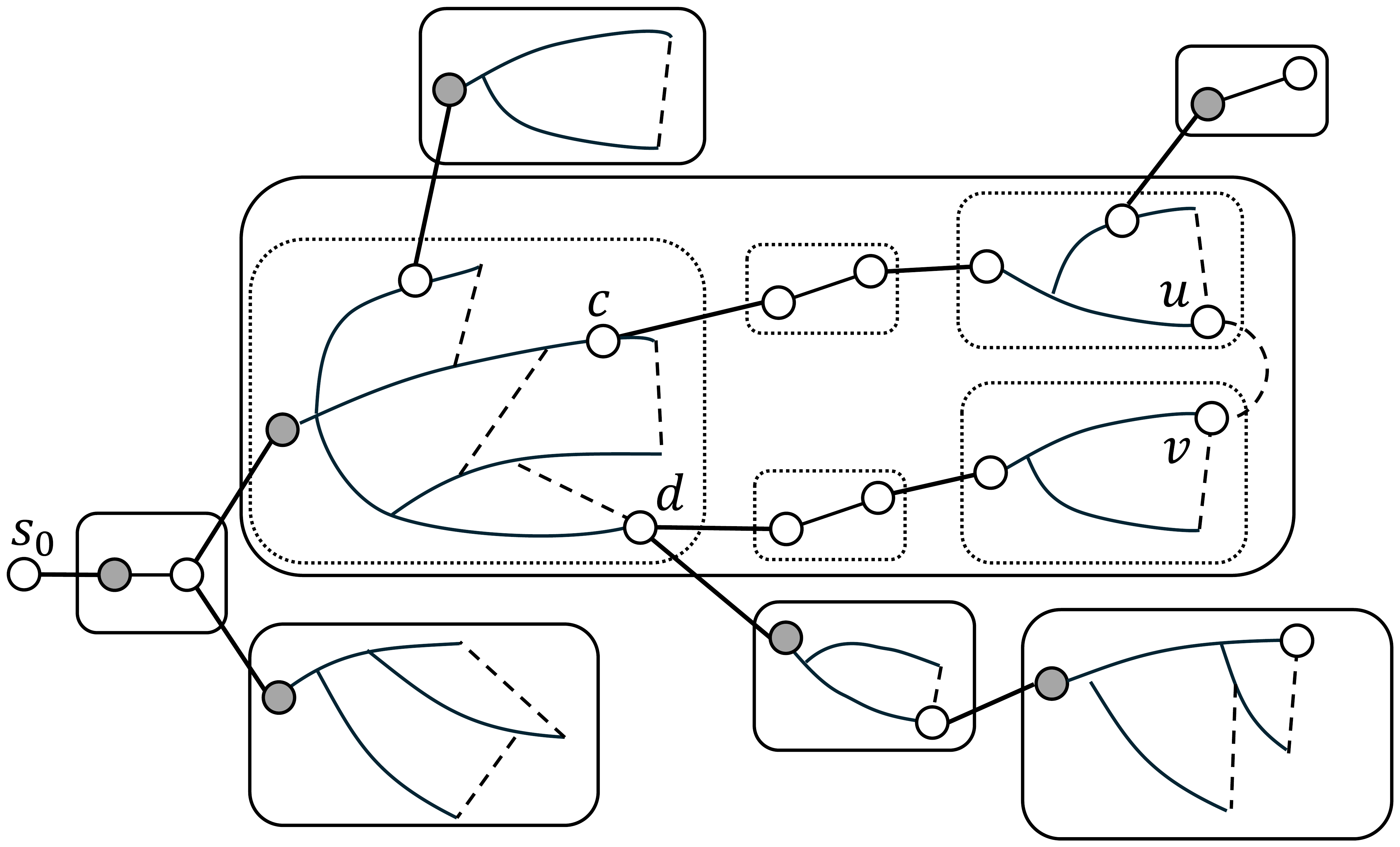}
  \caption{Update of the partition by a $\blossom$ operation in the case where $c \neq d$.}
  \label{fig:c_neq_d}
\end{figure}

\begin{lemma}
  \label{prop:kx_properties}
  For every vertex $x \in V(\st) \setminus \qty{s_0}$, the maximal blossom $K(x)$ satisfies the following properties.
  \begin{enumerate}[label=(MB\arabic*), leftmargin=*]
    \item $K(x)$ contains exactly one unlabeled vertex.
    \item Any search path $P(v)$ that enters $K(x)$ must do so via its unique unlabeled vertex.
    \item The induced subgraph $\st[K(x)]$ is a tree.
  \end{enumerate}
\end{lemma}

To prove this lemma, we also need to show that every search path is simple, i.e., that it does not contain any repeated vertices. Such basic properties of search paths are summarized below.
\begin{lemma}
  \label{prop:path_property}
  Let $P(x) = x_1 x_2 \dots x_l$ be a search path, where $x_1 = s_0$ and $x_l = x$. Then, $P(x)$ is simple, i.e., $x_i \neq x_j$ if $i \neq j$. Also, the following properties hold:
  \begin{enumerate*}[label=(\arabic*)]
    \item $l$ is odd,
    \item $x_{2i - 1}$ and $x_{2i}$ form an edge in $E^*$ for each $i$,
    \item $x_{2i}$ and $x_{2i + 1}$ form a line for each $i$, and
    \item if $x_j$ is unlabeled, then $j$ is even.
  \end{enumerate*}
\end{lemma}

We prove \cref{prop:kx_properties,prop:path_property} simultaneously by induction, since each lemma is necessary to establish the other in the inductive step. Before presenting the proof, we provide a core observation based on these lemmas that simplifies the identification of $c$ in \ref{item:detail_2}.
\begin{remark}
    \label{rem:kc_is_lca}
    By (MB3), the contraction of each maximal blossom $K(x)$ on $\st$ yields a tree. Furthermore, because search paths are simple and blossom-forming edges cannot bridge different maximal blossoms, \cref{prop:path_edge_taxonomy} ensures that each search path $P(v)$ is mapped to the unique path from $s_0$ to $K(v)$ in this contracted tree. Therefore, for a $\blossom(u, v)$ operation, the blossom $K(c) = K(d)$ defined in \ref{item:detail_2} corresponds to the lowest common ancestor of $K(u)$ and $K(v)$ in this contracted tree.
\end{remark}

\begin{proof}[Proof of \cref{prop:kx_properties,prop:path_property}.]
  We proceed by induction on the $\grow$ and $\blossom$ operations. It is straightforward to verify all properties in these lemmas when $\grow$ adds a new original line to $\st$. Hereafter, we focus on a $\blossom$ operation on $u$ and $v$. Note that \cref{rem:kc_is_lca} holds inductively prior to the operation.

  First, suppose $c = d$. We first focus on the update of search paths in \ref{item:detail_4}. In this case, \cref{rem:kc_is_lca} ensures that the set $H$ in \ref{item:detail_2} is the union of the maximal blossoms on the paths from $K(u)$ and $K(v)$ to $K(c)$, excluding $K(c)$. Consequently, every search path $P(x)$ intersecting $H$ must enter $H$ by exiting $K(c)$ at $c$ via a search-tree edge, which implies $c \in P(x)$. Hence, the updated path $P(c) bt P(x|c)$ is well-defined. Moreover, this path remains simple because $P(c)$ and $P(x|c)$ are disjoint; specifically, in the contracted tree, $P(c)$ is mapped to the path from $s_0$ to $K(c)$, whereas $P(x|c)$ is mapped to the subsequent path from a child component of $K(c)$ to $K(x)$.

  Still under the assumption $c = d$, we now show that (MB1) and (MB2) are inductively preserved and that the search paths remain simple. For the blossoms comprising $H$, their unique unlabeled vertices lie on $P(u|c)$ or $P(v|d)$ by (MB2), and thus become labeled in \ref{item:detail_5}. Consequently, the resulting blossom $H \cup \qty{b, t}$ has $b$ as its unique unlabeled vertex, satisfying (MB1). Furthermore, the updated paths in \ref{item:detail_4} enter the resulting blossom $H \cup \qty{b, t}$ through $b$, maintaining (MB2). The search paths newly defined in \ref{item:detail_5} also preserve (MB2), because their extended segments remain within the resulting blossom. Finally, these paths are simple, because the contraction structure guarantees that $P(c)$, $P(u|c)$, and $P(v|d)$ are disjoint before their updates.
  
  Next, suppose $c \neq d$. The set $H$ now includes $K(c)$, and a similar argument shows that the unique unlabeled vertex of $K(c)$ is the only unlabeled vertex within $H$, satisfying (MB1). Property (MB2) holds similarly, because every search path entering the new blossom $H$ must visit $K(c)$ before any other blossom comprising $H$. Finally, the search paths newly defined in \ref{item:detail_5} are simple, because $P(u)$ and $P(v|d)$ are disjoint, as are $P(v)$ and $P(u|c)$.
  
  We now verify the remaining conditions: properties (1)--(4) from \cref{prop:path_property} and (MB3). Properties (1)--(4) from \cref{prop:path_property} follow by an analogous inductive argument. Regarding (MB3), existing maximal blossoms preserve (MB3) because their internal edges cannot be deleted. The new blossom also satisfies (MB3) because search-tree edges link the components in $H$ (and $\qty{b, t}$ if $c = d$) into a tree. This completes the proof.
\end{proof}

The contraction structure discussed in this proof has several implications for the description of the $\blossom$ operation. First, the definitions of $c$ and $d$ in \ref{item:detail_2} are symmetric in the sense that $d$ satisfies the symmetric condition, i.e., $d$ is the last vertex in $P(v)$ such that $K(d) \cap P(u) \neq  \emptyset$. Furthermore, the algorithm only designates previously unlabeled vertices as $g$, which ensures that active vertices never become retired once they are labeled. Also, the characterization of possible ties described in \ref{item:detail_5} can be verified.

Next, we will prove \cref{cor:justify_line_4}, which ensures that the update of each search path intersecting $H$ in \ref{item:detail_4} corresponds to inserting $bt$ into $P(x)$. This implies the invariance of levels under the insertion of dummy lines. To prove this corollary, we utilize the following lemma.
\begin{lemma}
  \label{lem:search_path_prefix}
  Every unlabeled vertex $y$ in a search path $P(x)$ is an upper vertex, i.e., $\rho(y)$ is defined. Moreover, $P(x)$ decomposes as $P(\rho(y)) y \bar{y} P(x | \bar{y})$.
\end{lemma}
\begin{proof}
  Since $y$ is unlabeled, $\qty{y, \bar{y}}$ is either found by $\grow$ or created as a dummy line by $\blossom$. In either case, $y$ is an upper vertex and $\rho(y)$ is defined. Furthermore, $P(\bar{y}) = P(\rho(y)) y \bar{y}$ holds.

  We proceed by induction to show that any search path containing $y$ has $P(\bar{y})$ as its prefix as long as $y$ is unlabeled. If a new search path containing $y$ is generated by $\grow$, this condition is trivially satisfied. Suppose next that $\blossom(u, v)$ creates a search path $P(x)$ containing $y$ as an unlabeled vertex. If $y$ is the bud created by this operation, the statement follows immediately by definition. Otherwise, since $y$ existed prior to the operation, $y$ cannot lie on $P(u|c)$ or $P(v|d)$. Consequently, $y$ lies on the unreversed subpath in the definition of $P(x)$, namely $P(u)$ if $P(x) \coloneqq P(u) \overline{P(v|x)} x$, or $P(v)$ if $P(x) \coloneqq P(v) \overline{P(u|x)} x$. Therefore, $P(\bar{y})$ is a prefix of $P(x)$ by the inductive assumption.

  Finally, suppose that $P(\bar{y})$ is a prefix of $P(x)$, and that inserting a dummy line modifies at least one of these paths. The possible cases are that both paths are updated, or that only $P(x)$ is updated. In the former case, $P(\bar{y})$ and $P(x)$ are updated to $P(c) bt P(\bar{y}|c)$ and $P(c)btP(x|c)$, respectively, preserving the condition. In the latter case, $y$ must precede or equal $c$ in $P(x)$. Hence, if $y \notin P(c)$, the updated path $P(x) = P(c) bt P(x|c)$ no longer contains $y$. If $y \in P(c)$, then $P(\bar{y})$ is a prefix of $P(c)$ by the inductive assumption, and consequently of $P(x)$ after the update.
\end{proof}
\begin{corollary}
  \label{cor:justify_line_4}
  If a search path $P(x)$ intersects $H$ in \ref{item:detail_4}, then $c \in P(x)$ and $P(x) = P(c) P(x|c)$.
\end{corollary}
\begin{proof}
  Given that $P(u)$ enters $H$ through the first vertex $u'$ of $P(u|c)$, \cref{prop:kx_properties}~(MB2) ensures that $u'$ is the unlabeled vertex of $K(u')$. Similarly, the first vertex $v'$ of $P(v|d)$, if it exists, is the unlabeled vertex of $K(v')$. Hence, applying \cref{lem:search_path_prefix} to $u'$ and $v'$ shows that all search paths intersecting $H$ must contain $P(\bar{u}')$ or $P(\bar{v}')$ as their prefix. Since these paths further contain $P(c)$ as their common prefix, the statement follows.
\end{proof}

Finally, we present a structural lemma of the paths identified by the algorithm, which will be crucial for establishing the connection between the nonsingularity of $C[P \setminus T]$ and that of $C^*[P]$.
\begin{definition}
    \label{def:coherence}
    A path $P \subseteq V^*$ is called \textit{coherent} if $P \cap H \neq \emptyset$ implies $P \cap \qty{b, t} \neq \emptyset$ for every dummy line $\qty{b, t}$ introduced by the blossom $H$.
\end{definition}
\begin{lemma}
    \label{lem:path_coherence}
    Every search path $P(v)$ and the path $P$ found by the search procedure are coherent.
\end{lemma}
\begin{proof}
    Search paths newly defined by $\grow$ inductively satisfy this condition, and those generated by $\blossom$ also satisfy it by \cref{prop:kx_properties,rem:kc_is_lca}. When \ref{item:detail_4} introduces a dummy line and updates search paths, the condition holds for the new dummy line by the definition of the path updates and remains valid for existing dummy lines by \cref{cor:justify_line_4}. Finally, the path $P$ found by the search procedure is coherent, since its constituent search paths are coherent.
\end{proof}

\subsection{Validity}
\label{sec:lmp_validity}

\subsubsection{Invariance of Nonsingularity}

This section is devoted to proving that, for the path $P$ detected by the algorithm, $P \setminus T$ is indeed an augmenting path for $B$. To this end, we first characterize the nonsingularity of $C[P]$ via the coherence of $P$ and the nonsingularity of $C^*[P]$ as follows.
\begin{lemma}
    \label{lem:nonsingular_invariance}
    Suppose that a coherent path $P \subseteq V^*$ either includes or is disjoint from each dummy line. Then, the following are equivalent:
    \begin{enumerate*}[label=(\roman*)]
        \item $C^*[P]$ is nonsingular; 
        \item $C^*[P \cup T]$ is nonsingular; and
        \item $C[P \setminus T]$ is nonsingular.
    \end{enumerate*}
\end{lemma}
\begin{proof}
    We first demonstrate the equivalence between (i) and (ii). Let $P$ be a coherent path such that $C^*[P]$ is nonsingular, and consider adding the remaining dummy lines $T \setminus P$. Since $P$ is coherent, the blossoms that inserted these dummy lines are disjoint from $P$. Among these blossoms, let $H$ be one of the outermost blossoms that creates a dummy line $\qty{b, t}$. Because $t$ is adjacent only to vertices within $H$, its neighbor within $G^*[P \cup \qty{b, t}]$ is restricted to $b$. Consequently, it follows that $\det C^*[P \cup \qty{b, t}] = C^*[\qty{b, t}] \cdot \det C^*[P]$, implying that $C^*[P]$ is nonsingular if and only if $C^*[P \cup \qty{b, t}]$ is nonsingular. Since $P \cup \qty{b, t}$ remains coherent by the choice of the outermost blossom, applying this argument inductively yields the equivalence between (i) and (ii).

    We now prove the equivalence between (ii) and (iii). Observe that $C^*[P \cup T]$ is nonsingular if and only if $B^* \symdiff (P \cup T)$ forms a base of the extended linear matroid $\bfm(A^*)$ with $A^*=\left[I \mid C^*\right]$. Since $B^* \symdiff (P \cup T) = (B^* \symdiff T) \symdiff (P \setminus T)$, this condition is equivalent to the nonsingularity of $C'[P \setminus T]$, where $C'$ is obtained from $C^*$ by pivoting around $T$. By (BT1), this submatrix coincides with $C[P \setminus T]$, which completes the proof.
\end{proof}

By \cref{lem:path_coherence,lem:nonsingular_invariance}, we can restrict our attention to proving that the final output path $P$ yields a nonsingular submatrix $C^*[P]$.

\subsubsection{The Search Tree}

Given \cref{lem:path_coherence,lem:nonsingular_invariance}, we only need to show that $C^*[P]$ is nonsingular for the final output path $P$. In the remainder of this section, we establish this nonsingularity by showing that the induced subgraph $G^*[P]$ has a unique perfect matching. We assume the reduction introduced in \cref{sec:super_singleton_reduction} unless otherwise specified.

To facilitate the subsequent proofs, we first introduce the notation for subtrees of $\st$. Taking $s_0$ as the root of $\st$, let $\st(v)$ denote the subtree containing $v$ and all its descendants. When the context is clear, we also use $\st(v)$ to refer to the vertex set of this subtree. 

Following \cite{IzumiKY2025ForgettingAlternation}, we define an edge $(u, v) \in E^*$ as an \textit{incoming edge} of $\st(x)$ if $(u, v) \notin E(\st)$, $u \in V(\st) \setminus \st(x)$, and $v \in \st(x)$. Crucially, the proof of \cref{prop:kx_properties} shows that the condition $u \in \st(v)$ for $u, v \in V(\st)$ is invariant under modifications of the search tree. Hence, it is unnecessary to consider the possibility of an existing edge becoming or ceasing to be an incoming edge.

Based on these definitions, we now discuss the properties of upper vertices in the search tree. Specifically, \cref{prop:incoming_edge_label} establishes a sufficient condition for an upper vertex to be labeled.
\begin{lemma}
  \label{lem:subtree_rho}
  For any upper vertices $x$ and $y \in \st(\bar{x})$, we have $\rho(x) < \rho(y)$. Consequently, we also have $\rho(x) \le \rho(z)$ for any upper vertex $z \in \st(x)$. 
\end{lemma}
\begin{proof}
  The definition of $\rho$ implies that $\rho(v) < v$ and $\rho(v) < \bar{v}$ for any upper vertex $v$. Since $x$ is an ancestor of $y$ on $\st$, repeated application of this property yields $\rho(x) < \rho(y)$. The latter statement follows from $\st(x) = \st(\bar{x}) \cup \qty{x}$.
\end{proof}

\begin{lemma}
  \label{prop:incoming_edge_label}
  An upper vertex $x$ is labeled if there exists an incoming edge $(u, v)$ of $\st(x)$ such that $K(u) = K(v)$. In particular, this condition is automatically satisfied whenever $\blossom$ has been applied to at least one incoming edge of $\st(x)$.
\end{lemma}
\begin{proof}
  We prove the contrapositive: if an upper vertex $x$ is unlabeled, any incoming edge $(u, v)$ of $\st(x)$ satisfies $K(u) \neq K(v)$. By \cref{prop:kx_properties}~(MB1)--(MB3), we have $K(x) \subseteq \st(x)$. Thus, $\st[V(\st) \setminus K(x)]$ decomposes into a connected component $\st[V(\st) \setminus \st(x)]$ and the remaining connected components, where the union of the latter equals $\st(x) \setminus K(x)$. Hence, (MB3) guarantees that $K(y)$ for any $y \in \st(x) \setminus K(x)$ is disjoint from $V(\st) \setminus \st(x)$. As the same property trivially holds for $y \in K(x)$, it follows that any $y \in \st(x)$ satisfies $K(y) \subseteq \st(x)$, which completes the proof.
\end{proof}

\begin{corollary}
  \label{lem:subtree_valve}
  Let $x \in V(\st)$ be any unlabeled vertex on the search tree. For any search path $P(v)$ satisfying $P(v) \cap \st(x) \neq \emptyset$, the path contains $x$ and decomposes as $P(v) = P(u) x P(v|x)$, where $u = \rho(x)$ holds and $P(v|x)$ is contained within $\st(\bar{x}) = \st(x) \setminus \qty{x}$.
\end{corollary}
\begin{proof}
  The contrapositive of \cref{prop:incoming_edge_label} implies that no incoming edge of $\st(x)$ is examined by $\blossom$, and thus no such edge is contained in any search path by \cref{prop:path_edge_taxonomy}. Consequently, if $P(v)$ enters the subtree $\st(x)$, it must do so via the tree edges $(\rho(x), x)$ and $(x, \bar{x})$. Furthermore, the subsequent path cannot leave $\st(\bar{x})$; namely, $P(v|x) \subseteq \st(\bar{x})$.
  The decomposition follows from \cref{lem:search_path_prefix}.
\end{proof}

\subsubsection{Search Order}

In what follows, we characterize the fundamental properties of the search order, which are crucial for proving the uniqueness of a matching. We begin by analyzing the level values assigned by the algorithm.
\begin{lemma}
    \label{prop:level_values_bound}
    The levels assigned by the algorithm satisfy the following properties.
    \begin{enumerate}[label={(\roman*)}, leftmargin=*]
        \item Every non-tip vertex labeled by $\setlower(k)$ has a level of $k + 1$.
        \item For every upper vertex $x$ labeled by $\setupper(k, \kmax)$, its level satisfies $\lv(\rho(x)) + \lv(x) = 2k + 1$.
    \end{enumerate}
\end{lemma}
\begin{proof}
    Since statement~(i) is straightforward, we verify statement~(ii). Suppose that $\blossom(u, v)$ labels the vertex $x$, where $\lv(u) = \kmax$ and $\lv(v) = 2k + 1 - \kmax$. Assume $x \in P(u)$, as the other case $x \in P(v)$ follows by symmetry. Given that $P(u)$ decomposes as $P(\rho(x)) x \bar{x} P(u|\bar{x})$ by \cref{lem:search_path_prefix}, we have $\lv(\rho(x)) + \lv(x) = \lv(u) + \lv(v) = 2k + 1$ as required.
\end{proof}
\begin{corollary}
    \label{cor:level_at_least_k_1}
    During the search level $k$, levels can be bounded as follows.
    \begin{enumerate}[label={(\roman*)}, leftmargin=*]
        \item Every existing upper vertex $x$ satisfies $\lv(\rho(x)) \le k$.
        \item Every non-tip vertex newly labeled during this process has a level of at least $k + 1$.
    \end{enumerate}  
\end{corollary}
\begin{proof}
    Statement~(i) follows immediately by examining the steps during the search level $k$ that modify $\rho$, namely, $\grow$ and \ref{item:detail_4} of $\blossom$.
    Statement~(ii) is trivial if $x$ is labeled by $\setlower(k)$. If $x$ is labeled by $\setupper(k, \kmax)$, substituting $\lv(\rho(x)) \le k$ from statement~(i) for $\lv(\rho(x)) + \lv(x) = 2k + 1$ from \cref{prop:level_values_bound} (ii) completes the proof.
\end{proof}
For upper vertices labeled during $\setupper(k, \kmax)$, we will bound their levels from below by $\kmax$ in \cref{prop:level_at_least_kmax}, providing a tighter bound than the $k + 1$ value given in the above corollary. To derive this tighter bound, we establish the following lemma regarding the search order by leveraging the above corollary. 
\begin{lemma}
  \label{lem:search_order_grow}
  The following properties hold for any upper vertex $x$.
  \begin{enumerate}[label=(SO\arabic*), leftmargin=*]
    \item Any vertex $y$ with $y < \rho(x)$ is not adjacent to $x$ or $\bar{x}$.
    \item If $x$ is active, $\rho(x)$ is not adjacent to $\bar{x}$.
  \end{enumerate}
\end{lemma}
\begin{proof}
  In what follows, we inductively verify that these properties are maintained by each $\grow$ and $\blossom$ operation. Notably, any non-tip vertex $y$ labeled after the assignment of $\rho(x)$ cannot violate (SO1), because $y > \rho(x)$ always follows from \cref{cor:level_at_least_k_1}. Moreover, although the value of $\rho$ for existing vertices may be modified by replacing $\rho(x) = c$ with $\rho(x) = t$ in \ref{item:detail_4}, this update does not violate (SO1) for any two existing vertices $x$ and $y$. This is because the relation $y < \rho(x)$ is preserved, except that $y = c$ newly satisfies this relation, in which case $y$ cannot be adjacent to $\qty{x, \bar{x}}$. 
    
  Thus, it suffices to track two events to verify the statements: the assignment or update of $\rho$, and edge addition via pivoting operations. The following claim confines the endpoints of new edges, which simplifies the subsequent discussion. Note that the claim is inapplicable to an internal $\blossom$ call within a $\grow$ operation, because (SO2) is temporarily violated at this point.
  \begin{claim}
      \label{cl:new_edge_rho_c}
      Assume that (SO1) and (SO2) inductively hold before a $\blossom$ operation. Suppose this operation adds or deletes an edge $(w, z)$ between existing vertices; that is, this edge does not satisfy condition~(i) from \cref{prop:pivoting_new_edge} and thus satisfies only condition~(ii). Without loss of generality, assume that $w \in H$ and $z \notin H$. Then, $w$ is an upper vertex and $\rho(w) = c$ before the operation. Furthermore, $w$ is labeled during this $\blossom$ operation.
  \end{claim}
  \begin{proof}
      Let $w^+$ be the upper vertex of $\qty{w, \bar{w}}$. As $w$ was adjacent to $c$ before the pivoting operation by condition~(ii), $c \ge \rho(w^+)$ follows from (SO1). On the other hand, since $w^+ \in H$, \cref{rem:kc_is_lca} ensures that $c$ is an ancestor of $w^+$ on $\st$ before the pivoting operation. Hence, there exists a child $v$ of $c$ on $\st$ with $w^+ \in \st(v)$, yielding $\rho(w^+) \ge \rho(v) = c$ from \cref{lem:subtree_rho}. Hence, these inequalities imply $\rho(w^+) = c$. It also follows that $w^+$ is a child of $c$ on $\st$.
      
      Now, we show $w = w^+$ as follows. Given $w \in H$ and $c \notin H$, we have $K(w) \neq K(c)$ by the definition of $H$. Hence, the existence of an edge $(w, c) \in E^*$ ensures that these vertices are active. This also shows by (SO2) that $w$ is an upper vertex, i.e., $w = w^+$. 

      Finally, since $w$ is a child of $c$ on $\st$ and $K(w) \neq K(c)$, it follows from (MB1)--(MB3) and \cref{rem:kc_is_lca} that $w$ was unlabeled before the operation, which then implies that $w$ must be labeled during this operation.
  \end{proof}
  
  We now divide the argument into the following cases. Here, \cref{cl:new_edge_rho_c} is used in the second case to identify the affected vertices; namely, newly inserted buds and tips, and vertices at which $\rho$ is updated. 
  \begin{enumerate}[label=\textbf{Case \arabic*.}, leftmargin=*]
      \item A new original line $\qty{u, \bar{u}}$ identified by $\grow$ does not violate (SO1) or (SO2). Furthermore, if a dummy line is introduced by this process, the dummy line and newly added edges do not violate (SO1) or (SO2), either.
      \item Suppose a $\blossom$ operation inserts a dummy line $\qty{b, t}$ and modifies the edges via the pivoting operation. Then, (SO1) holds for $x = b$, $y = t$, or any existing vertex $x$ with $\rho(x) = c$. Similarly, (SO2) holds for $x = b$ or any existing vertex $x$ with $\rho(x) = c$. Note that the incident edges of these vertices may be added or removed.
  \end{enumerate}
  In the following proof, we extend the order $\le$ so that $y \ge \rho(x)$ holds for any upper vertex $x$ and any unlabeled vertex $y$. \cref{cor:level_at_least_k_1} guarantees that this extended relation is preserved when $y$ is labeled, allowing us to write the negation of $y < \rho(x)$ as $y \ge \rho(x)$ even for unlabeled $y$. Notably, although this extension makes the relation a partial order, the subsequent arguments rely only on its transitivity.
  
  \paragraph{Case 1.} First, suppose $\qty{u, \bar{u}}$ is an original line identified by $\grow(v)$, where $u$ is an upper vertex. It is straightforward to verify (SO2) for $x = u$. To prove (SO1) for $x = u$, it suffices to show that no vertex $y$ with $y < v$ is adjacent to $u$ or $\bar{u}$ before the $\grow$ operation, as the following argument shows. Since this statement is exactly (SO1) for $x = u$ if a dummy line is not inserted, we assume that a dummy line is inserted. Then, any neighbor of $u$ or $\bar{u}$ after the operation, other than the new tip $t$, must have been adjacent to $u$ or $\bar{u}$ before the operation. Since $\rho(u)$ is updated from $v$ to $t$, (SO1) for $x = u$ similarly follows by splitting $y < t$ into $y < v$ and $y = v$.

  Now, we prove this sufficient condition. Assume for contradiction that a vertex $y$ with $y < v$ is adjacent to $u$ or $\bar{u}$. If $y$ is not a tip, then $\grow(y)$ must have identified $\qty{u, \bar{u}}$ earlier than $\grow(v)$, which is a contradiction. Here, note that the edges $(y, u)$ and $(y, \bar{u})$ cannot be added after $y$ is labeled but before $\qty{u, \bar{u}}$ is identified; otherwise, $y \in H$ would hold at the time of addition, implying that $y$ must be unlabeled by \cref{cl:new_edge_rho_c}, a contradiction. If $y$ is a tip, we reduce the situation to the non-tip case. For the $\blossom$ operation that inserts $y$, \cref{prop:neighbor_of_bt} shows that the vertex $c$ from \ref{item:detail_2} must have been adjacent to $u$ or $\bar{u}$ prior to the operation. Given $c < y < v$, this reduces to the case where some vertex $c$ with $c < v$ is adjacent to $u$ or $\bar{u}$. Iterative application of this reduction eventually finds a non-tip vertex, resulting in a contradiction.

  We next show that a dummy line $\qty{b, t}$ inserted during a $\grow$ operation 
  does not violate (SO1) or (SO2). To verify (SO1) for $x = b$, we can assume $y \neq t$ because $t \ge \rho(b)$. By \cref{prop:neighbor_of_bt}, any neighbor $y$ of $b$ except $t$ was adjacent to $u$ prior to this insertion. As we have shown that $y \ge v$ follows in this case, $v = \rho(b)$ verifies (SO1) for $x = b$. The remaining cases are (SO1) for $y = t$ and (SO2) for $x = b$, which follow from the fact that $t$ is adjacent only to $b$ and $\qty{u, \bar{u}}$ by (BT2). 

  Finally, we consider a newly introduced edge $(w, z)$ between existing vertices. Although we cannot apply \cref{cl:new_edge_rho_c} to this case, we can still confine one of its endpoints to $w \in \qty{u, \bar{u}}$ by \cref{prop:pivoting_new_edge}. As $(z, w)$ exists after the operation, it follows that $(z, u)$ or $(z, \bar{u})$ must exist before the operation. Consequently, as we have shown that no vertex $y$ with $y < v$ is adjacent to $u$ or $\bar{u}$, the other endpoint $z$ satisfies $z \ge v$. Since an edge $(v, w)$ does not exist after the operation, we have $z \neq v$, from which we obtain $z \ge t = \rho(u)$. This completes the proof for case 1.

  \paragraph{Case 2.} We then consider a $\blossom(u, v)$ operation that inserts a dummy line $\qty{b, t}$. We distinguish the map $\rho$ before and after the operation by $\rho$ and $\hat{\rho}$.
  
  We first verify the contrapositive of (SO1) for $x = b$ or $y = t$. If both hold, then $t > c = \hat{\rho}(b)$ and thus (SO1) follows. Otherwise, for $x = b$ and $y \neq t$, assume that $y$ is adjacent to $b$ or $t$. If $y$ is adjacent to $b$, then $y$ must have been a neighbor of $g$ before the pivoting operation by \cref{prop:neighbor_of_bt}. Thus, $y \ge \rho(g)$ inductively follows from (SO1), which then yields (SO1) for $x = b$ by $\rho(g) = c = \hat{\rho}(b)$. If $y$ is adjacent to $t$, then $y$ must belong to the blossom $H$ by (BT2). Thus, \cref{lem:subtree_rho} and $H \subseteq \st(b)$ together imply the desired inequality $y \ge \hat{\rho}(b)$. Finally, for $x \neq b$ and $y = t$, assume that $y = t$ is adjacent to $x$ or $\bar{x}$. As the neighbor of $t$ must have been a neighbor of $c$ before the pivoting operation by \cref{prop:neighbor_of_bt}, this inductively implies from (SO1) that $c \ge \rho(x)$. Then, since replacing $\rho(x) = c$ with $\hat{\rho}(x) = t$ is the only possible update at the existing vertices, we have $t \ge \hat{\rho}(x)$, which completes the proof for this part.

  We then consider (SO1) for any existing vertex $x$ with $\rho(x) = c$. This holds trivially unless an edge $(y, x)$ is introduced. We can also assume $y \notin \qty{b, t}$, as $y = t$ is already handled and $y = b$ does not satisfy $y < \rho(x) = c$. Suppose such an edge $(x, y)$ is added. If $x \notin H$ (i.e., $w = y$ in \cref{cl:new_edge_rho_c}), then $y$ is labeled during this $\blossom$ operation, leading to $y \ge \hat{\rho}(x)$ by \cref{cor:level_at_least_k_1}. If $x \in H$ (i.e., $w = x$ in \cref{cl:new_edge_rho_c}), then \cref{prop:pivoting_new_edge} ensures $(g, y) \in E^*$ before the operation. Consequently, this implies $y \ge \rho(g) = c$ by (SO1). Furthermore, since the edge $(c, x)$ existed before the operation, a new edge $(y, x)$ cannot be added if $y = c$, implying $y \neq c$. Thus, we obtain $y \ge t$, which is exactly the desired inequality $y \ge \hat{\rho}(x)$.

  We next show that (SO2) holds for $x = b$ and any existing vertex $x$ such that $\rho(x)$ is modified. For the former, $\hat{\rho}(b) = c \notin H$ is not connected to $\bar{b} = t$ by (BT2). For the latter, $\rho(x) = c$ is replaced with $\hat{\rho}(x) = t$. Assume for contradiction that $\bar{x}$ is adjacent to $\hat{\rho}(x) = t$. Then, $\bar{x}$ must have been adjacent to $c$ before the pivoting operation. However, (SO2) inductively shows that $\bar{x}$ cannot be adjacent to $\rho(x) = c$, which is a contradiction.
\end{proof}

Using the above lemma, we establish the fundamental properties of the levels.
\begin{lemma}
  \label{prop:level_at_least_kmax}
  Every upper vertex $x$ labeled during $\setupper(k, \kmax)$ has a level of at least $\kmax$.
\end{lemma}
\begin{proof}
  Suppose that $\blossom(u, v)$ labels the vertex $x$, where $\lv(u) = \kmax$ and $\lv(v) = 2k + 1 - \kmax$. Since the case where $x \in P(v)$ is straightforward, we assume $x \in P(u)$. Let $u^+$ be the upper vertex in $\qty{u, \bar{u}}$. As $x$ is unlabeled prior to the operation, \cref{lem:subtree_rho,lem:subtree_valve} together imply $\rho(x) \le \rho(u^+)$. Furthermore, \cref{lem:search_order_grow} ensures that $\rho(u^+) \le v$. Hence, we have $\lv(\rho(x)) \le \lv(v)$. Since $\lv(\rho(x)) + \lv(x) = 2k + 1$ by \cref{prop:level_values_bound}~(ii), we obtain $\lv(x) = 2k + 1 - \lv(\rho(x)) \ge 2k + 1 - \lv(v) = \kmax$, as desired.
\end{proof}
\begin{corollary}
  \label{cor:blm_monotone}
  Suppose $\blossom(u, v)$ labels an upper vertex $w$ during $\setupper$. Then, for any edge $(w, z)$ where $z$ is also labeled, $(u, v)$ precedes $(w, z)$ in the order of $\setupper$.
\end{corollary}
\begin{proof}
  The precondition of $\blossom$ ensures that $u > v$. Since $\lv(u) \le \lv(w)$ follows from \cref{prop:level_at_least_kmax} and $w$ is labeled after $u$, it follows that $u < w$. Given these relations, it suffices to prove that $\lv(u) + \lv(v) \le \lv(w) + \lv(z)$; the precedence is immediate if the inequality is strict, and otherwise follows from $v < u < w$. To derive the desired inequality, we can utilize \cref{lem:search_order_grow} to show $\rho(w) \le z$, and thus $\lv(\rho(w)) \le \lv(z)$. Since \cref{prop:level_values_bound}~(ii) ensures $\lv(\rho(w)) + \lv(w) = \lv(u) + \lv(v)$, combining these inequalities completes the proof.
\end{proof}
\cref{prop:level_values_bound}~(i) and \cref{cor:blm_monotone} together guarantee the monotonicity of the assigned levels with respect to the execution order; that is, they prevent assigning a level so small that it would require a vertex to be examined by $\setlower$ or $\setupper$ prior to its labeling. Although tips appear to violate such monotonicity, this is inconsequential because every tip is connected to the vertices within the blossom by (BT2), implying that $\grow$ and $\blossom$ operations involving tips would have no influence. Based on these observations, we prove the following lemma, which is fundamental to establishing optimality in the absence of an augmenting path (see \cref{sec:lmp_optimality}).
\begin{lemma}
  \label{lem:effectively_all_executed}
  If $\grow$ has been applied to a vertex $u$, every line $\qty{w, \bar{w}}$ with $(v, w) \in E^*$ for some vertex $v \le u$ has its lower level assigned. Similarly, if $\blossom(u, v)$ has been performed, we have $K(w) = K(z)$ for every edge $(w, z)$ that precedes $(u, v)$ in the order of $\setupper$.
\end{lemma}
\begin{proof}
  Since labeling vertices preserves these conditions, we only need to verify that the addition of new edges does not violate them. By \cref{prop:pivoting_new_edge}, the new edges $(x, y)$ are classified into two cases: (i) $\qty{x, y} \cap \qty{b, t} \neq \emptyset$, and (ii) $x \in H$, $y \notin H$, and $(x, c), (g, y) \in E^*$ prior to the operation. It is straightforward to verify that edges satisfying condition~(i) do not violate the desired conditions. Otherwise, for edges satisfying only condition~(ii), \cref{cl:new_edge_rho_c} shows that $x$ becomes labeled during the current $\blossom$ operation, which consequently preserves the conditions by \cref{cor:blm_monotone}.
\end{proof}

\subsubsection{Path Decomposition}

This section discusses the properties of $L(x)$ and $P(x)$. By leveraging $L(x)$ as a basic segment and decomposing the search path $P(x)$ into these segments, we inductively verify that $G^*[L(x) \setminus \qty{x}]$ and $G^*[P(x) \setminus \qty{x}]$ have a unique perfect matching. 

Before discussing the properties of $L(x)$, we first focus on each vertex $y \in P(x)$ that is unlabeled or satisfies $y > x$, because such vertices may subsequently determine the segment $L(x)$. The following lemma generalizes \cref{prop:path_property,lem:search_path_prefix}.
\begin{lemma}
    \label{lem:modified_ell_prefix}
    For a search path $P(x)$, let $y \in P(x)$ be a vertex that is unlabeled or satisfies $y > x$. Then, $y$ is an upper vertex and lies at an even position in $P(x)$. Furthermore, $P(x)$ decomposes as $P(\rho(y)) y \bar{y} P(x|\bar{y})$. Moreover, if $x$ is active, every upper vertex $w \in P(x|\bar{y})$ satisfies $\rho(y) < \rho(w)$.
\end{lemma}
\begin{proof}
    We proceed by induction on the $\grow$ and $\blossom$ operations. We have to consider the following three cases: (i) $P(x)$ is defined by $\grow$, (ii) $P(x)$ is modified by $\blossom$, and (iii) $P(x)$ is defined by $\blossom$. Since it is straightforward to verify case (i), we focus on cases (ii) and (iii).

  \textbf{Case (ii).}
  Suppose $P(x) = P(c) P(x|c)$ is replaced with $P(c) bt P(x|c)$. Since the case when $y = b$ follows directly, we assume that $y$ existed prior to the operation. Then, the decomposition is preserved as follows. If $\rho(y) = c$, then $P(x)$ is updated to $P(c) bt y \bar{y} P(x|\bar{y})$ and thus decomposes as $P(t) y \bar{y} P(x|\bar{y})$. Since $\rho(y)$ is also updated to $t$, this preserves the decomposition. If $c$ precedes $\rho(y)$ in $P(x)$, then $P(\rho(y))$ is also updated to $P(c) bt P(\rho(y)|c)$, preserving the decomposition. If $\rho(y)$ precedes $c$ in $P(x)$, then $P(\rho(y)) y \bar{y}$ is a prefix of $P(c)$ and thus of the updated path $P(x) = P(c) bt P(x|c)$, as required. The remaining conditions can be easily verified.

  \textbf{Case (iii).}
  Since the conditions follow directly when $x$ is a new tip, we assume that $x$ existed prior to the operation. Suppose $x \in P(u)$, because the other case $x \in P(v)$ follows by symmetry. 
  
  We first show that $y$ is an upper vertex and satisfies $\rho(y) \le \rho(x)$, which is strict if $x$ is active. If $y$ is already labeled prior to the operation, we have $\lv(y) > \lv(x)$, which implies that $y$ is an upper vertex by \cref{prop:level_values_bound,cor:level_at_least_k_1}. Furthermore, since $y$ is labeled prior to $x$, \cref{prop:level_values_bound}~(ii) ensures $\lv(\rho(x)) + \lv(x) \ge \lv(\rho(y)) + \lv(y)$, which yields $\lv(\rho(y)) < \lv(\rho(x))$ and thus $\rho(y) < \rho(x)$. If $y$ also becomes labeled during the blossom operation, the definition of the labeling order and the tie-breaking rule in \ref{item:detail_5} ensure $\rho(y) \le \rho(x)$, which are equal only when $x$ is retired. If $y$ remains unlabeled after the operation, then we have $x \in \st(\bar{y})$ by \cref{lem:subtree_valve}, which implies $\rho(y) < \rho(x)$ by \cref{lem:subtree_rho}.

  We then verify that $y \in P(v)$ holds. Since $\rho(y) \le \rho(x)$, the contrapositive of \cref{lem:subtree_rho} implies $y \notin \st(\bar{x})$. Since $P(u|x) \subseteq \st(\bar{x})$ by \cref{lem:subtree_valve}, $y \notin P(u|x)$ and thus $y \in P(v)$. Hence, the inductive assumption for $P(v)$ ensures that $y$ lies at an even position in $P(x)$, and $P(x)$ decomposes as $P(\rho(y)) y \bar{y} P(x|\bar{y})$.

  Lastly, assuming $x$ is active, we show that every upper vertex $w \in P(x|\bar{y})$ satisfies $\rho(y) < \rho(w)$. The case $w = x$ has already been verified. If $w \in P(v|\bar{y})$, the inequality $\rho(y) < \rho(w)$ follows from the inductive assumption for $P(v)$. When $w \in P(u|x)$, the inductive assumption for $P(u)$ yields $\rho(x) < \rho(w)$, which can be combined with $\rho(y) < \rho(x)$ to derive the desired inequality. This exhausts all cases.
\end{proof}

We now establish fundamental properties of $L(x)$ as corollaries. \cref{lem:ell_disjoint_path} is crucial for decomposing a search path.
\begin{corollary}
  \label{lem:lx_property}
  The cardinality $\abs{L(x)}$ is odd. Furthermore, $L(x) \subseteq K(x)$ holds.
\end{corollary}
\begin{proof}
  The statements are nontrivial only when $L(x)$ is defined in \ref{item:detail_5}. The former statement is a direct consequence of \cref{lem:modified_ell_prefix}.
  For the latter statement, since the bud $b$ of $K(x)$ is unlabeled, $y$ in \ref{item:detail_5} never precedes $b$ in $P(x)$, which yields $L(x) \subseteq P(x | b) \subseteq K(x)$.
\end{proof}

\begin{corollary}
  \label{lem:ell_disjoint_path}
  For any labeled vertex $x$ other than $s_0$, the search path $P(x)$ decomposes as $P(x) = P(z) y L(x)$, where $z = \rho(y) < x$ holds. Moreover, $z < \rho(w)$ for any upper vertex $w \in L(x)$.
\end{corollary}
\begin{proof}
    The case where $x$ is active can be verified by applying \cref{lem:modified_ell_prefix} to the vertex $y$ in the definition of $L(x)$. If $x$ is retired, then \cref{lem:subtree_rho,prop:level_values_bound} ensure that no vertex $y \in P(x|b)$ satisfies $y > x$ unless a tie occurs in the labeling order. Consequently, we obtain $L(x) = P(x|b)$ regardless of the existence of a tie, and the statements follow from \cref{lem:search_path_prefix,lem:subtree_rho}.
\end{proof}

Every vertex $u \in L(x)$ satisfies $u \le x$, with the possible exception of a single vertex $y$; in this case, either $x$ or $y$ is retired. The following lemma extends \cref{lem:search_order_grow} to $L(x)$.
\begin{lemma}
  \label{lem:grow_ell_disjoint}
  Let $x$ be an upper vertex and $y$ be an active labeled vertex with $y \le \rho(x)$. No vertex in $L(y)$ is adjacent to $x$ or $\bar{x}$, with the following exceptions: (i) $y$ is adjacent to $x$ when $y = \rho(x)$, and (ii) $y$ is adjacent to $\bar{x}$ when $y = \rho(x)$ and $x$ is marked as retired during a $\grow$ call.
\end{lemma}
\begin{proof}
  First, consider $z \in L(y)$ with $z \le y$. Given $z \le y \le \rho(x)$, if at least one of these inequalities is strict, \cref{lem:search_order_grow} (SO1) ensures that $z$ is not adjacent to $x$ or $\bar{x}$. If $z = y = \rho(x)$, since the edge $(y, x)$ exists by the definition of $\rho$, we only need to examine the existence of the edge $(y, \bar{x})$. If $x$ is active, this edge does not exist by \cref{lem:search_order_grow} (SO2). The remaining cases are straightforward; the edge $(y, \bar{x})$ exists (resp.~does not exist) if $x$ is marked retired by a $\grow$ (resp.~$\blossom$) call.
  
  Next, consider $z \in L(y)$ with $z > y$. In this case, $z$ must be retired, as $y$ is active. Consider the $\blossom$ operation that marks $z$ as retired. If this operation is called within a $\grow$ operation, then $z$ is adjacent only to the tip, from which the statement follows. Otherwise, this case never occurs, because such a situation would correspond to a tie mentioned in \ref{item:detail_5} with $g = z$ and $h = y$, contradicting the algorithm specification $g < h$ and the assumption $z > y$.
\end{proof}

\begin{lemma}
  \label{lem:ell_disjoint_decomp}
  Every search path $P(x)$ can be decomposed as follows:
  \begin{equation}
    P(x) = L(x_1) y_1 \dots L(x_k) y_k L(x_{k+1}),
  \end{equation}
  where $x_{k+1} = x$ and $x_1 = \rho(y_1) < \dots < x_k = \rho(y_k) < x_{k+1}$.
  Also, the vertex $x_i$ is active for each $i = 1, \dots, k$.
  Moreover, for any $i$ and $j$ satisfying $1 \le i < j \le k + 1$, there are no edges between $L(x_i)$ and $L(x_j)$.
\end{lemma}
\begin{proof}
  The decomposition can be obtained by recursively applying \cref{lem:ell_disjoint_path}.
  The fact that $x_i$ is active follows since search paths to retired vertices will not compose a prefix of any other search paths.
  We now prove that $L(x_i)$ is not adjacent to $L(x_j)$ for $j > i$. Since any upper vertex $v \in L(x_j)$ satisfies $\rho(y_{j - 1}) < \rho(v)$ by \cref{lem:ell_disjoint_path}, no vertex of $L(x_i)$ is adjacent to $\qty{v, \bar{v}}$ by \cref{lem:grow_ell_disjoint}. Also, since $x_i \le x_{j - 1} = \rho(y_{j - 1})$ and $y_{j - 1}$ cannot be marked as retired during a $\grow$ operation, no vertex of $L(x_i)$ is adjacent to $\bar{y}_{j - 1}$, either. This exhausts all vertices of $L(x_j)$.
\end{proof}

\begin{lemma}
  \label{lem:ell_disjoint_blossom}
  Suppose $\blossom(u, v)$ labels an existing vertex $x$. Define $L(x_1), \dots, L(x_k) \subseteq P(u)$ and $L(x_{k+1}), \dots, L(x_{l}) \subseteq P(v)$ as follows: if $x \in P(u)$ (resp.~$x \in P(v)$), then decompose $P(u|x)$ (resp.~$P(u)$) as $L(x_1) y_1 \dots L(x_k) y_k L(u)$, and $P(v)$ (resp.~$P(v|x)$) as $L(x_{k + 1}) y_{k + 1} \dots L(x_l) y_l L(v)$. Then, $u' \in L(x_1) \cup \dots \cup L(x_k) \cup L(u)$ and $v' \in L(x_{k+1}) \cup \dots \cup L(x_l) \cup L(v)$ are adjacent if and only if $(u', v') = (u, v)$.
\end{lemma}
\begin{proof}
  By symmetry, it suffices to consider the case $x \in P(u)$. Since $P(u|x) \subseteq \st(x)$ and $P(v) \cap \st(x) = \emptyset$ hold, \cref{prop:incoming_edge_label,lem:effectively_all_executed} verify the case where $u' \le u$ and $v' \le v$. If $u'$ or $v'$ fails to satisfy its corresponding inequalities, such a vertex must have been retired before the current $\blossom$ operation. However, since $x$ is unlabeled before the operation, this retired vertex cannot form an endpoint of any incoming edge of $\st(x)$.
\end{proof}

\begin{lemma}
  \label{prop:ell_unique_matching}
  For every labeled vertex $x$, the induced subgraph $G^*[L(x) \setminus \qty{x}]$ has a unique perfect matching.
\end{lemma}
\begin{proof}
  Since the statement can be easily verified when $\grow$ defines $L(x)$, we focus on the case where $\blossom$ defines $L(x)$.
  Assume $x \in P(u)$ by symmetry, and suppose that $L(x) \setminus \qty{x}$ can be represented as $P(v|w) \overline{P(u|x)}$. 

  We show that $G^*[L(x) \setminus \qty{x}]$ has a unique perfect matching by the following induction. Since $v < u < x$ holds, and $w$ is unlabeled or satisfies $w > x$, we can decompose $P(u|x)$ and $P(v|w)$ as $L(x_1) y_1 \dots L(x_k) y_k L(u)$ and $L(x_{k + 1}) y_{k + 1} \dots L(x_l) y_l L(v)$, respectively. The base case is $X = L(u) \cup L(v)$, and each inductive step adds $L(x_i) y_i$ to $X$, where $x_i$ is chosen in decreasing order of the total order $\le$.

  We begin with the base case. Combining the unique perfect matchings of $G^*[L(u) \setminus \qty{u}]$ and $G^*[L(v) \setminus \qty{v}]$ with the edge $(u, v)$ yields a perfect matching of $G^*[L(u) \cup L(v)]$. We then prove its uniqueness. Any perfect matching on $G^*[L(u) \cup L(v)]$ should contain $(u, v)$, since \cref{lem:ell_disjoint_blossom} ensures that it is the only edge bridging these odd-sized components $L(u)$ and $L(v)$. Therefore, the remaining edges constitute the unique perfect matchings of $G^*[L(u) \setminus \qty{u}]$ and $G^*[L(v) \setminus \qty{v}]$, which finishes the proof for the base case.

  Next, we focus on the inductive step. Combining the unique perfect matchings of $G^*[L(x_i) \setminus \qty{x_i}]$ and $G^*[X]$ with the edge $(x_i, y_i)$ yields a perfect matching on $G^*[X']$ for $X' \coloneqq L(x_i) \cup \qty{y_i} \cup X$. We then prove its uniqueness. We first show that $L(x_i)$ is not adjacent to $X$. For vertices in $X$ that belong to the decomposed segments, it follows from Lemmas~\ref{lem:ell_disjoint_decomp} and \ref{lem:ell_disjoint_blossom}. For each separating vertex $y_j \in X$, the order of choosing $x_i$ implies $x_i < x_j = \rho(y_j)$, which ensures that $L(x_i)$ is not adjacent to $y_j$ by \cref{lem:grow_ell_disjoint}. Therefore, the edge $(x_i, y_i)$ is the only edge bridging $L(x_i)$ and $X' \setminus L(x_i)$ by \cref{lem:grow_ell_disjoint}, which must be contained by any perfect matching on $X'$ because $\abs{L(x_i)}$ is odd. As the remaining matching edges split into the unique perfect matchings on $G^*[L(x_i) \setminus \qty{x_i}]$ and $G^*[X]$, the uniqueness follows.
\end{proof}

A similar inductive argument on the path decomposition yields the following corollary.
\begin{corollary}
  \label{cor:search_path}
  For every search path $P(x)$, the induced subgraph $G^*[P(x) \setminus \qty{x}]$ has a unique perfect matching.
\end{corollary}

Finally, we show that the concatenated path found by the algorithm is an augmenting path. Note that the following corollary does not assume the reduction in \cref{sec:super_singleton_reduction}.

\begin{corollary}
  \label{cor:augmenting_path_nonsingular}
  Let $P$ be the augmenting path found by the algorithm. The induced subgraph $G^*[P]$ has a unique perfect matching.
\end{corollary}
\begin{proof}
  We apply the super singleton reduction to reduce this corollary to \cref{prop:ell_unique_matching}. Suppose the algorithm reports an augmenting path $P$ in the original instance, and let $s_i$ and $s_j$ be the singletons in $P$. In the reduced instance, the algorithm labels $\bar{s}_i$ and $\bar{s}_j$, forming a blossom containing these vertices. Since $G^*[L(\bar{s}_i) \setminus \qty{\bar{s}_i}]$ possesses a unique perfect matching, removing the edge between $\bar{s}_j$ and its adjacent tip yields a unique perfect matching on $G^*[P]$, which completes the proof.
\end{proof}

\begin{lemma}
  For the augmenting path $P$ found by the algorithm, $P \setminus T$ is an augmenting path for the initial base $B$.
\end{lemma}
\begin{proof}
    As the set $P \setminus T$ obviously satisfies conditions (i) and (ii) from \cref{def:lmp_aug}, we only need to show that $C[P \setminus T]$ is nonsingular. Given \cref{cor:augmenting_path_nonsingular}, $C^*[P]$ is nonsingular, which in turn shows that $C[P \setminus T]$ is nonsingular by \cref{lem:path_coherence,lem:nonsingular_invariance}. 
\end{proof}

\subsection{Minimality}
\label{sec:lmp_minimality}

This section is dedicated to proving that the algorithm finds a shortest augmenting path.
Unless otherwise stated, the super singleton reduction is also assumed throughout this section.

\subsubsection{Formalization}

Since an augmenting path is constructed by concatenating search paths, the pursuit of the shortest augmenting path necessitates a formal identification of ``shortest search paths.''
This is straightforward in graph matching, because search paths are readily characterized as alternating paths, whose concatenation consistently yields a valid augmenting path. In contrast, the situation in linear matroid parity is significantly more complicated due to its dependence on the processing order. Specifically, if the algorithm skips a concatenation step that would yield an augmenting path, subsequent concatenation steps may result in an invalid augmenting path due to accidental numerical cancellation. Such dependence makes the characterization of search paths inseparable from the search order, thereby hindering an a priori definition of shortest search paths.

To address this issue, we instead give a lower bound on the search path length and prove that our algorithm attains this bound. To this end, we leverage \cref{lem:path_coherence,cor:search_path}, which characterize $P(v) \setminus \qty{v}$ as an ``augmenting path'' connecting $s_0$ and $\bar{v}$, because we can similarly define the shortest one for this structure. This idea is formalized via a \textit{coherent nonsingular path}, where coherence is given in \cref{def:coherence} and nonsingularity is defined below.
\begin{definition}
  A vertex set $P \subseteq V^*$ is called a \textit{nonsingular path} between $v$ and $v'$ if $P$ consists of $v$, $v'$, and lines, and the submatrix $C^*[P]$ is nonsingular.
\end{definition}
To simplify the notation, when a nonsingular path $P$ connects $s_0$ and some vertex $v$, we omit the mention of $s_0$ and refer to $P$ as a nonsingular path to $v$. Note that a nonsingular path $P$ has a perfect matching of $G^*[P]$ and thereby decomposes into an alternating path and possibly alternating cycles on the graph $G^*$. Hence, similar to augmenting paths, nonsingular paths may not necessarily form a path on $G^*$.

Next, we define the length of a coherent nonsingular path and thereby define a minimum distance to a node.
\begin{definition}
  The length of a coherent nonsingular path $P$ is defined as the number of original lines intersecting the path $P$. We also define $\dist(v)$ as the minimum length of a coherent nonsingular path to $v$. We define $\dist(v) \coloneqq \infty$ if there are no such paths. The paths that attain the minimum are called the shortest coherent nonsingular paths to $v$.
\end{definition}

The following two lemmas are crucial for bounding the length of the search paths.
\begin{lemma}
  The value $\dist(v)$ is invariant under the insertion of dummy lines.
\end{lemma}
\begin{proof}
    It suffices to show that this value remains unchanged upon the introduction of a dummy line $\qty{b, t}$. Let $P$ be a shortest coherent nonsingular path to $v$ before the insertion. In this case, $P \cup \qty{b, t}$ is a coherent nonsingular path after the insertion, which implies that $\dist(v)$ cannot increase.    
    
    Conversely, consider a shortest coherent nonsingular path $P'$ to $v$ after the operation. Suppose $P'$ does not contain the dummy line $\qty{b, t}$. Then, since $P'$ is disjoint from the corresponding blossom by the coherence, the neighbor of the tip $t$ within $G^*[P' \cup \qty{b, t}]$ is restricted to $b$. This yields $\det C^*[P' \cup \qty{b, t}] = C^*[\qty{b, t}] \cdot \det C^*[P'] \neq 0$, showing that $P' \cup \qty{b, t}$ is another shortest coherent nonsingular path to $v$. Now, assume $P' \supseteq \qty{b, t}$ without loss of generality. We then have $P' \setminus \qty{b, t}$ as a coherent nonsingular path before the operation, implying that $\dist(v)$ cannot decrease.
\end{proof}

\begin{lemma}
  \label{prop:level_lower_bound}
  For any labeled vertex $x \in V^* \setminus \qty{s_0}$, we have $\lv(x) \ge \dist(\bar{x})$.
\end{lemma}
\begin{proof}
  \cref{lem:path_coherence,cor:search_path} guarantee that $P(x) \setminus \qty{x}$ is a coherent nonsingular path of length $\lv(x)$ to $\bar{x}$.
\end{proof}

The remainder of this section is dedicated to proving that the above lower bound is attained by our algorithm. Using \cref{prop:level_values_bound,prop:level_at_least_kmax}, we can consider the time interval in which a certain level is assigned.
\begin{enumerate}[label=(Pre-LV\arabic*), leftmargin=*]
  \item If a non-tip lower vertex $v$ becomes labeled, this must occur during the execution of $\setlower(k)$, where $k = \lv(v) - 1$.
  \item If an upper vertex $v$ becomes labeled, this must occur during the $k$-th iteration and before the completion of $\setupper(k, \kmax)$, where $k$ and $\kmax$ satisfy $2k + 1 = \lv(\rho(v)) + \lv(v)$ and $\kmax = \lv(v)$.
\end{enumerate}
Note that the assignment for the lower level of a dummy line is immediate upon its insertion and thus has no counterpart.

Assuming that the minimum levels are attained, these properties naturally define the ``deadlines'' for level assignments. The following lemma verifies that our algorithm indeed labels within these time intervals.
\begin{lemma}
  \label{prop:level_minimum}
  The levels assigned by the algorithm satisfy the following properties.
  \begin{enumerate}[label=(LV\arabic*), leftmargin=*]
    \item The lower level of an original line $\qty{v, \bar{v}}$ is assigned by $\setlower(k)$, where $k = \min\{\dist(v), \dist(\bar{v})\} - 1$.
    \item For an upper vertex $v$, its level is assigned during the $k$-th iteration and before the completion of $\setupper(k, \kmax)$, where $k$ and $\kmax$ satisfy $2k + 1 = \lv(\rho(v)) + \dist(\bar{v})$ and $\kmax = \dist(\bar{v})$.
    \item For any labeled vertex $x \in V^* \setminus \qty{s_0}$, we have $\lv(x) = \dist(\bar{x})$.
  \end{enumerate}
  If $k = \infty$ in (LV1) or (LV2), the condition imposes no constraint on the level assignment.
\end{lemma}
\begin{remark}
  \label{rem:theorem_odd}
  A simple parity argument shows that a labeled vertex $x$ belongs to the base if and only if $\lv(x)$ is even; the same applies to $\dist(\bar{x})$. Consequently, since exactly one of $\rho(v)$ and $v$ belongs to the base, the sum $\lv(\rho(v)) + \dist(\bar{v})$ in (LV2) is necessarily odd whenever it is finite.
\end{remark}

We show by induction that (LV1)--(LV3) are invariant under each execution of the subroutines $\setlower(k)$ and $\setupper(k, \kmax)$. Under the reduction in \cref{sec:super_singleton_reduction}, we can easily verify that the iterations for $k = 0$ and $k = 1$ satisfy these conditions. Hence, we focus on an iteration for $k \ge 2$ and prove that the procedure maintains these properties.

Before presenting the formal proof, we first outline the main idea of the inductive step. Intuitively, for every target vertex that (LV1) or (LV2) requires to be labeled, we show that there exists an associated ``triggering'' vertex or edge upon which the labeling operation is performed; this trigger is an adjacent labeled vertex for a $\grow$ operation, and an incoming edge for a $\blossom$ operation. We then demonstrate that this trigger is prepared well in advance of the operation.

In particular, for every line $\qty{v, \bar{v}}$ for which the lower level $k + 1$ must be assigned, we prove that there exists an associated labeled vertex $u$ such that $\grow(u)$ finds $\qty{v, \bar{v}}$ and $\lv(u) = k$ is assigned prior to the execution of $\setlower(k)$. Similarly, we show that every upper vertex $v$ in (LV2) has an associated incoming edge $(z, y)$ such that $\blossom(z, y)$ labels $v$, and both $z$ and $y$ are labeled before the completion of $\setupper(k, \kmax)$.

The remainder of this section is devoted to the formal proof that $\setlower(k)$ and $\setupper(k, \kmax)$ for each $\kmax$ preserve (LV1)--(LV3). To this end, we classify vertices into three types and formulate the requirements for each type.

\begin{claim}[\textbf{Tips}]
  \label{cl:tips}
  For each dummy line introduced during the $k$-th iteration, the tip $t$ satisfies (LV3), and the bud $b$ does not violate (LV2) at the moment it is created.
\end{claim}
\begin{claim}[\textbf{Non-tip lower vertices}]
  \label{cl:non_tip_lower}
  The lower vertex of an original line $\qty{v, \bar{v}}$ with $k = \min\{\dist(v), \dist(\bar{v})\} - 1$ is labeled before the completion of $\setlower(k)$, ensuring (LV1). Moreover, its mate does not violate (LV2) at the time it is designated as an upper vertex.
\end{claim}
\begin{claim}[\textbf{Upper vertices}]
  \label{cl:upper}
  If an upper vertex (including a bud) satisfies $2k + 1 = \lv(\rho(v)) + \dist(\bar{v})$ and $\kmax = \dist(\bar{v})$, it is labeled before $\setupper(k, \kmax)$ completes, ensuring (LV2).
\end{claim}

Regarding the second and third claims, (LV3) follows from \cref{prop:level_values_bound,prop:level_at_least_kmax}. Hence, these claims verify the inductive step of (LV1)--(LV3). Assuming that these claims are established and thus \cref{prop:level_minimum} is proved, we can also show that the augmenting path found by the algorithm is the shortest by leveraging the reduction as in the proof of \cref{cor:augmenting_path_nonsingular}.
\begin{lemma}
  For the augmenting path $P$ found by the algorithm, $P \setminus T$ is a shortest augmenting path for the initial base $B$.
\end{lemma}
\begin{proof}
  We reduce this lemma to \cref{prop:level_minimum} using the super singleton reduction. Let $P'$ be a shortest augmenting path satisfying the conditions in \cref{def:lmp_aug}. Let $s_i$ and $s_j$ be the singletons in $P'$. Using the super singleton reduction, $\qty{s_0,q,\bar{q}, \bar{s}_i} \cup P'$ in the reduced instance is a shortest coherent nonsingular path to $s_j$, which yields $\dist(s_j) = \abs{P'} / 2 + 2$. Since \cref{prop:level_minimum} ensures that the algorithm assigns $\lv(\bar{s}_j) = \abs{P'} / 2 + 2$ in the reduced instance, the original instance detects an augmenting path no later than the corresponding moment. Therefore, the number of original lines in $P$ is at most that of $P'$, which guarantees that $P \setminus T$ is also the shortest.
\end{proof}

The following subsections prove these three claims by induction on $k$ and $\kmax$. We first verify the inductive step for $\setlower(k)$ in \cref{sec:setlower_minimality}, followed by the verification for $\setupper(k, \kmax)$ in \cref{sec:setupper_minimality}.

\subsubsection{Lower Vertices}
\label{sec:setlower_minimality}

This section verifies \cref{cl:tips,cl:non_tip_lower} for $\setlower(k)$ and establishes the inductive step for $\setlower(k)$. Throughout this section, we inductively assume that all the subroutines prior to $\setlower(k)$ preserve (LV1)--(LV3).

The following lemma splits a coherent nonsingular path into a smaller coherent nonsingular path and an edge. In terms of the intuition described earlier, this lemma guarantees the existence of a triggering vertex for the $\grow$ operation.
\begin{lemma}
  \label{lem:laplace_for_v}
  A coherent nonsingular path $P$ to $v \in V$ can be decomposed into an edge $(u, v) \in E^*$ and a set $P \setminus \qty{u, v}$, which is either a smaller coherent nonsingular path to $\bar{u}$ or the empty set.
\end{lemma}
\begin{proof}
  Such a vertex $u$ is obtained via the Laplace expansion of the row or column of $v$ in $C^*[P]$. The coherence of $P \setminus \qty{u, v}$ is immediate from the coherence of $P$ and the fact that $v \notin T$.
\end{proof}
Applying this lemma to a shortest coherent nonsingular path yields the following corollary.
\begin{corollary}
  \label{cor:dist_pred_v}
  For any vertex $v \in V$ with $\dist(v) < \infty$, there exists a vertex $u$ with $(u, v) \in E^*$ such that $\dist(\bar{u}) < \dist(v)$.
\end{corollary}

Assuming that (LV1)--(LV3) have been maintained throughout all the preceding iterations, we show that all vertices requiring level $k$ are already labeled and thus are ready to be processed by $\setlower(k)$ at the start of the $k$-th iteration.
\begin{lemma}
  \label{prop:at_most_k}
  During the search level $k$, every vertex $v \in V^* \setminus \qty{s_0}$ with $\dist(\bar{v}) \le k$ is labeled and satisfies $\lv(v) = \dist(\bar{v})$, as long as (LV1)--(LV3) are preserved for $\qty{v, \bar{v}}$. In particular, the statement inductively holds at the beginning of the search level $k$.
\end{lemma}
\begin{proof}
  It suffices to show that the vertex $v$ is labeled, since the level value follows from (LV3). If $\qty{v, \bar{v}}$ is an original line, its lower level must have been assigned before the completion of $\setlower(k - 1)$ by (LV1). If it is a dummy line, its lower level must have been assigned at the time of its insertion. In either case, the lower vertex is already labeled. As this shows the assertion for the case where $v$ is the lower vertex, it remains to consider the case where $v$ is an upper vertex. While we have $\lv(\rho(v)) \le k$ by \cref{cor:level_at_least_k_1} and $\dist(\bar{v}) \le k$, their sum is odd by \cref{rem:theorem_odd}, and thus cannot exceed $2k - 1$. Hence, (LV2) ensures that $v$ is labeled by an earlier execution of $\setupper$ for a smaller value of $k$, which completes the proof.
\end{proof}

Finally, we verify the correctness of $\setlower(k)$.
\begin{lemma}
  The subroutine $\setlower(k)$ preserves (LV1)--(LV3).
\end{lemma}
\begin{proof}
  We first verify \cref{cl:non_tip_lower} for the non-tip lower vertices that must be labeled during this process. Let $\qty{u, \bar{u}} \subseteq V$ be any original line satisfying $\dist(u) = k + 1$ and $\dist(\bar{u}) \ge k + 1$. Applying \cref{cor:dist_pred_v} at the beginning of $\setlower(k)$ yields a vertex $v$ such that $(u, v) \in E^*$ and $\dist(\bar{v}) \le k$. Thus, \cref{prop:at_most_k} ensures that $v$ is labeled and its level is at most $k$. Since the removal of the edge $(u, v)$ occurs only after $\qty{u, \bar{u}}$ has been found, it follows that the line $\qty{u, \bar{u}}$ is found no later than $\grow(v)$ and thus before the completion of $\setlower(k)$.
  Also, the upper vertex does not violate (LV2) at the time it is examined by $\grow$, since both $\dist(u)$ and $\dist(\bar{u})$ are at least $k + 1$.

  Next, we verify \cref{cl:tips} for the dummy lines created during this process. Suppose $\grow(v)$ during $\setlower(k)$ finds an original line $\qty{u, \bar{u}}$ and creates a blossom during this operation. We first prove (LV2) for the bud $b$. If $\dist(t) \ge k + 1$, we have $\lv(\rho(b)) + \dist(t) \ge 2k + 1$ and thus $b$ does not violate (LV2). Assume $\dist(t) \le k$ for contradiction. Let $P$ be a shortest coherent nonsingular path to $t$. Following the proof of \cref{lem:laplace_for_v}, we apply the Laplace expansion to $C^*[P]$ twice, first at $u$ and then at $\bar{u}$. Since $u$ is only adjacent to $t$, this yields the edges $(t, u), (\bar{u}, w)$ for some vertex $w$ and a nonsingular path $P \setminus \qty{w, u, \bar{u}, t}$ to $\bar{w}$, which is coherent and has a length of less than $k$. Since $w$ cannot be a tip created during $\setlower(k)$, \cref{prop:at_most_k} shows that $\lv(w) < k$. However, the existence of an edge $(\bar{u}, w)$ implies $w \ge \rho(u) = t$ by \cref{lem:search_order_grow}, which contradicts $\lv(w) < k$ and $\lv(t) = k$.
  
  Now, it remains to show that the tip $t$ satisfies (LV3). Given $\lv(t) = \lv(v) = k$, we need to prove $\dist(b) \ge k$. Assume $\dist(b) < k$ for contradiction. Let $P$ be a shortest coherent nonsingular path to $b$. The Laplace expansion decomposes $P$ into an edge $(w, b) \in E^*$ and a nonsingular path $P' \coloneqq P \setminus \qty{w, b}$ to $\bar{w}$. Furthermore, $P'$ cannot contain $\qty{u, \bar{u}}$; otherwise, since $u$ is only adjacent to $t \notin P'$, the corresponding row/column within $C^*[P']$ would be zero, contradicting its nonsingularity. Hence, $P'$ is coherent and thus $\dist(\bar{w}) < k$. Since $w$ cannot be a tip created during $\setlower(k)$, \cref{prop:at_most_k} ensures $\lv(w) = \dist(\bar{w})$. However, the existence of an edge $(w, b)$ implies $w \ge \rho(b) = v$ by \cref{lem:search_order_grow}, which contradicts $\lv(w) < k$ and $\lv(v) = k$.
\end{proof}

\subsubsection{Upper Vertices}
\label{sec:setupper_minimality}

This section proves that $\setupper(k, \kmax)$ satisfies \cref{cl:tips,cl:upper} under the inductive assumption that all preceding subroutines preserve (LV1)--(LV3). Roughly speaking, the following lemma utilizes a variant of the technique in \cref{sec:complete_approach} to decompose a coherent nonsingular path to $v$ into an edge $(y, z)$ entering $\st(v)$ and two smaller coherent nonsingular paths $Y$ and $Z$, where $Y \setminus T$ and $Z \setminus T$ are disjoint subsets of $P \setminus T$ (see \cref{fig:yz_incoming}). This lemma is especially crucial when $v$ is a lower vertex, constituting the key step to guarantee the existence of an incoming edge on which the $\blossom$ operation labels $\bar{v}$ (see \cref{rem:y_found_incoming}). The other case where $v$ is an upper vertex is used for verifying \cref{cl:tips}.

\begin{lemma}
  \label{lem:yz_decomp}
  Let $P$ be a coherent nonsingular path from $s_0$ to $v \in V^* \setminus \qty{s_0, q, \bar{q}}$. Suppose the upper vertex of $\qty{v, \bar{v}}$, denoted by $v^+$, is unlabeled. Then, there exist an edge $(y, z)$ satisfying $y \in V^* \setminus \st(v^+)$ and $z \in \st(v^+)$, and vertex sets $Y$ and $Z$ such that
  \begin{enumerate*}[label=(\roman*), nosep]
    \item $Z \subseteq \st(v^+)$;
    \item $Y \setminus T$ and $Z \setminus T$ are disjoint and are both contained in $P \setminus T$;
    \item $Y$ is a coherent nonsingular path from $s_0$ to $\bar{y}$; and
    \item $Z$ is a coherent nonsingular path from $v$ to $\bar{z}$, or $Z = \emptyset$ and $v = z$.
  \end{enumerate*}
\end{lemma}
\begin{proof}
    \newcommand{\tin}{T_{\mathrm{in}}}
    \newcommand{\tout}{T_{\mathrm{out}}}
    Since $v^+$ is unlabeled, every blossom is either disjoint from or contained within $\st(v^+)$. Let $\tin \coloneqq T \cap \st(v^+)$ be the union of dummy lines within the subtree. We now prove that $G^*[\tin]$ has a unique perfect matching as follows. Since the dummy lines in $\tin$ individually serve as edges of a perfect matching, we only need to verify its uniqueness. Choose a minimal blossom $H$ among the blossoms that introduce the dummy lines in $\tin$, and let $\qty{b, t}$ be the dummy line introduced by $H$. Since $t$ is only adjacent to vertices within $H$, it must be matched with $b$ in any perfect matching. Hence, an inductive argument from smaller to larger blossoms along the laminar structure establishes the uniqueness of a perfect matching. Since we have shown that $G^*[\tin]$ admits a unique perfect matching, it immediately follows that $C^*[\tin]$ is nonsingular.

    As in the former part of the proof of \cref{lem:nonsingular_invariance}, an induction from larger to smaller blossoms proves that $P \cup (T \setminus \qty{v, \bar{v}})$ is a coherent nonsingular path. Hence, we assume $P \supseteq T \setminus \qty{v, \bar{v}}$ without loss of generality. We now consider pivoting around $\tin$ on $C^*$ to obtain $\widehat{C}^*$. We similarly represent the resulting objects by adding a hat, such as $\widehat{G}^*$ and $\widehat{B}^*$. Then, the nonsingularity of $C^*[P]$ leads to that of $\widehat{C}^*[\widehat{P}]$ for $\widehat{P} \coloneqq P \symdiff \tin$. Note that $\widehat{P} \cap (\tin \setminus \qty{v, \bar{v}}) = \emptyset$ and $\widehat{P} \supseteq T \setminus \tin \eqqcolon \tout$. Furthermore, the endpoint of $\widehat{P}$ other than $s_0$, denoted by $\hat{v}$, remains $v$ if $v \notin T$ and changes to $\bar{v}$ if $v \in T$.
    
    In what follows, we apply a variant of the decomposition technique in \cref{sec:complete_approach} to the nonsingular matrix $\widehat{C}^*[\widehat{P}]$. As its determinant is nonzero, we have
    \begin{equation}
        \label{eq:nonsingular_path_det_nonzero_x}
        \sum_{\sigma \in \symgrp_{m}} \qty(\sgn \sigma \prod_{i=1}^{m} \qty(\widehat{C}^*[\widehat{P}])_{i, \sigma(i)}) \neq 0,
    \end{equation}
    where $m$ denotes the order of $\widehat{C}^*[\widehat{P}]$. Let $f(\sigma)$ denote the corresponding term in Eq.~\eqref{eq:nonsingular_path_det_nonzero_x}. For each $\sigma$ with $f(\sigma) \neq 0$, the corresponding edge set $M_\sigma$ constitutes a perfect matching on $\widehat{G}^*[\widehat{P}]$. Likewise, the set $\widehat{L}$ of lines contained within $\widehat{P}$ forms a perfect matching on $\widehat{P} \setminus \qty{s_0, \hat{v}}$. Consequently, the graph $G_{\sigma} \coloneqq (\widehat{P}, M_\sigma \uplus \widehat{L})$, which may contain parallel edges, consists of an alternating path between $s_0$ and $\hat{v}$ and a collection of alternating cycles (possibly of length two).

  Recall that the collection of all partitions of $\widehat{P}$, denoted by $\Pi(\widehat{P})$, forms a lattice $(\Pi(\widehat{P}), \le)$ under the refinement order $\le$. Let $\Gamma \colon \symgrp_{m} \to \Pi(\widehat{P})$ be a map such that $\Gamma(\sigma)$ is the collection of all the connected components of $G_\sigma$. We then merge the components that contain any element of $\qty{s_0} \cup \tout$ into a single component. In the lattice formulation, this is equivalent to taking the join $\Gamma(\sigma) \vee \calr$ with the partition $\calr \in \Pi(\widehat{P})$ defined as 
  \begin{equation}
      \calr \coloneqq \qty{\qty{s_0} \cup \tout} \cup \qty{ \qty{x} \relmiddle| x \in \widehat{P} \setminus (\qty{s_0} \cup \tout)},
  \end{equation}
  i.e., the partition containing $\qty{s_0} \cup \tout$ as a component and the remaining elements as singletons.
  
  Then, as a refinement of $\Gamma(\sigma) \vee \calr$, we define $\tilde{\Gamma}(\sigma)$ by isolating the subpath from $\hat{v}$ to $y$ from the rest of its component, where $y$ is the first vertex outside $\st(v^+)$ encountered when traversing the path from $\hat{v}$ on $G_\sigma$ (see \cref{fig:gamma_tilde}). Note that the other subpath between $s_0$ and $\bar{y}$ is nonempty because $v \notin \qty{s_0, q, \bar{q}}$.

  Grouping the terms by $\tilde{\Gamma}$, we obtain
  \begin{equation}
    \det \widehat{C}^*[\widehat{P}] = \sum_{\sigma \in \symgrp_m} f(\sigma) = \sum_{\calp \in \Pi(\widehat{P})} \qty(\sum_{\sigma \in \tilde{\Gamma}^{-1}(\calp)} f(\sigma)).
  \end{equation}
  Since this is nonzero, there exists a partition $\calp \in \Pi(\widehat{P})$ such that $\sum_{\sigma \in \tilde{\Gamma}^{-1}(\calp)} f(\sigma) \neq 0$. Here, we choose $\calp$ to be minimal with respect to the refinement order $\le$; that is, $\sum_{\sigma \in \tilde{\Gamma}^{-1}(\calq)} f(\sigma) = 0$ for every $\calq \lneq \calp$. Summing up these terms for all the refinements of $\calp$ yields
  \begin{equation}
    \label{eq:tilde_refinement_sum_x}
    \sum_{\calq \in \Pi(\widehat{P}) \colon \calq \le \calp} \qty(\sum_{\sigma \in \tilde{\Gamma}^{-1}(\calq)} f(\sigma)) \neq 0.
  \end{equation}

  We now prove \cref{cl:gamma_tilde_x}, which ensures that Eq.~\eqref{eq:tilde_refinement_sum_x} equals the determinant of a block-diagonal matrix $\widehat{C}_\calp^*[\widehat{P}]$ with each block being the square submatrix $\widehat{C}^*[X]$ for $X \in \calp$.
  This confirms the nonsingularity of each block, thereby decomposing $\widehat{P}$ into smaller nonsingular components. Specifically, since there exists $\sigma \in \tilde{\Gamma}^{-1}(\calp)$ with $f(\sigma) \neq 0$, the partition $\calp$ consists of an alternating path $P_1$ from $\hat{v}$ to $y$, the union $P_2$ of an alternating path from $s_0$ to $\bar{y}$ and the alternating cycles containing dummy lines, and the remaining alternating cycles without dummy lines. Here, $P_1$ satisfies $P_1 \setminus \qty{y} \subseteq \st(v^+)$ by the choice of $y$, which further implies that $P_2 \cup \qty{y} \supseteq \tout$. These nonsingular blocks $\widehat{C}^*[P_1]$ and $\widehat{C}^*[P_2]$ are subsequently used to construct the desired nonsingular paths $Y$ and $Z$.
  \begin{claim}
    \label{cl:gamma_tilde_x}
    A permutation $\sigma \in \symgrp_m$ is included in the above summation, i.e, $\tilde{\Gamma}(\sigma) \le \calp$, if and only if no edge of $M_\sigma$ bridges two different components of $\calp$.
  \end{claim}
  \begin{proof}
    First, suppose that an edge of $M_\sigma$ bridges two different components of $\calp$. Then, some component of $\tilde{\Gamma}(\sigma)$ contains this edge by definition, and hence it cannot be a refinement of $\calp$. Next, suppose that no edge of $M_\sigma$ bridges two different components of $\calp$. Using $\calp' \coloneqq \calp \setminus \qty{P_1, P_2} \cup \qty{P_1 \cup P_2}$, we can easily verify that $\Gamma(\sigma) \le \calp'$. Given $\calr \le \calp'$, we have $\Gamma(\sigma) \vee \calr \le \calp'$. Taking the meet with $\calp_1 \coloneqq \qty{P_1, \widehat{P} \setminus P_1}$ on both sides yields $(\Gamma(\sigma) \vee \calr) \wedge \calp_1 \le \calp' \wedge \calp_1$, where the right side is obviously $\calp$. Hence, it suffices to show that $(\Gamma(\sigma) \vee \calr) \wedge \calp_1 = \tilde{\Gamma}(\sigma)$. Since no edge of $M_\sigma$ bridges $P_1$ and another component, $\qty{y, \bar{y}}$ is the only edge in $G_\sigma$ that bridges $P_1$ and another component. Hence, when traversing the unique path of $\Gamma(\sigma)$ from $\hat{v}$, the first vertex outside $\st(v^+)$ is $y$ and all preceding vertices are contained in $P_1$. Since no other edges of $G_\sigma$ cross $P_1$, we obtain $\tilde{\Gamma}(\sigma) = (\Gamma(\sigma) \vee \calr) \wedge \calp_1$, which finishes the proof of the claim.
  \end{proof}

  Having shown the nonsingularity of $\widehat{C}^*[P_1]$ and $\widehat{C}^*[P_2]$, we now translate these properties to those within $C^*$ by pivoting around $\tin$ again. Consequently, $P_1 \symdiff \tin$ and $P_2 \symdiff \tin$ are nonsingular paths with respect to $C^*$, where the endpoints of the former are now $v$ and $y$ while those of the latter remain $s_0$ and $\bar{y}$.
  
  We now define $Y \coloneqq P_2 \symdiff \tin$ and show that this is coherent. Observe that $P_1$ and $P_2$ are included in $\hat{P}$ and thus are disjoint from $\tin \setminus \qty{v, \bar{v}}$. Since $P_2$ does not contain either of $\qty{v, \bar{v}}$, we have $Y = P_2 \symdiff \tin = P_2 \cup \tin \supseteq \tin$. Given $P_2 \cup \qty{y} \supseteq \tout$, it also follows that $Y \cup \qty{y} \supseteq \tout$. Consequently, $Y$ is the desired coherent nonsingular path between $s_0$ and $\bar{y}$. 

  To obtain $Z$, we apply the Laplace expansion to the row/column of $y$ within $C^*[P_1 \symdiff \tin]$. As the case $v = z$ can be handled by setting $Z = \emptyset$, assume $v \neq z$. Then, the expansion yields an edge $(y, z)$ satisfying $y \in V^* \setminus \st(v^+)$ and $z \in \st(v^+)$, and a nonsingular path $Z$ between $v$ and $\bar{z}$. Since $P_1 \setminus \qty{y} \subseteq \st(v^+)$, we have $Z \subseteq \st(v^+)$. Moreover, since we can prove $Z \cup \qty{\bar{v}, z} \supseteq \tin$, the path is coherent. This completes the proof of the lemma.
\end{proof}

\begin{figure}
  \centering
  \begin{minipage}[b]{0.47\columnwidth}
    \includegraphics[width=\columnwidth]{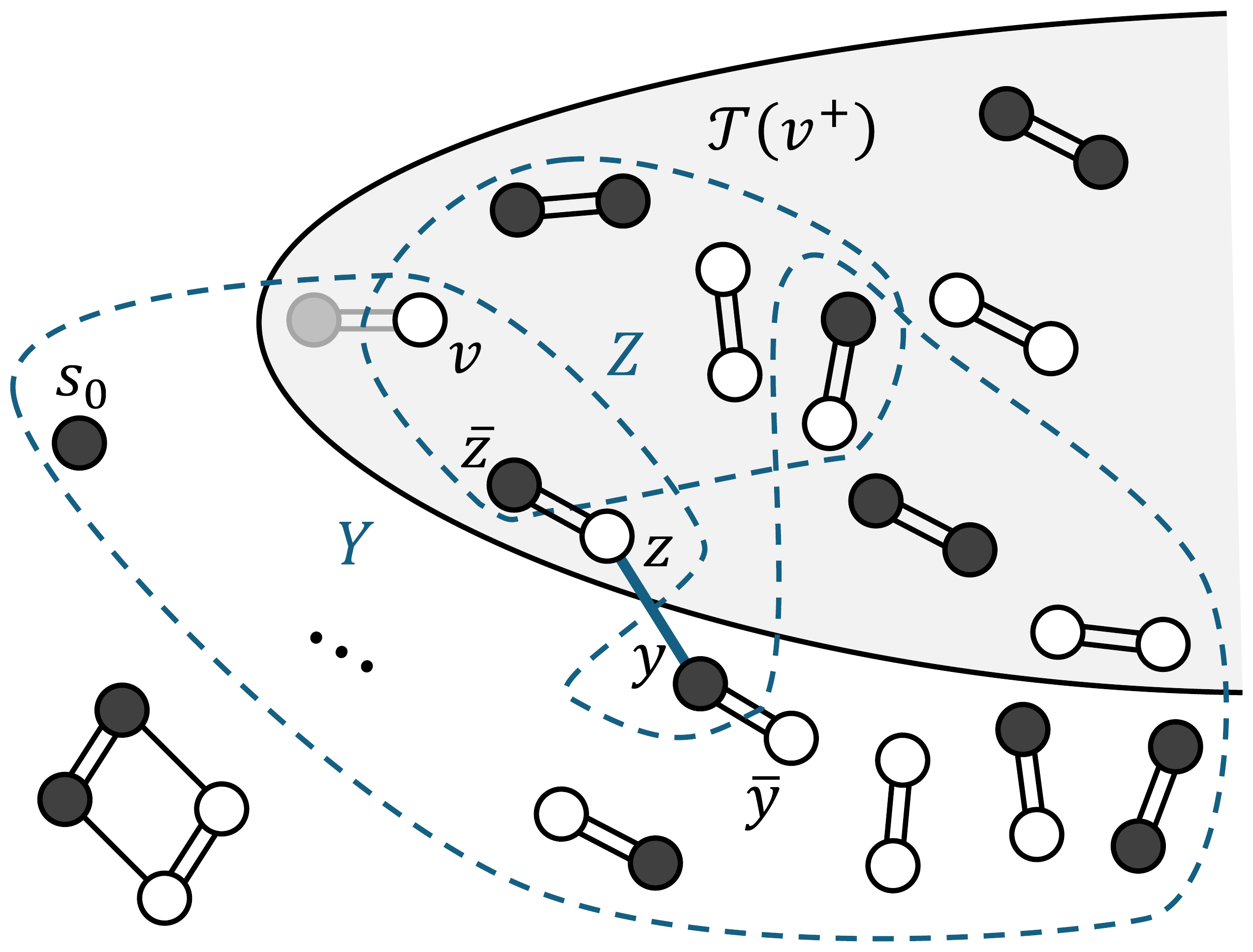}
    \subcaption{Example of $Y$, $Z$, and $(y, z)$.}
    \label{fig:yz_incoming}
  \end{minipage}
  \hfill
  \begin{minipage}[b]{0.47\columnwidth}
    \includegraphics[width=\columnwidth]{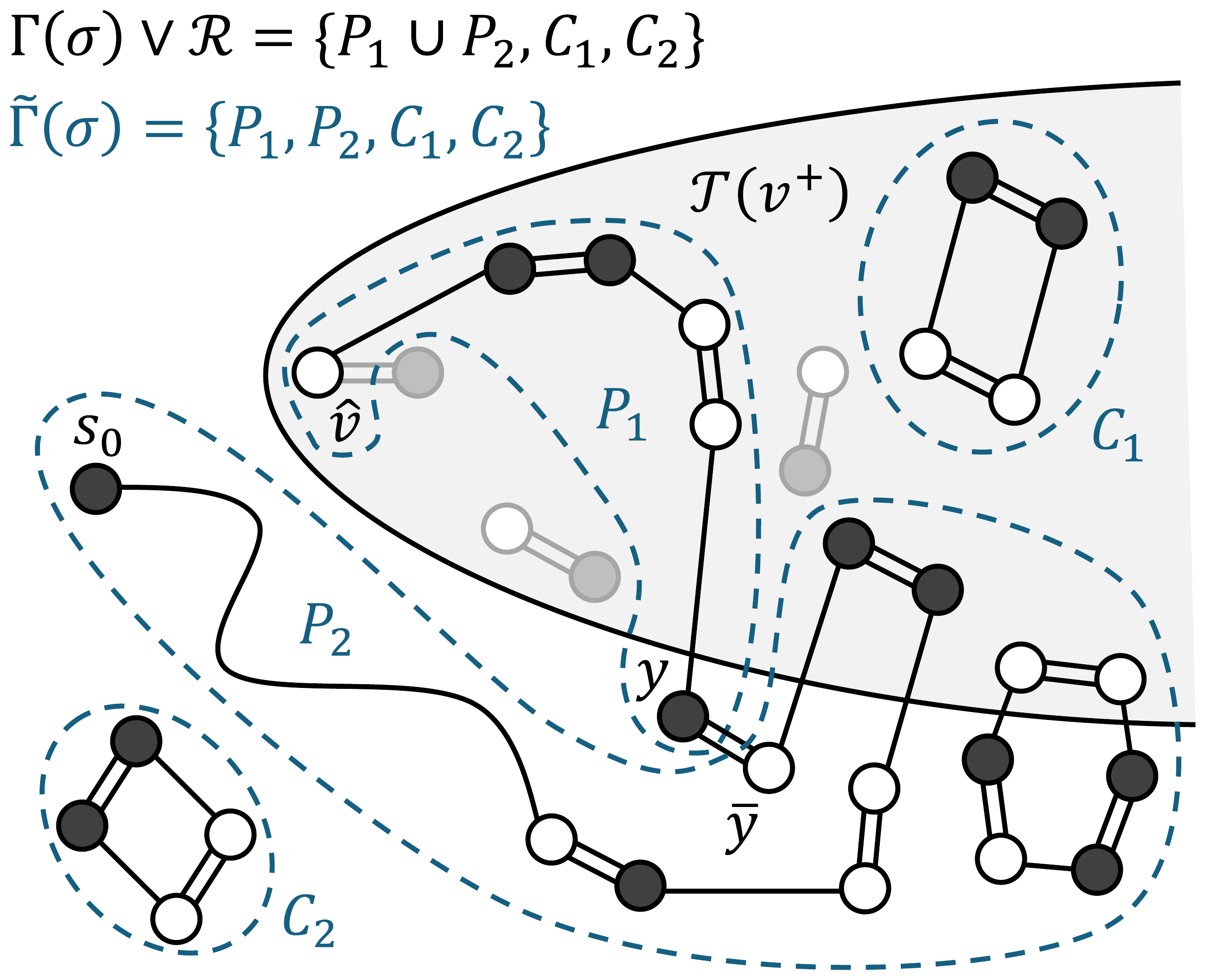}
    \subcaption{Illustration of $\Gamma(\sigma) \vee \calr$ and $\tilde{\Gamma}(\sigma)$.}
    \label{fig:gamma_tilde}
  \end{minipage}
  \caption{Illustration of \cref{lem:yz_decomp} and its proof. Double bonds represent lines, while single bonds in \subref{fig:gamma_tilde} represent $M_\sigma$. \subref{fig:yz_incoming} Black and white vertices represent those in $B^*$ and $V^* \setminus B^*$, with vertices outside $P$ grayed out. \subref{fig:gamma_tilde} Black and white vertices represent those in $\widehat{B}^*$ and $V^* \setminus \widehat{B}^*$, with vertices outside $\widehat{P}$ grayed out.}
\end{figure}

\begin{remark}
    \label{rem:y_found_incoming}
    Suppose \cref{lem:yz_decomp} is applied to a lower vertex $v$. If $y \in V(\st)$, the edge $(y, z)$ is an incoming edge of $\st(v)$. This is because $Z$ cannot contain $\bar{v}$ and thus $z \neq \bar{v}$, which also rules out the possibility that $(y, z) \in E(\st)$.
\end{remark}

Since \cref{lem:yz_decomp} for a lower vertex $v$ provides an edge $(y, z)$, which is an incoming edge of $\st(v)$ under the assumption that $y \in V(\st)$, it remains to evaluate the levels of these endpoints via $\dist(\bar{y})$ and $\dist(\bar{z})$. These level values are also crucial to ensure $y \in V(\st)$.

However, the above lemma only provides a coherent nonsingular path between $v$ and $\bar{z}$, whereas we need one between $s_0$ and $\bar{z}$ to evaluate $\dist(\bar{z})$. The following lemma addresses this gap by concatenating the coherent nonsingular path with a search path.
\begin{lemma}
  \label{lem:search_path_extension}
  For any unlabeled upper vertex $v$ and a coherent nonsingular path $Z$ from $\bar{v}$ to $\bar{z} \in \st(\bar{v})$ with $Z \subseteq \st(\bar{v})$, the set $Z' \coloneqq P(\rho(v)) \cup \qty{v} \cup Z$ is a coherent nonsingular path from $s_0$ to $\bar{z}$.
\end{lemma}
\begin{proof}
  Since the coherence is straightforward, we focus on the nonsingularity of $C^*[Z']$.
  \cref{lem:ell_disjoint_decomp} decomposes $P(\rho(v))$ as $L(x_1) y_1 \dots L(x_l)$, where $x_1, \dots, x_{l-1}$ and $x_l = \rho(v)$ are all active. Setting $y_l = v$, we also have $x_i = \rho(y_i)$ for each $i = 1, \dots, l$. Let $Z_{l+1} = Z$ and define $Z_i = L(x_i) \cup \qty{y_i} \cup Z_{i+1}$ for $i = l, \dots, 1$. We proceed by induction on $i$ in decreasing order to show that $Z_i$ is a nonsingular path. It suffices to demonstrate that any perfect matching of $G^*[Z_i]$ necessarily includes the unique perfect matching of $G^*[L(x_i) \setminus \qty{x_i}]$ and the edge $(x_i, y_i)$. Such a structure implies $\det C^*[Z_i] = \det C^*[L(x_i) \cup \qty{y_i}] \cdot \det C^*[Z_{i + 1}]$ and thereby establishes the nonsingularity of $Z' = Z_1$ by induction.

  For $j > i$, the relation $x_i < x_j = \rho(y_j)$ and \cref{lem:grow_ell_disjoint} imply that $L(x_i)$ is not adjacent to $y_j$. Similarly, $L(x_i)$ is not adjacent to $L(x_j)$ by the property of decomposition. The condition $x_i \le \rho(v)$ and \cref{lem:grow_ell_disjoint} also ensure that no vertex of $L(x_i)$ is adjacent to $\bar{v}$. Finally, no vertex of $L(x_i)$ is adjacent to $\st(\bar{v}) \setminus \qty{\bar{v}} \supseteq Z \setminus \qty{\bar{v}}$ because every upper vertex $u$ of $\st(\bar{v}) \setminus \qty{\bar{v}}$ satisfies $\rho(v) < \rho(u)$ by \cref{lem:subtree_rho}, and \cref{lem:grow_ell_disjoint} prevents adjacency between $L(x_i)$ and $\qty{u, \bar{u}}$. Consequently, no vertex of $L(x_i)$ is adjacent to $Z_{i + 1}$. Hence, any perfect matching of $G^*[Z_i]$ must contain an edge between $L(x_i)$ and $y_i$, which is necessarily $(x_i, y_i)$ by \cref{lem:grow_ell_disjoint}. Since no edge bridges $L(x_i) \setminus \qty{x_i}$ and $Z_{i + 1}$, the remaining edges constitute the unique perfect matching of $G^*[L(x_i) \setminus \qty{x_i}]$ and a perfect matching of $G^*[Z_{i + 1}]$, as required.
\end{proof}

We also provide a technical lemma that is useful for characterizing nonsingular paths consisting solely of dummy lines.
\begin{lemma}
  \label{lem:dummy_line_alternating_path}
  Suppose a vertex set $P \subseteq T$ and its two elements $x, y \in P$ satisfy the following conditions:
  \begin{enumerate*}[label=(\arabic*)]
    \item $x$ is a tip,
    \item $P \setminus \qty{x, y}$ is a union of dummy lines, and
    \item $G^*[P]$ has a perfect matching.
  \end{enumerate*}
  Then, $y$ is a bud and $K(x) = K(y)$.
\end{lemma}
\begin{proof}
  Since the statements are trivial if $\qty{x, y}$ forms a line, assume otherwise in the following. Then, combining a perfect matching of $G^*[P]$ with dummy lines within $P$ yields an alternating path $P'$ between $x$ and $y$. 
  
  Letting $v_1 v_2 \dots v_{2k}$ be the sequence of vertices in $P'$, where $v_1 = x$ and $v_{2k} = y$, we now inductively prove that $v_{2i - 1}$ is a tip and $v_{2i}$ is a bud. Suppose $v_{2i - 1}$ is a tip. Since each tip is adjacent only to vertices within its blossom by (BT2), the laminar structure of blossoms shows that no edge connects two tips. Hence, as $v_{2i}$ is adjacent to $v_{2i - 1}$, it follows that $v_{2i}$ is a bud. As $\qty{v_{2i}, v_{2i + 1}}$ forms a line, this in turn shows that $v_{2i + 1}$ is a tip. This inductive argument shows that $y$ is a bud. Since each tip is adjacent only to vertices within the blossom, a similar induction verifies $K(v_{2i - 1}) = K(v_{2i})$ and $K(v_{2i}) = K(v_{2i + 1})$, implying $K(x) = K(y)$.
\end{proof}

The following lemma is crucial for establishing the existence of an incoming edge whose endpoints are labeled.
\begin{lemma}
  \label{prop:mie_existence}
  Let $v$ be any unlabeled upper vertex such that $\dist(\bar{v})$ is finite. Suppose $\setlower(k)$ is complete for $k$ with $2k + 1 = \lv(\rho(v)) + \dist(\bar{v})$, and (LV1)--(LV3) are currently preserved. Then, there exists an incoming edge $(y, z)$ of $\st(\bar{v})$ with $z \in \st(\bar{v})$ such that $\dist(\bar{y})$ and $\dist(\bar{z})$ satisfy the following inequalities:
  \begin{enumerate*}[label=(\roman*)]
    \item $\dist(\bar{y}) < \dist(\bar{v})$,
    \item $\dist(\bar{z}) \le \dist(\bar{v})$, and
    \item $\dist(\bar{y}) + \dist(\bar{z}) \le \lv(\rho(v)) + \dist(\bar{v})$.
  \end{enumerate*}
\end{lemma}
\begin{proof}
  Define $Y$, $Z$, and $(y, z)$ by applying \cref{lem:yz_decomp} to a shortest coherent nonsingular path $P$ to $\bar{v}$. We also define $Z' \coloneqq P(\rho(v)) \cup \qty{v} \cup Z$, which is a coherent nonsingular path to $\bar{z}$ by \cref{lem:search_path_extension}; even if $Z = \emptyset$ and $\bar{v} = z$ hold, this is a coherent nonsingular path to $\bar{z}$ by \cref{lem:path_coherence,cor:search_path}. As $Y$ and $Z'$ are coherent nonsingular paths to $\bar{y}$ and $\bar{z}$, we only need to evaluate the lengths of these paths. In this proof, let $\wt(P)$ denote the length of a coherent nonsingular path $P$.

  We first verify inequality (i). As $Y \setminus T \subseteq P \setminus T$ holds, we have $\wt(Y) \le \wt(P)$. Given that $\dist(\bar{y}) \le \wt(Y)$ and $\dist(\bar{v}) = \wt(P)$ hold, it suffices to show that the inequality $\wt(Y) \le \wt(P)$ is strict. If $\qty{v, \bar{v}}$ is an original line, it intersects $P$ but not $Y$, making the inequality strict. Now, suppose that it is a dummy line. Since $\bar{v}$ must be a tip, $\bar{v}$ cannot be adjacent to $y \notin \st(v)$, leading to $\bar{v} \neq z$ and $Z \neq \emptyset$. Assume for contradiction that $Z$ consists exclusively of dummy lines. Applying \cref{lem:dummy_line_alternating_path} to $Z$ shows that the other endpoint $\bar{z}$ is a bud. Hence, $z$ is a tip, which contradicts the fact that tips within $\st(v)$ cannot be adjacent to $y \notin \st(v)$. Therefore, $Z$ intersects at least one original line, which shows $\wt(Y) < \wt(P)$.

  Inequality (iii) can be verified as follows:
  \begin{align*}
    \dist(\bar{y}) + \dist(\bar{z})
     & \le \wt(Y) + \wt(Z')                \\
     & \le \wt(P) + \wt(P(\rho(v)))        \\
     & = \lv(\rho(v)) + \dist(\bar{v}).
  \end{align*}

  Hence, it remains to verify inequality (ii). Given that inequality (iii) holds, it suffices to verify that $\dist(\bar{y}) \ge \lv(\rho(v))$ holds. Assume $\dist(\bar{y}) < \lv(\rho(v))$ for contradiction. Since $\lv(\rho(v))$ is at most the current search level, the assumptions allow for applying \cref{prop:at_most_k} to $y$, showing that $\lv(y)$ is assigned and satisfies $\lv(y) < \lv(\rho(v))$. On the other hand, the upper vertex $z^+$ of $\qty{z, \bar{z}}$ satisfies $\rho(v) \le \rho(z^+)$ by \cref{lem:subtree_rho}. Moreover, the existence of an edge $(y, z)$ and the contrapositive of \cref{lem:search_order_grow} together imply $\rho(z^+) \le y$. Consequently, we obtain $\rho(v) \le y$, which contradicts $\lv(y) < \lv(\rho(v))$.

  Finally, we prove $y \in V(\st)$, which implies that $(y, z)$ is an incoming edge of $\st(\bar{v})$ by \cref{rem:y_found_incoming}. By inequality (iii), $\dist(\bar{y})$ or $\dist(\bar{z})$ is at most $k$. If $\dist(\bar{y}) \le k$, then \cref{prop:at_most_k} shows that $\lv(y) = \dist(\bar{y})$ is assigned, implying $y \in V(\st)$. If $\dist(\bar{z}) \le k$, then \cref{prop:at_most_k} shows that $\lv(z) = \dist(\bar{z})$ is assigned. As this level is at most $k$, $\grow(z)$ during $\setlower(k)$ or earlier must have detected $\qty{y, \bar{y}}$, which implies $y \in V(\st)$.
\end{proof}

\begin{remark}
    \label{rem:weakened_assumption}
    The assumption of (LV2) at $v$ can be dropped in \cref{prop:mie_existence}, because (LV1)--(LV3) are required solely to apply \cref{prop:at_most_k} to $y$ and $z$, both of which are distinct from $v$. Removing this assumption avoids a circular argument when establishing (LV2) in the proof of \cref{prop:setupper_dummy_line}.
\end{remark}

Finally, we establish the correctness of the subroutine $\setupper(k, \kmax)$ by the following two lemmas. We first verify \cref{cl:tips} for this subroutine in \cref{prop:setupper_dummy_line}, ensuring that (LV1)--(LV3) are maintained as long as the subroutine remains incomplete. Next, we demonstrate \cref{cl:upper} for the subroutine in \cref{prop:setupper_original_line}, guaranteeing that (LV1)--(LV3) are also preserved when the subroutine terminates.
\begin{lemma}
  \label{prop:setupper_dummy_line}
  For each dummy line introduced during $\setupper(k, \kmax)$, the tip $t$ satisfies (LV3) and the bud $b$ does not violate (LV2) at the moment it is created.
\end{lemma}
\begin{proof}
  We proceed by induction to show that each insertion of dummy lines preserves (LV1)--(LV3), as long as the subroutine remains incomplete. Let $u$ and $v$ be the vertices such that $\blossom(u, v)$ introduces $b$ and $t$. We denote the resulting blossom by $H$, and define $c$ and $g$ as in the description of the algorithm. Since $\lv(v) = 2k + 1 - \kmax \le k$, we have $\lv(c) = \lv(t) \le \lv(v) \le k$.

  We first focus on (LV3). Assume $\dist(b) < \lv(c)$ for contradiction. Let $P$ be a shortest coherent nonsingular path to $b$. Applying \cref{lem:yz_decomp} to $P$ provides a coherent nonsingular path $Y$ to $\bar{y}$ and an edge $(z, y)$ satisfying $z \in \st(b)$ and $y \notin \st(b)$. Since the length of $Y$ is at most $\dist(b) < \lv(c) \le k$, we have $\lv(y) \le \dist(b) < \lv(c)$ by \cref{prop:at_most_k}. On the other hand, since $z \in \st(b)$, we obtain $c = \rho(b) \le \rho(z^+)$ by \cref{lem:subtree_rho}, where $z^+$ denotes the upper vertex of $\qty{z, \bar{z}}$. Furthermore, the existence of an edge $(z, y)$ implies $\rho(z^+) \le y$ by \cref{lem:search_order_grow}, which can be combined with the previous inequality to obtain $c \le y$, a contradiction.
  
  To prove (LV2) for the bud $b$, it suffices to prove $\lv(c) + \dist(t) \ge 2k + 1$; the statement is straightforward if this inequality is strict, and otherwise it follows from $\lv(c) \le \lv(v)$. Thus, assume for contradiction that $\lv(c) + \dist(t) \le 2k - 1$, as the left side is always odd from \cref{rem:theorem_odd}. Then, \cref{prop:mie_existence,rem:weakened_assumption} ensure the existence of an incoming edge $(y, z)$ of $\st(t)$ satisfying $\dist(\bar{y})$ + $\dist(\bar{z}) \le 2k - 1$. Consequently, (LV1)--(LV3) inductively prove that $y$ and $z$ are both labeled and satisfy $\lv(y) + \lv(z) \le 2k - 1$. However, since $(y, z)$ precedes $(u, v)$ in the order of $\setupper$, \cref{prop:incoming_edge_label,lem:effectively_all_executed} ensure that $b$ must be already labeled, which is a contradiction.
\end{proof}

\begin{lemma}
  \label{prop:setupper_original_line}
  The subroutine $\setupper(k, \kmax)$ preserves (LV1)--(LV3).
\end{lemma}
\begin{proof}
  Let $v$ be any upper vertex satisfying the conditions in (LV2). If $v$ is labeled prior to the start of $\setupper(k, \kmax)$, (LV2) follows immediately. Thus, we assume that $v$ is either unlabeled at the beginning of $\setupper(k, \kmax)$ or introduced as a bud during its execution. Then, we obtain an incoming edge $(y, z)$ of $\st(\bar{v})$ satisfying inequalities (i)--(iii) by applying \cref{prop:mie_existence} at the corresponding moment---either at the start of $\setupper(k, \kmax)$ if $v$ already exists, or immediately after $v$ is inserted during the subroutine. If $y$ and $z$ are labeled at this point, it follows from these inequalities and (LV3) that $\blossom(y, z)$ must be invoked before the completion of $\setupper(k, \kmax)$, which then ensures that $v$ is labeled by \cref{prop:incoming_edge_label,lem:effectively_all_executed}. Note that, once $y$ and $z$ are labeled, \cref{cl:new_edge_rho_c} rules out the possibility that $(y, z)$ is deleted by subsequent pivoting operations. To prove that $y$ and $z$ are labeled, we divide the argument into two cases: (a)~$\dist(\bar{z}) \le k$ and (b)~$\dist(\bar{y}) \le k$; these are exhaustive since the values are integers and their sum is at most $2k + 1$.

  \textbf{Case (a).}
  Given that $\dist(\bar{z}) \le k$, \cref{prop:at_most_k} guarantees that $z$ is labeled. The case where $y$ is the lower vertex is trivial, so we focus on the case where $y$ is the upper vertex. Setting $\kmax' = \dist(\bar{y})$, we have $\kmax' = \dist(\bar{y}) < \dist(\bar{v}) = \kmax$ by inequality (i). Similarly, setting $2k' + 1 = \lv(\rho(y)) + \dist(\bar{y})$, we obtain the following:
  \begin{align*}
    2k' + 1
     & = \lv(\rho(y)) + \dist(\bar{y})                                                               \\
     & \le \lv(z) + \dist(\bar{y})       & \text{($(y,z) \in E^*$ and \cref{lem:search_order_grow})} \\
     & = \dist(\bar{y}) + \dist(\bar{z})                                                             \\
     & \le 2k + 1.
  \end{align*}
  Therefore, (LV2) inductively implies that $y$ is labeled before $\setupper(k', \kmax')$ completes. As this precedes $\setupper(k, \kmax)$, the statement follows.

  \textbf{Case (b).}
  Given that $\dist(\bar{y}) \le k$, \cref{prop:at_most_k} guarantees that $y$ is labeled. The case where $z$ is a lower vertex is trivial, and the case where $\dist(\bar{z}) < \dist(\bar{v})$ follows by a symmetric argument. Hence, we now assume that $z$ is an upper vertex and $\dist(\bar{z}) = \dist(\bar{v})$. In this case, we have the following:
  \begin{align*}
    \lv(\rho(v))
     & \le \lv(\rho(z))                                   & \text{(\cref{lem:subtree_rho})}          \\
     & \le \lv(y)                                         & \text{(\cref{lem:search_order_grow})}    \\
     & = \dist(\bar{y})                                                                              \\
     & \le \lv(\rho(v)) + \dist(\bar{v}) - \dist(\bar{z}) & \text{(\cref{prop:mie_existence}~(iii))} \\
     & = \lv(\rho(v)),
  \end{align*}
  which yields $\lv(\rho(v)) = \lv(\rho(z))$. Since \cref{lem:subtree_valve} ensures that $P(\rho(z))$ decomposes as $P(\rho(v))vP(\rho(z)|v)$, the previous equation implies that the set $X \coloneqq P(\rho(z)|v) \setminus \qty{\rho(z)}$, which may be empty, consists solely of dummy lines. If $X$ is nonempty, $G^*[X]$ has a perfect matching by \cref{prop:path_property}. As $\bar{v}$ is a tip, \cref{lem:dummy_line_alternating_path} ensures that $K(v)$ = $K({\rho(z)})$ and that $\rho(z)$ is a tip, which further imply $K(v) = K(z)$. This property also holds in the case where $X = \emptyset$, from which $\bar{v} = \rho(z)$ follows immediately. Because $v$ is unlabeled and $K(z) = K(v)$, it follows from (MB1) that $z$ is labeled unless $z = v$; however, this is impossible since $z \in \st(\bar{v})$ and $v \notin \st(\bar{v})$. 
\end{proof}

\subsection{Optimality}
\label{sec:lmp_optimality}

This section is devoted to proving that the current independent parity set $J$ is maximum if the algorithm does not detect any augmenting paths. We relate the maximum cardinality $\nu$ and the current cardinality $\abs{J}$ via the \skewsymmetric{} matrix $\Phi^*$ of size $(n + r + \abs{T})$ defined as follows:
\begin{equation}
  \Phi^* \coloneqq
  \begin{bNiceArray}{cccccc}[first-row,first-col,margin,hvlines,corners]
                    & T \cap B^*                     & S & J  & J'                  & V \setminus J    & T \setminus B^* \\
    T \cap B^*      & \Block{3-3}{O}                 &   &    & \Block{2-1}{O}      & \Block{3-2}{C^*} &                 \\
    S               &                                &   &    &                     &                  &                 \\
    J               &                                &   &    & I                   &                  &                 \\
    J'              & \Block{1-2}{O}                 &   & -I & \Block{2-2}{\Delta} &                  & \Block{2-1}{O}  \\
    V \setminus J   & \Block{2-3}{-\qty(C^*)^{\top}} &   &    &                     &                  &                 \\
    T \setminus B^* &                                &   &    & \Block{1-2}{O}      &                  & O               \\
  \end{bNiceArray}.
\end{equation}

We first express the rank of $\Phi^*$ using $\nu$. Property (BT1) in \cref{prop:bt} ensures that the following transformation via row operations is possible:
\begin{equation}
  \begin{bNiceArray}{ccc}[first-col,first-row,margin,hvlines,corners]
    & J' & V \setminus J & T \setminus B^* \\
    T \cap B^* & \Block{2-1}{O} & \Block{3-2}{C^*} \\
    S &   & & \\
    J & I & & \\
  \end{bNiceArray}
  \quad \xrightarrow{\qqtext{Row operations}} \quad
  \begin{bNiceArray}{ccc}[first-col,first-row,margin,hvlines,corners]
    & J' & V \setminus J & T \setminus B^* \\
    T \cap B^* & \Block{2-1}{O} & X & I \\
    S &   & \Block{2-1}{C} & \Block{2-1}{O} \\
    J & I & & \\
  \end{bNiceArray}.
\end{equation}
Here, $X$ denotes the resulting nonzero part. Therefore, we can transform $\Phi^*$ into the following matrix $\Phi'$ via row and column operations:
\begin{equation}
  \Phi' =
  \begin{bNiceArray}{ccc|ccc}[first-row,first-col,margin,hvlines,corners]
                  & T \cap B^*                   & S                                         & J  & J'                  & V \setminus J  & T \setminus B^* \\
  T \cap B^*      & \Block{3-3}{O}               &                                           &    & \Block{2-1}{O}      & X              & I               \\
  S               &                              &                                           &    &                     & \Block{2-1}{C} & \Block{2-1}{O}  \\
  J               &                              &                                           &    & I                   &                &                 \\
  J'              & \Block{1-2}{O}               &                                           & -I & \Block{2-2}{\Delta} &                & \Block{2-1}{O} \\
  V \setminus J   & \raisebox{-2pt}{$-X^{\top}$} & \Block{1-2}{\raisebox{-2pt}{$-C^{\top}$}} &    &                     &                &                 \\
  T \setminus B^* & -I\phantom{{}^\top}          & \Block{1-2}{O}                            &    & \Block{1-2}{O}      &                & O               \\
  \end{bNiceArray}.
\end{equation}
Note that $\Phi$ appears at the center as a submatrix. Since $\rank \Phi = \nu + n$ holds by \cref{prop:rank_lmp}, we have
\begin{equation}
  \label{eq:phi_prime_nu}
  \rank \Phi^* = \rank \Phi' = \rank \Phi + \abs{T} = \nu + n + \abs{T}.
\end{equation}

We then evaluate $\rank \Phi^*$ using the current size $\abs{J}$. To this end, we first analyze the structure of the graph $G^*$ when the algorithm terminates without finding any augmenting paths. Let $U$ be the set of unlabeled vertices that have unlabeled mates. Let $W \coloneqq V^* \setminus \qty(S \cup U)$ be the union of lines whose lower level is assigned. We define a partition $\calw$ of $W$ based on $K(x)$, i.e., $\calw \coloneqq \qty{K(x) \mid x \in W}$.
Let $u(W)$ denote the unique unlabeled vertex in $W \in \calw$.
By excluding $u(\calw) \coloneqq \qty{u(W) \mid W \in \calw}$ from the graph $G^*$, the blossoms and unreached part become disjoint. More formally, the following lemma holds.
\begin{lemma}
  \label{lem:partition_graph}
  Define a partition of $V^*$ by $\call \coloneq \qty{\qty{s} \mid s \in S} \cup \qty{U} \cup \calw$. By eliminating $u(\calw)$ and their incident edges from the graph $G^*$, no remaining edge connects two different components of $\call$.
\end{lemma}
\begin{proof}
  Assume for contradiction that such an edge $e$ exists. Since $e$ bridges different components, its endpoints are both active. If $e$ bridges $U$ and a component of $\call \setminus \qty{U}$, the endpoint that does not belong to $U$ must be labeled. However, given the negation of Line~\ref{line:while_worth_searching} in \cref{alg:sap_lmp}, this labeled endpoint cannot be connected to any vertex of $U$ by \cref{lem:effectively_all_executed}, a contradiction. Similarly, $e$ cannot bridge two distinct components of $\call \setminus \qty{U}$.
\end{proof}

We then examine the support graph $G_{\Phi}^* = (V_{\Phi}^*, E_{\Phi}^*)$ of $\Phi^*$. Crucially, $G^*_{\Phi}$ contains $G^*$ as its subgraph (see \cref{fig:matroid_parity_pfaffian_graph}).
We define $\kappa$ as the following map: given a set $X \subseteq V^*$, it attaches the copies of $X \cap J$ to $X$. More formally, we define $\kappa \colon 2^{V^*} \to 2^{V^*_{\Phi}}$ by $\kappa \colon X \mapsto X \cup \iota(X \cap J)$, where $\iota \colon J \to J'$ is a natural bijection.
Then, the following lemma holds.
\begin{lemma}
  \label{lem:block_phi}
  Define a partition of $V^*_{\Phi}$ by $\kappa(\call) \coloneq \qty{\kappa(X) \mid X \in \call}$. By eliminating $u(\calw)$ and their incident edges from the graph $G^*_{\Phi}$, no remaining edge connects two different components of $\kappa(\call)$.
\end{lemma}
\begin{proof}
  The graph $G_\Phi^*$ is constructed from $G^*$ by adding the following:
  \begin{itemize}[nosep]
    \item the copies $J'$ of $J$,
    \item the edges given by $\Delta$ that connect the two vertices of each line in $J' \cup (V \setminus J)$, and
    \item the edges given by unit matrices between each element in $J$ and its copy in $J'$.
  \end{itemize}
  Then, the lemma follows from \cref{lem:partition_graph} and the fact that these new edges do not bridge different components of $\kappa(\call)$.
\end{proof}

Finally, we evaluate $\rank \Phi^*$ as follows. While the following discussion derives the bound in Eq.~\eqref{eq:tutte_like} by a linear-algebraic argument, it also admits a graph-theoretic proof as described in \cref{rem:graphic}.

Since removing one row or column reduces the rank by at most one, we have
\begin{equation}
  \label{eq:remove_cost}
  \rank \Phi^* - 2 \abs{u(\calw)} \le \rank \Phi^*[V^*_{\Phi} \setminus u(\calw)].
\end{equation}
After removing $u(\calw)$, \cref{lem:block_phi} ensures that the resulting matrix $\Phi^*[V^*_{\Phi} \setminus u(\calw)]$ has a block-diagonal structure. The partition of $V^*_{\Phi} \setminus u(\calw)$ can be explicitly given by
\begin{equation}
  \label{eq:partition_k_def}
  \calk = \qty{\qty{s} \mid s \in S} \cup \qty{\kappa(U)} \cup \qty{\kappa(W) \setminus \qty{u(W)} \mid W \in \calw}.
\end{equation}

Using the partition $\calk$, we have
\begin{equation}
  \rank \Phi^*[V^*_{\Phi} \setminus u(\calw)] = \sum_{X \in \calk} \rank \Phi^*[X].
\end{equation}
Since an \skewsymmetric{} matrix of odd size is singular, we have $\rank \Phi^*[X] \le \abs{X} - 1$ for odd-sized $X \in \calk$. Applying this bound to the summation yields
\begin{equation}
  \label{eq:removed_ub}
  \rank \Phi^*[V^*_{\Phi} \setminus u(\calw)] \le \abs{V^*_{\Phi} \setminus u(\calw)} - \rmodd(\calk),
\end{equation}
where $\rmodd(\calk)$ denotes the number of odd components in $\calk$.
Combining Eqs.~\eqref{eq:remove_cost} and \eqref{eq:removed_ub} gives
\begin{equation}
  \label{eq:tutte_like}
  \rank \Phi^* \le \abs{V^*_{\Phi}} + \abs{u(\calw)} - \rmodd(\calk).
\end{equation}
The description of $\calk$ given in Eq.~\eqref{eq:partition_k_def} implies that each component of $\calk$ is odd-sized except for $\kappa(U)$, which leads to $\rmodd(\calk) = \abs{S} + \abs{\calw}$. Substituting this equality, along with $\abs{u(\calw)} = \abs{\calw}$ and $\abs{V^*_\Phi} = n + r + \abs{T}$, into Eq.~\eqref{eq:tutte_like} yields
\begin{equation}
  \label{eq:graphic_ub}
  \rank \Phi^* \le \qty(n + r + \abs{T}) + \abs{\calw} - \qty(\abs{S} + \abs{\calw}) = n + r + \abs{T} - \abs{S}.
\end{equation}

Finally, combining Eqs.~\eqref{eq:phi_prime_nu} and \eqref{eq:graphic_ub} leads to
\begin{equation}
  \nu + n + \abs{T} = \rank \Phi^* \le n + r + \abs{T} - \abs{S},
\end{equation}
which shows $\nu \le r - \abs{S} = \abs{J}$. This completes the optimality proof.

\begin{remark}
  \label{rem:graphic}
  We can rephrase the derivation of Eq.~\eqref{eq:tutte_like} in a graph-theoretic manner as follows. The rank of an \skewsymmetric{} matrix is bounded above by twice the maximum size of a matching on its support graph~\cite[Remark~7.3.13]{Murota2010MatricesMatroids}. Thus, $\rank \Phi^*$ is at most twice the maximum size of a matching on $G^*_{\Phi}$. This can also be bounded as in Eq.~\eqref{eq:tutte_like} by using the easier direction of the Tutte--Berge formula and the partition given in \cref{lem:block_phi}. Notably, this alternative proof indicates that the graph structure $G^*$, without the numeric information $C^*$, suffices to certify optimality.
\end{remark}

\subsection{Implementation Details}
\label{sec:implementation_details}

In this section, we provide several implementation details that are crucial for achieving the desired time bounds. The objective of this section is to demonstrate that a single search for a shortest augmenting path completes in $\Order(nr)$ time.
\begin{lemma}
    Let $B = J \cup S$ be a base consisting of an independent parity set $J$ and a set $S$ of singletons. Given the fundamental cocircuit matrix $C$ with respect to $B$, a shortest augmenting path can be identified in $\Order(nr)$ time, or its non-existence can be determined within the same time bound.
\end{lemma}

To this end, we begin by establishing two fundamental bounds. We first evaluate the number of nodes $\abs{V^*}$. Since blossoms form a laminar family and every non-minimal blossom merges two original lines from different blossoms, the number of blossoms is $\Order(n)$, which leads to $\abs{V^*} = \Order(n)$. Therefore, the variables related to the vertices, such as $\rho$ and $\vl$, can be managed in $\Order(n)$ space. 

We next analyze the nested loops of $\setupper$. Given that $g$ and $t$ are not added to $\vl$ in each insertion of a dummy line, $\vl$ has at most $r$ elements from $B^*$ and at most $n$ elements from $V^* \setminus B^*$. Consequently, the algorithm
examines at most $nr$ pairs of vertices in the nested loop.

\paragraph{Blossoms.}

To maintain blossoms, we utilize Tarjan's algorithm~\cite{Tarjan1975EfficiencyGood} combined with a trick for optimizing the time complexity. Given that $\Order(nr)$ $\find$ operations are performed over $\Order(n)$ vertices, a naive application results in a time bound of $\Order(nr \alpha(nr, n))$, where $\alpha$ denotes the inverse Ackermann function.

While Gabow and Stallmann~\cite{GabowS1986AugmentingPath} suggested using the incremental tree set union algorithm~\cite{GabowT1985LineartimeAlgorithm} to remove this $\alpha$ factor, their description lacks details and appears to struggle with blossoms sharing the same root.\footnote{Blossoms in their framework correspond to ours without dummy lines; hence, children of a vertex can be arbitrarily grouped into blossoms.} Furthermore, the RAM model is necessitated.

Alternatively, we present a simple trick that restricts the number of vertices to $\Order(r)$, thereby reducing the factor $\alpha(nr, n)$ to $\alpha(nr, r)$, which is constant.\footnote{The overall bound $\Order(nr)$ can also be seen from the fact that the total cost of path compression is $\Order(r \log r)$.} By excluding singletons and relabeling $g$ as $b$ upon the introduction of a dummy line, the problem becomes a disjoint set union (DSU) over the original lines $L$. We partition $L$ into $L_\rmb \coloneqq \qty{\ell \in L \mid \ell \subseteq B}$ and $L_\rmn = L \setminus L_\rmb$. To implement the DSU over $L$, we construct an internal DSU on $L_\rmb$ via Tarjan's algorithm and maintain a pointer $\pi \colon L_\rmn \to L_\rmb$. Note that this internal DSU handles $\Order(r)$ elements. Let $\union_\rmb$ and $\find_\rmb$ denote the internal operations on $L_\rmb$. Given that $\union$ always merges a line from $L_\rmb$ and another from $L_\rmn$, we can define the operation $\union(\ell_\rmb, \ell_\rmn)$ for $\ell_\rmb \in L_\rmb$ and $\ell_\rmn \in L_\rmn$ as follows.
\begin{itemize}
  \item If $\pi(\ell_\rmn)$ is defined, perform $\union_\rmb(\pi(\ell_\rmn),\ell_\rmb)$.
  \item Otherwise, set $\pi(\ell_\rmn) \coloneqq \ell_\rmb$.
\end{itemize}
Correspondingly, $\find(\ell)$ for $\ell \in L$ is supported as follows.
\begin{itemize}
  \item If $\ell \in L_\rmb$, perform $\find_\rmb(\ell)$.
  \item If $\ell \in L_\rmn$ and $\pi(\ell)$ is defined, execute $\find_\rmb(\pi(\ell))$.
  \item Otherwise, return $\ell$.
\end{itemize}
Consequently, this DSU over $L$ ensures that the $\Order(nr)$ $\find$ operations take only linear time in total.

\paragraph{Search Paths.}

We next describe the management of search paths. While path recovery is associated with both backtracking during $\blossom$ operations and identifying augmenting paths, the former does not strictly require explicit recovery.

Alternatively, we employ the DSU for blossoms. By maintaining a map from the results of $\find$ operations to the unique unlabeled vertices of the components, we can assume that each $\find$ operation retrieves such a unique unlabeled vertex. Consequently, the unlabeled vertices on $P(u|c)$ and $P(v|d)$ can be efficiently collected by iteratively applying $\find$ and $\rho$. This procedure effectively performs backtracking on a tree where each maximal blossom is contracted into a node. Given that each blossom is merged immediately after being traced, the total cost of this backtracking process can be bounded by the number of blossoms, $\Order(n)$.

Hence, we only need to recover a search path when an augmenting path is detected. The following recovery process is based on \cite{GabowS1986AugmentingPath,IzumiKY2025ForgettingAlternation}. For each vertex $x$ labeled by \ref{item:detail_5}, the algorithm associates the blossom-forming edge $(u, v)$ and records whether $x$ belongs to $P(u)$ or $P(v)$.

The recovery process takes $x$ and an optional $y$ to construct $yP(x|y)$, where $yP(x|y) = P(x)$ if $y$ is null. If $y$ is specified, it must have been unlabeled at the moment $x$ was labeled. The procedure is defined recursively as follows.
\begin{itemize}
  \item As the base case, directly compute the path segment $yP(x|y)$ if $x$ is a singleton.
  \item If $x$ is associated with a blossom-forming edge $(u, v)$, where $x \in P(u)$, recursively construct $xP(u|x)$ and $yP(v|y)$, which are then concatenated to form $yP(v|y)\overline{P(u|x)}x$.
  \item Otherwise, define $P$ as follows: if $\bar{x} = y$, set $P$ to an empty path; if $\bar{x} \neq y$, recursively construct $yP(\rho(\bar{x})|y)$ and set $P$ to this path. Finally, return $P \bar{x} x$.
\end{itemize}
The validity of this process is straightforward for newly defined search paths and is preserved under the insertion of dummy lines by \cref{cor:justify_line_4}. The procedure runs in time linear in the path length and thereby completes in $\Order(n)$ time.

Alternatively, we can efficiently recover a search path excluding its tips and retired vertices in $\Order(r)$ time by simulating the behavior of the Gabow--Stallmann algorithm. Essentially, this is equivalent to constructing another search tree, where dummy lines are not inserted, $g$ is replaced with $b$, and search paths are similarly maintained.

\paragraph{Pivoting.}

Next, we address the pivoting operations utilized during the search procedure. A naive implementation increases the height of $C^*$ from $r$ to $\Order(n)$, which inevitably leads to a total time complexity of $\Omega(n^2)$. Preventing the growth of $C^*$ is therefore crucial to achieving the desired bound of $\Order(nr)$. Our method first saves a copy of $C$ and then prunes $C^*$ upon pivoting by eliminating rows and columns associated with tips and retired vertices. Since these vertices are neither involved in the search nor designated as $c$ or $g$ in subsequent pivoting operations, their removal allows us to maintain $C^*$ as an $r \times n$ matrix throughout the search procedure.

Furthermore, each pivoting operation can be optimized so that the total time complexity is bounded by $\Order(nr)$. The analysis depends on whether the bud $b$ belongs to the base; we categorize blossoms as \textit{basic} if their bud is part of the base, and \textit{nonbasic} otherwise. The following lemmas are essential for obtaining the overall complexity.
\begin{lemma}
  \label{lem:blm_count}
  The number of basic (resp.~nonbasic) blossoms that insert a dummy line is $\Order(r)$ (resp.~$\Order(n)$).
\end{lemma}
\begin{proof}
  As the total number of blossoms is $\Order(n)$, it suffices to show the bound of $\Order(r)$ on the number of basic blossoms that insert a dummy line. Since each basic blossom created by a $\blossom$ operation within $\grow$ contains a single original line from $B$, the number of such blossoms is $\Order(r)$. The number of remaining basic blossoms is likewise $\Order(r)$, because each merges two distinct original lines from $B$ when creating a dummy line.
\end{proof}
\begin{lemma}
  \label{lem:at_most_two}
  At the start of \ref{item:detail_3}, there are at most two vertices in $H$ that are adjacent to $c$.
\end{lemma}
\begin{proof}
  As the claim holds trivially if the $\blossom$ operation is called within a $\grow$ operation, we assume otherwise. By \cref{prop:kx_properties,rem:kc_is_lca}, $H$ is contained within at most two child subtrees of $c$, each containing $P(u|c)$ or $P(v|d)$. For every such subtree $\st(w)$, where $w$ is a child of $c$, \cref{lem:subtree_rho,lem:search_order_grow} together ensure that $w$ is the unique vertex in $\st(w)$ that is adjacent to $c$. Consequently, $c$ has at most two neighbors within these subtrees, and thus in $H$.
\end{proof}
\begin{lemma}
  \label{lem:pivoting_time}
  A pivoting operation forming a basic (resp.~nonbasic) blossom can be implemented in $\Order(n)$ (resp.~$\Order(r)$) time, provided that the tip and the retired vertex are removed from $C^*$ immediately after every operation.
\end{lemma}
\begin{proof}
  Since there are at most two vertices in $H$ that are adjacent to $c$ by \cref{lem:at_most_two}, the pivoting operation only involves three rows (resp.~columns) when forming a basic (resp.~nonbasic) blossom. Hence, the time bound of $\Order(n)$ follows directly for a basic blossom. To attain the $\Order(r)$ time bound for nonbasic blossoms, it suffices to avoid explicitly appending a new row to $C^*$, which corresponds to the tip after the operation. Given that this row is only used for the pivoting operation and is subsequently deleted, we can omit explicitly inserting this row, reducing the time complexity from $\Order(n)$ to $\Order(r)$.
\end{proof}
Given the bounds in \cref{lem:blm_count,lem:pivoting_time}, we obtain the claimed bound of $\Order(nr)$ for all pivoting operations during the search.

\section{Overall Complexity}
\label{sec:overall_time_bound}

Since the search procedure completes in $\Order(nr)$ time, the total time complexity of each augmentation step is dominated by the pivoting operation along an augmenting path. For a shortest augmenting path of length $\ell$, the time complexity is $\Order(n r \ell)$ or $\Order(n r \ell^{\omega - 2})$, depending on the application of fast matrix multiplication. Given the bound on $\ell$ from \cref{cor:lmp_sap_length}, we can evaluate the total time complexity as follows.

\ourtimebound*

\begin{proof}
    Recall that \cref{cor:lmp_sap_length} ensures $\ell \le 2 \abs{J} / (\nu - \abs{J}) + 1 \le 2\nu / (\nu - \abs{J})$, where $J$ is the current independent parity set and $\nu$ denotes the maximum size of an independent parity set.
   
    We first consider the analysis without fast matrix multiplication. By summing the upper bound of $\ell$ across $\abs{J} = 0, 2, \dots, \nu - 2$, the total length of the shortest augmenting paths can be bounded by $\Order(\nu \log \nu)$ via the harmonic series. Hence, the total complexity of the algorithm can be evaluated as $\Order(n r \nu \log \nu)$, or more simply, $\Order(n r^2 \log r)$.

    We then consider the application of fast matrix multiplication. Setting $j = (\nu - \abs{J}) / 2$, we can bound the total running time by $\Order(\sum_{j = 1}^{\nu / 2} n r (\nu / j)^{\omega-2})$, which can be evaluated as $\Order(n r \nu)$ via the following integration:
   \begin{equation}
       \sum_{j = 1}^{\nu / 2} \qty(\frac{\nu}{j})^{\omega-2}
       \le 
       \nu^{\omega-2} + \int_{1}^{\nu / 2} \qty(\frac{\nu}{x})^{\omega-2} \dd{x}
       = \frac{(\nu / 2^{3-\omega}) - (\omega - 2) \nu^{\omega - 2}}{3 - \omega}
       = \Order(\nu).
   \end{equation}
   Hence, we obtain the bound of $\Order(n r \nu)$, and hence of $\Order(n r^2)$.
\end{proof}

The proof of \cref{thm:our_bound} implies the following refined bounds.

\ourrefinedtimebound*

The former bound in \cref{thm:our_bound} generalizes Cunningham's time bound for the linear matroid intersection algorithm~\cite{Cunningham1986ImprovedBounds}. Although Cunningham's analysis did not account for fast matrix multiplication, a counterpart to the latter bound can also be derived for this intersection algorithm, which shaves off the log factor. Note that any exponent $\omega \in [2, 3)$, including the original value $\omega = \log_2 7$ by Strassen~\cite{Strassen1969GaussianElimination}, yields the same time bounds up to a constant factor.

\section{Conclusion}
\label{chap:conclude}

In this paper, we have provided a shortest augmenting path algorithm for linear matroid parity. This work improves upon the fastest known deterministic algorithm of Gabow and Stallmann~\cite{GabowS1986AugmentingPath} by extending the efficacy of the shortest augmenting path framework for linear matroid intersection and graph matching.

In \cref{chap:sap_length}, we have derived a bound on the lengths of shortest augmenting paths for linear matroid parity, which generalizes Cunningham's bound on the lengths of shortest augmenting paths for linear matroid intersection. As a corollary, we have verified in \cref{sec:local_search} that the local search algorithm, studied by Lee, Sviridenko, and Vondr{\'a}k~\cite{LeeSV2013MatroidMatching} for general matroid parity, requires a much smaller neighborhood to guarantee a PTAS when applied to linear matroid parity.

In \cref{chap:sap_algo_lmp,sec:overall_time_bound}, we have developed a shortest augmenting path algorithm for linear matroid parity, established its correctness, and analyzed its time complexity. Consequently, we have demonstrated that the problem can be deterministically solved in $\Order(n r^2 \log r)$ or $\Order(n r^2)$ time, depending on the use of fast matrix multiplication. These time bounds improve on the previous best deterministic bounds of $\Order(n r^3)$~\cite{GabowS1986AugmentingPath,Orlin2008FastSimpler} and $\Order(n r^\omega)$~\cite{GabowS1986AugmentingPath}, and narrow the gap with the best randomized bounds of $\Order(n r^2)$ and $\Order(n r^{\omega - 1})$~\cite{CheungLL2014AlgebraicAlgorithms}.

A natural research direction is to examine whether the shortest augmenting path framework generalizes to the linear delta-matroid parity algorithm~\cite{GeelenIM2003LinearDeltamatroid}. Since our verification process builds upon their proof techniques, such an extension appears highly plausible. However, removing singletons and generalizing the analysis of the shortest augmenting path lengths may present unforeseen challenges.

Given that the time complexity of our algorithm aligns with that of Cunningham~\cite{Cunningham1986ImprovedBounds} for linear matroid intersection, it would be intriguing to explore computational improvements or weighted extensions by leveraging existing techniques for linear matroid intersection. Specifically, Gabow and Xu~\cite{GabowX1996EfficientTheoretic} developed more efficient linear matroid intersection algorithms for unweighted and integer-weighted cases by extending Cunningham's observations and utilizing the weight-splitting theorem by Frank~\cite{Frank1981WeightedMatroid}. Generalizing this approach to linear matroid parity using our framework and the primal-dual relation for the weighted problem~\cite{IwataK2022WeightedLinear} is a compelling future direction.

\section*{Acknowledgements}

The authors thank Takuro Fukunaga, Jim Geelen, and Vijay Vazirani 
for their suggestions and encouragements.
This work was supported by JSPS KAKENHI Grant Number JP26KJ0928 and JST ASPIRE Grant Number JPMJAP2302. 

\bibliography{reference.bib}
\bibliographystyle{abbrvnat}

\end{document}